\documentclass[acmtosn]{acmtrans2m}

\usepackage{amsmath,epsfig,amsfonts,amssymb,mathrsfs,yhmath,times}
\usepackage{graphicx}
\usepackage[usenames,dvipsnames]{color}

\newtheorem{corollary}{Corollary}[section]

 \usepackage[vlined,ruled]{algorithm2e}

 \newcommand{{\HC}}{{\textsc{Push \& Pull}}}

\newtheorem{theorem}{Theorem}[section]
 
 \newtheorem{lemma}[theorem]{Lemma}

\def\whitebox{{\hbox{\hskip 1pt
    \vrule height 6pt depth 1.5pt
    \lower 1.5pt\vbox to 7.5pt{\hrule width
        3.2pt\vfill\hrule width 3.2pt}%
    \vrule height 6pt depth 1.5pt
    \hskip 1pt } }}

\def\qed{\ifhmode\allowbreak\else\nobreak\fi\hfill\quad\nobreak
   \whitebox\medbreak}
 
\newcommand{\ls}[1]
    {\dimen0=\fontdimen6\the\font
     \lineskip=#1\dimen0
     \advance\lineskip.5\fontdimen5\the\font
     \advance\lineskip-\dimen0
     \lineskiplimit=.9\lineskip
     \baselineskip=\lineskip
     \advance\baselineskip\dimen0
     \normallineskip\lineskip
     \normallineskiplimit\lineskiplimit
     \normalbaselineskip\baselineskip
     \ignorespaces}
     
\newcommand{\stackeq}[1]{
\stackrel{#1}{=}
}

\let\oldmarginpar\marginpar\renewcommand\marginpar[1]
{{\ls{1}\oldmarginpar[\raggedleft\footnotesize\textcolor{black}{}]{\raggedright\footnotesize\textcolor{black}{}}}}

\def\alg {\tt SARA}
\def\dlm {\tt DLM}
\def\gup {\tt VRCSC}

\newcommand{\BibTeX}{{\rm B\kern-.05em{\sc i\kern-.025em b}\kern-.08em
    T\kern-.1667em\lower.7ex\hbox{E}\kern-.125emX}}

\title{TECHNICAL REPORT\\ Sensor Activation and Radius Adaptation (SARA) in Heterogeneous Sensor Networks}
\author{NOVELLA BARTOLINI, TIZIANA CALAMONERI\\``Sapienza'' University of Rome, Italy\\
TOM LA PORTA\\Pennsylvania State University, USA \and
CHIARA PETRIOLI, SIMONE SILVESTRI\\``Sapienza'' University of Rome, Italy}
\begin{abstract}
In this paper we address the problem of prolonging the lifetime of wireless sensor networks (WSNs) deployed to monitor an area of interest.
In this scenario, a helpful approach is to reduce coverage redundancy and therefore the energy expenditure due to coverage.

We introduce the first algorithm which reduces coverage redundancy by means of Sensor Activation and sensing Radius Adaptation ($\alg$)
 in a general applicative scenario with
two classes of devices: sensors that can adapt their sensing range (adjustable sensors) and sensors that cannot (fixed sensors).
In particular, $\alg$ activates only a subset of all the available sensors and reduces the sensing range of the adjustable sensors that have been activated.
In doing so,  $\alg$ also  takes possible heterogeneous coverage capabilities of sensors belonging to the same class into account.
It specifically addresses device  heterogeneity by modeling the coverage problem in the Laguerre geometry through Voronoi-Laguerre diagrams.

$\alg$ executes quickly and is guaranteed  to terminate.  It provides  a configuration of the active set of sensors
that  
{\color{black} meets} 
lifetime and coverage requirements of demanding WSN applications, not met by current solutions.

 By means of extensive simulations we show that $\alg$ 
achieves a network lifetime that is significantly superior to that obtained by previous algorithms in all the considered scenarios.

\end{abstract}
\begin{document}
  
[Distributed networks; network topology]


\setcounter{page}{111}
\begin{bottomstuff}
Author's address: Novella Bartolini, Tiziana Calamoneri, Chiara Petrioli and Simone Silvestri, Department of Computer Science, ``Sapienza'' University of Rome, Via Salaria 113, 198 Rome, Italy . E-mail: \{novella, calamo, petrioli, silvestris\}@di.uniroma1.it\newline
Tom La Porta, Pennsylvania State University, 
University Park, PA, USA. E-mail:tlp@cse.psu.edu\\
This work has been partially supported by the Italian Ministry of Education and University
 PRIN project COGENT (COmputational and GamE-theoretic aspects of 
uncoordinated NeTworks) and by the EC FP7 project GENESI (Green sEnsor NEtworks
for Structural monItoring.)
\end{bottomstuff}
\maketitle


\section{Introduction}\label{sec:introduction}
As large collections of networked, inexpensive devices, Wireless Sensor Networks (WSNs) are the technology of choice 
for applications requiring seamless and pervasive coverage of geographic areas, buildings and public or private spaces 
and structures.
Critical applications such as access control and intrusion/hazard detection as well as less critical tasks of 
which wildlife monitoring and precision agriculture are typical examples, are best served by the {\color{black} 
infrastructure-less}
\marginpar{ {\color{black} was: infrastructureless, il vocab non la passa}}
 and unobtrusive nature of WSNs.

Since sensor nodes typically have limited battery power, meeting coverage requirements with minimal energy 
expenditure is a primary issue.
For years this problem has been tackled by designing protocol stacks that are energy efficient, implicitly 
assuming that the culprit of most of the energy consumption of a node is the communication circuitry.
As a consequence, solutions that enhance network performance (lifetime, capacity, etc.) have been proposed 
that are  based on methods
 that reduce communication costs:
Data fusion and filtering techniques (for limiting the number of transmissions)~\cite{NakamuraLF07}, new and advanced forms of energy provisioning~\cite{SharmaMJG10,MoserTBB10,KansalHZS07}, clever exploitation of the mobility of network components~\cite{BasagniCP08} as well as optimized protocol design~\cite{YickMG08}.
However, the level of improvement that energy-efficient techniques for communication can produce is starting to plateau because of the 
inevitable trade-offs that they impose (e.g., energy conservation {\color{black}versus}
 \marginpar{  {\color{black} was: vs.}}
latency).
At the same time, the sensing devices mounted on the wireless node have become more numerous and more sophisticated.
Along with the cheap sensors, e.g., those for temperature and humidity, it is now common to endow even small nodes with cameras and active sensors such as radars and sonars, which demand non-negligible energy from the node.
Therefore, for providing critical enhancement to network performance, it is no longer possible to focus only on reducing communication costs. Careful consideration must be also given to the sensory component of the node.
We also note that, unlike``on-off'' sensors, like those for temperature, light, and humidity, more sophisticated sensors consume energy 
depending on their sensing range.
Therefore, similar to communication power control, \emph{sensing coverage control} becomes an important element in the overall 
WSN performance optimization process.
In particular, \emph{sensor activation and radius adaptation}, the ability of selecting which sensor to activate%
\footnote{\ With \emph{sensor activation} we indicate the turning on of the sensing and communication units of a node. When this happens, 
the sensor is \emph{awake}. A sensor goes to \emph{sleep} by turning off (or by switching to low power mode) both its sensing and 
communication units. 
}
\marginpar{\color{black} NNN: Tolta la def. di dead dalla footnote duplicata nella sez sperimentale}

and to what level of coverage, are necessary new ingredients for the design of durable and reliable WSNs.

In this paper we present a new solution for the joint problem of dynamically scheduling the activation of different subsets of  sensor nodes and of tuning their sensing radii (if their technology allows)
for prolonging the network lifetime while ensuring the coverage of the given Area of Interest (AoI)\footnote{As usually done in the literature, we assume that the AoI is a convex region.}.
Sensor activation as a research area has received considerable attention in the recent past. In particular, two selective activation algorithms have been proposed that have been shown to outperform other solutions for the problem: The Distributed Lifetime Maximization ($\dlm$) scheme~\cite{Kaskebar2009}, and the Variable Radii Connected Sensor Cover ($\gup$)~\cite{Zou2009}.

In this paper we propose an algorithm called $\alg$, standing for \emph{Sensor Activation and Radius Adaptation}
that, to the best of the authors' knowledge, is the first algorithm working in the general scenario of heterogeneous networks.
Our algorithm follows an original approach to solve the coverage problem, as it makes use 
of the Voronoi diagrams in the Laguerre geometry to determine the coverage {\color{black}responsibility} \marginpar{ {\color{black} was: extent }} of each node.

$\alg$ achieves the following desirable properties (theoretically and experimentally proven in the following).

\begin{itemize}

\item It ensures maximum sensing coverage at all times, i.e.,  activated nodes are able to cover the same area that would be covered if 
all nodes that are still operational 
{\color{black}were} \marginpar{{\color{black} was: and}}
activated with their maximum transmission range.


\item It accommodates WSNs comprised of  heterogeneous nodes, endowed with active and passive 
sensors with fixed or adjustable sensing radius.

\item It is Pareto optimal, {\color{black} unlike} 
\marginpar{\color{black}was: unlike from} 
$\dlm$ and $\gup$. (This property constitutes 
a necessary requirement for a sensor activation and radius adaptation policy to be optimal.)

\item {\color{black} It} \marginpar{\color{black}was: $\alg$}  is robust with respect to different definitions of coverage requirements and network lifetime.

\end{itemize}

{\color{black} 
 \marginpar{{\color{black}NNN: I tried to fill the TO BE COMPLETED parts, this conclusion needs some more details from the experiments.}}
The performance of $\alg$ has been evaluated by means of simulation experiments on WSNs with heterogeneous nodes. 
The results of our experiments show that $\alg$ is able to quickly configure the network in a way that ensures low energy consumption and long lifetime.
We also conducted a comparative performance evaluation of $\alg$ with $\dlm$ and $\gup$, which revealed the superiority 
of $\alg$ in terms of coverage extension and network lifetime in a wide range of operative settings, including the ones for which those 
previous solutions were specifically designed.
}

{\color{black} 
 \marginpar{{\color{black}NNN: Added the sketch of the paper.}}
The paper is organized as follows.
Section \ref{sec:problem_formulation}
introduces the problem of radius adaptation and sensor activation.
Section \ref{sec:motivation_and_preliminaries} 
 motivates the use of the Voronoi-Laguerre measure
 to address device heterogeneity and provides the notions of computational geometry
needed to fully understand the proposed solution.
 In Sections \ref{sec:algo}  and \ref{sec:properties}
  we describe  $\alg$ and prove its Pareto optimality, convergence and termination.
Section \ref{sec:two_approaches} briefly describes the protocols selected as
benchmarks:  $\dlm$ and $\gup$. A thorough performance evaluation of  $\alg$
is then provided in Section \ref{sec:results}, including a comparison between
$\alg$, $\dlm$ and $\gup$ performance. 
Finally, Section \ref{sec:related_works} surveys the literature on related topics, while
Section~\ref{sec:conclusions} concludes the paper.
}

\section{Problem Formulation}
\label{sec:problem_formulation}
\label{sec.approach}

In this paper we consider heterogeneous WSNs, where the nodes are endowed with several kinds of sensing technologies.
In particular, we focus on the use of two types of sensors, namely, those with adjustable sensing radius and those with fixed radius.
The capability to adjust the sensing range is typical of devices based on active sensing technologies, such as those equipped with radars and sonars.
The power consumption of this kind of sensor depends on the extent of the sensing radius.
For this type of sensors setting the sensing range to the minimum necessary for coverage decreases energy consumption.
Although not all commercial active devices allow radius to be adapted, some sensors with variable sensing ranges are already commercially available~ \cite{Osi2008,Kompis2008}.
By contrast, for sensors {\color{black} based on passive sensing technologies} (e.g., those equipped with piezoelectric sensors or thermometers) the monitoring activity typically consists in taking single point measures. For these devices the sensing range is typically fixed and so is their energy consumption.
An exception is the case of low power CMOS cameras, based on a passive sensing approach, where the {\color{black}depth of field} \marginpar{\color{black} was: resolution} can be adjusted
to guarantee a given quality of monitoring at certain distances.

We consider a  set $\mathcal{S}=\mathcal{S}_\texttt{adjustable} \cup \mathcal{S}_\texttt{fixed}$ of $|S|=N$ sensors, where $\mathcal{S}_\texttt{adjustable}$ contains the nodes with adjustable sensing radius 
{\color{black} (hereby shortly called {\em adjustable sensors})} 
and $\mathcal{S}_\texttt{fixed}$ those with
a fixed radius {\color{black} (shortly called {\em fixed sensors})}  .
If a node $s_i$ belongs to the set $\mathcal{S}_\texttt{adjustable}$ its sensing radius $r_i$
can be set to any value from $0$  to $r_i^\texttt{max}$.
For a node $s_j \in \mathcal{S}_\texttt{fixed}$ the sensing radius $r_j$ is either 0{\color{black}, meaning that the sensing unit is turned off,} or $r_j^\texttt{fixed}${\color{black}, when the sensing unit is turned on}.
The sensors of the two sets can also have heterogeneous transmission radii $r^\texttt{tx}_i$, $i=1, \ldots, N$.
We assume that the transmission radii are such that
any two sensors with intersecting or tangential sensing circles are connected to each other.
Therefore, complete coverage implies also that the WSN is connected, and no sensor should be kept awake if it is not necessary for coverage.

\marginpar{Modifico in base all'unica ref che abbiamo, lasciamo la questione del 3 (vedi commento di Stefano) alla parte sperimentale, qui non \'e necessaria.}

{\color{black} An exact model of the relationship between the energy consumed by a node for sensing and the extent of its sensing radius cannot be
given as it is dependent on the sensing technology and electronic circuitry for detection. For the purpose of our work, 
we 
refer
to a general approximate model also used in \cite{Pattem2003,Zou2009} according to which if sensor $s_{i}$ 
has sensing radius $r_{i}$ the energy consumption per time unit is given by }
\begin{equation}
 E_\texttt{sensing} (r_i)=  a \cdot r^c_i + b.
\label{eq:e_active_sensing}
\end{equation}

\marginpar{
{\color{black}
was: 
The constant $a$ is a device specific constant that can be $0$, as in the case of some fixed radius sensing devices.
The parameter $b$ is a constant that takes into account the energy consumption of the sensing circuitry.
The value of  $c$ is generally greater than or equal to $3$ in the case of adjustable radius sensors.
}}

{\color{black}
The parameters $a$ and $b$ are device specific constants. The parameter $c$ 
is related to the sensing technology in use 
and typically varies 
in the range $[2,4]$ in case of sensors adopting an active sensing technology.
}

The energy consumption due to communications is  also dependent on the specific type of device being considered.
It is typically
an increasing function of the transmission radius, which
takes into account all the energy consuming activities related to radio communications, namely
transmissions, receptions and idle listening to the radio channel.
In this paper we consider the energy cost model of Telos nodes \cite{Telos}.

{\color{black} The problem addressed in this paper is the following: }
{\em
Given a WSNs each sensor $s_i \in S$ has to decide whether to activate itself or not at any given time and, if active,
 how to set its sensing radius $r_i$ at that time.
The objective is guaranteeing maximum sensing coverage while prolonging the network lifetime as much as possible.
}

Here we define the network lifetime as the time during which the network is able to guarantee the coverage of a given percentage $p$ of the AoI.
For instance, if $p=100\%$ the network lifetime is the time at which the first coverage hole appears.
If $p=x\%$ the network lifetime is the first time at which less than $x\%$ of the AoI is covered
  \footnote{
  {\color{black} Definitions of lifetime based on the percentage of alive nodes \cite{Blough2002}  
can be adopted as well. Although more commonly used in the literature, these
different notions of lifetime are less suitable than our when the applicative task is coverage of an AoI.}}.
{\color{black}
\marginpar{messo questo pezzetto perch\'e in passato Tom aveva chiesto un rif sulla def di lifetime. 
Sfortunatamente la maggior parte dei lavori si preoccupa solo della perc di sensori vivi, cosa che non 
va bene per lo scenario ma che non altera l'algo} }


\section{Preliminaries on Voronoi Laguerre diagrams and on their use to determine and reduce coverage redundancy}

\label{sec:motivation_and_preliminaries}

Prior works on sensor networks very often rely  on the use of Voronoi diagrams
  to model coverage, such as in~\cite{LaPorta06} for mobile sensors, in~\cite{Ammari2008} for energy aware routing,
or in~\cite{Zou2009} for selective activation.
Voronoi diagrams can  be used to model the coverage problem only in the case of sensors endowed with
equal sensing radii as discussed in ~\cite{Noi_ICNP2009}.
In order to address the problem of coverage in the presence of heterogeneous devices, namely devices with different sensing ranges and different capability to adapt their setting, 
in this section we introduce the notion of Voronoi diagrams in {\em Laguerre geometry}.
We also discuss how these diagrams can be exploited to decrease coverage redundancy (and thus the energy consumption due to sensing) while preserving network coverage and connectivity. 

In a Voronoi diagram, we call the axis generated by two sensors which is equidistant from them and perpendicular to their connecting segment the {\em Vor line} .
This line divides the plane into two halves.
In the case of sensors with the same sensing radius the Vor line properly delimits the responsibility regions of the two sensors as it is the
symmetry axis between the two.
If the sensors have heterogeneous radii, the Vor line may not determine the responsibility region correctly,
as depicted in Fig.~\ref{fig:vor_vorlag}. \marginpar{tolto il rif alla boldness della linea Vor, sono bold tutte e due.} 
Indeed, according to a Voronoi-based partition of coverage responsibilities, {\color{black} the sensor positioned in} $\textrm{C}_1$ has the responsibility to sense anything to
the left of the Vor line, and  {\color{black} the sensor positioned in} $\textrm{C}_2$ should sense anything to the right.
In particular, the grey areas in the figure
would incorrectly be assigned to {\color{black} the sensor in}  $\textrm{C}_1$, whereas they are covered only by {\color{black} the sensor in} $\textrm{C}_2$.
The line which correctly delimits the responsibility regions of the two sensors is the one that is equidistant from $\textrm{C}_1$ and $\textrm{C}_2$ in Laguerre geometry. In Figure~\ref{fig:vor_vorlag} this line is  called {\em VorLag}.

\begin{figure}[h]
{\color{black}
\centering
\begin{tabular}{cc}
{\scalebox{0.30}{\includegraphics[]{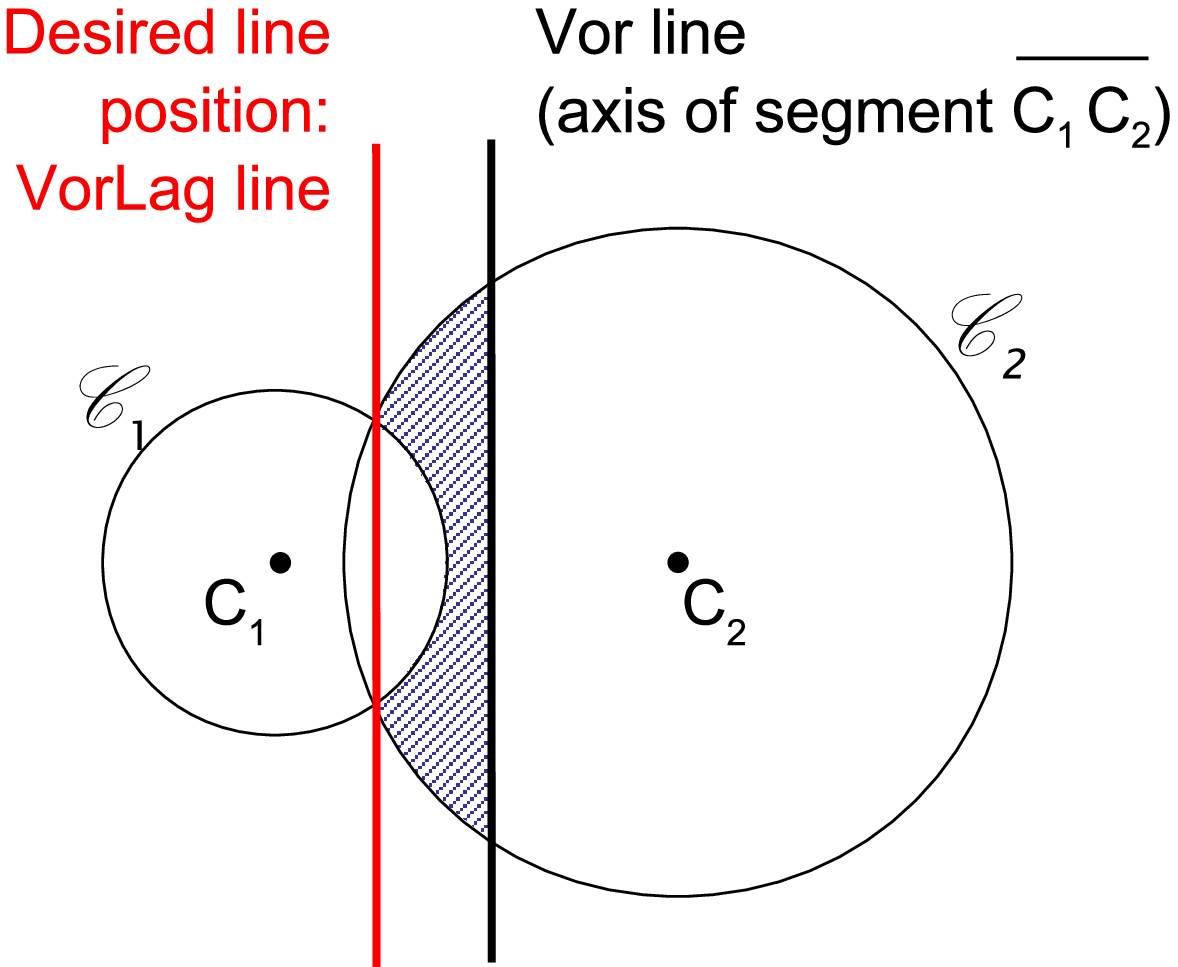}}}
&
{\scalebox{0.30}{\includegraphics[]{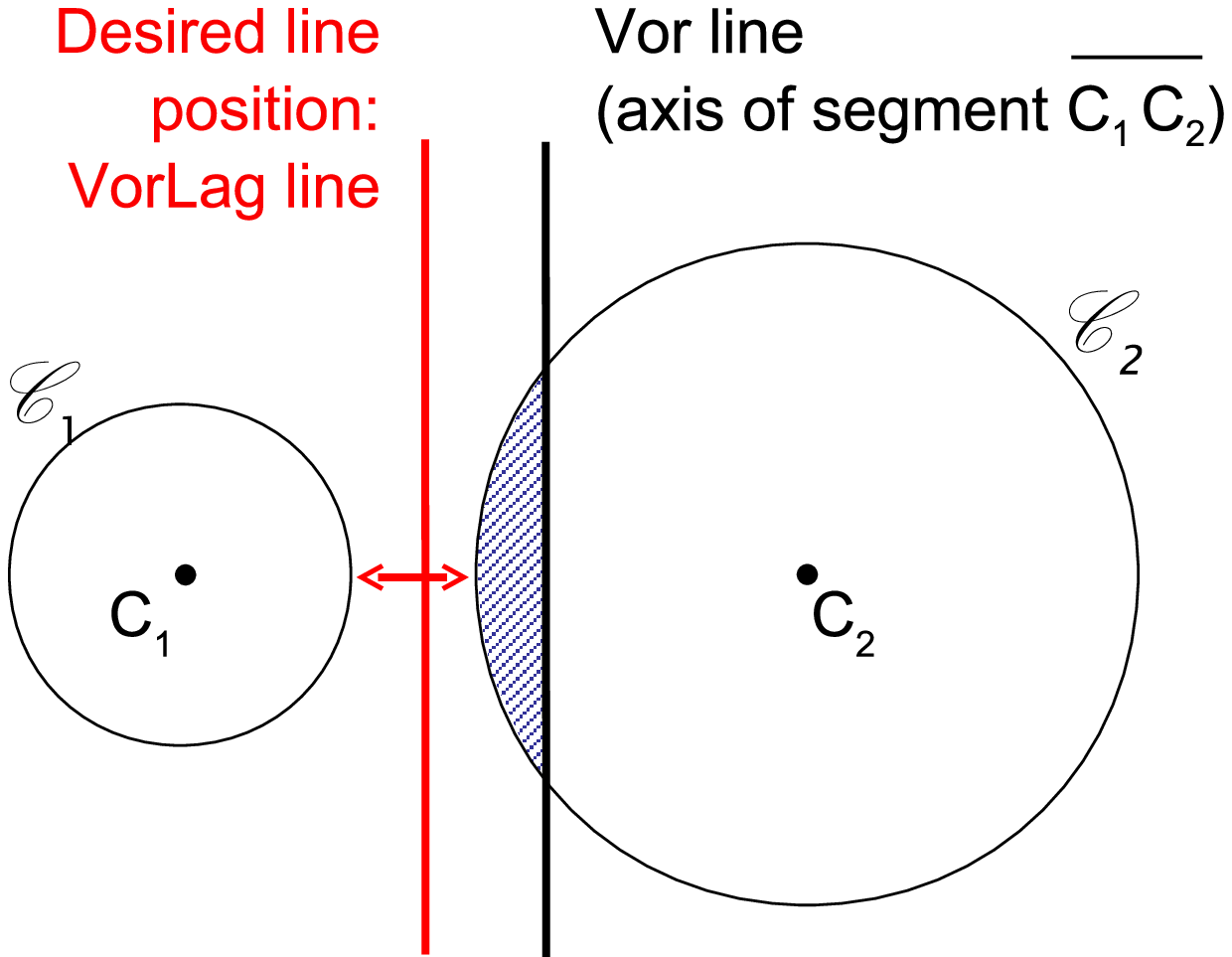}}}
\\
{\footnotesize{(a)}}&{\footnotesize{(b)}}
\end{tabular}
 \caption{\color{black}Different positions of the line equidistant from $\textrm{C}_1$ and $\textrm{C}_2$ according to the Euclidean (Vor)
and to the Laguerre (VorLag) distance in the case of intersecting~(a) and non-intersecting circles~(b).}
\label{fig:vor_vorlag}
}
\end{figure}

{\color{black}
Formally, given a circle $\mathscr{C}$ with center $\textrm{C}=(x_\textrm{C}, y_\textrm{C})$ and radius $r_\textrm{C}$,
and a point $\textrm{P}=(x_\textrm{P}, y_\textrm{P})$ in the plane $\Re^2$, the Laguerre distance $d_L(\mathscr{C}, \textrm{P})$ between $\mathscr{C}$
and $\textrm{P}$ is  defined  as follows:
}
{\color{black}
\begin{equation}
d^2_L(\mathscr{C},\textrm{P})=d^2_E(\textrm{C},\textrm{P})- r_\textrm{C}^2,
\label{eq:lag_dist}
\end{equation}
where  $d_E(\textrm{C},\textrm{P})$ is  the Euclidean distance between the points
$\textrm{C}$ and $\textrm{P}$.
In Laguerre geometry, given two circles
with distinct centers and possibly different radii, the locus of the points
equally distant from them
is a line, called  {\em VorLag line}, that is perpendicular to the segment
connecting the centers.
\marginpar{
{\color{black} NNN: mi sembra che il fatto che l'asse passi per le intersezioni non sia stato detto e dia invece l'intuizione su cosa implichi l'uso di Laguerre.}
}
{\color{black} If the two circles intersect each other, their VorLag line crosses their intersection points, as in Fig.~\ref{fig:vor_vorlag} (a)}
~\cite{IIM85}.

{\color{black}
Given $N$ circles $\mathscr{C}_i$ 
with centers $\textrm{C}_i=(x_i, y_i)$ and radii $r_i$, $i=1, \ldots, N$,
the {\em Voronoi-Laguerre polygons} $V(\mathscr{C}_i)$ for the circles $\mathscr{C}_i$ are defined as
$$V(\mathscr{C}_i)=  \{ \textrm{P} \in \Re^2 | d^2_L(\mathscr{C}_i, \textrm{P}) \leq d^2_L(\mathscr{C}_j, \textrm{P}) \textrm{,  }\forall j \neq i \}, i=1, \ldots, N.$$

}

\marginpar{
{\color{black}
(\$): da qui  spostate due frasi, vedi sotto in verde}
}
A Voronoi-Laguerre polygon 
is always convex.
{\color{black}A tessellation of the plane into Voronoi-Laguerre polygons  is called a  {\em Voronoi-Laguerre diagram}.}
\marginpar{{\color{black}Ridotta la frase}}
{\color{black}
Obviously, if $r_i=r_j$ for all $i,j =1, \ldots, N$, the Voronoi-Laguerre diagram reduces to the ordinary Voronoi diagram.
Notice that it may happen that the Voronoi-Laguerre polygon $V(\mathscr{C}_i)$ does not contain any point of the plane.
This happens when the half-planes generated by the VorLag lines formed by  $\mathscr{C}_i$ and its nearby circles have no overlap.
In this case, $V(\mathscr{C}_i)$ is called a {\em null polygon}.
{\color{black}The occurrence of null polygons is specific of Voronoi-Laguerre diagrams and reflects a situation of complete redundancy that 
is not captured by traditional Voronoi diagrams for which the generated polygons are always not null.}

\marginpar{{\color{black}Riportato da (\$), con qualche cambiamento}}
{\color{black} In the following the sensor $s_{i}$ whose sensing circle $\mathscr{C}_i$ generates the polygon $V(\mathscr{C}_i)$ is called the
      {\em generator} of $V(\mathscr{C}_i)$; the vertices of the same polygon are hereby shortly referred to as {\em Voronoi-Laguerre vertices}.
}

Two sensors are {\em Voronoi-Laguerre neighbors} if their polygons have one edge in common. Given a sensor $s_i \in S$, the
set of its Voronoi-Laguerre neighbors is hereafter referred to as $\mathcal{N}_{S} (s_i)$.
Furthermore, we refer to $\mathcal{N}_{S}^{\emptyset} (s_i)$ as the set of
sensors with null polygons which have a sensing overlap
with the sensor $s_i$:
{\small 
$$ \mathcal{N}_{S}^\emptyset(s_i) = \{ s_j \in S: d_E(s_i,s_j)\leq (r_i+r_j) \wedge V(\mathscr{C}_j) = \emptyset  \}. $$
}
The reason why Voronoi Laguerre diagrams perfectly model the coverage problem in the case of heterogeneous sensors is  their capability to partition the area of interest into polygonal regions which in fact represent the responsibility regions of the deployed sensors.
Indeed, a fundamental property of the Voronoi diagrams in the Laguerre geometry is the following:

\begin{theorem} (
\cite{Noi_ICNP2009})
\label{th:lag_coverage}
Let  us consider $N$ circles $\mathscr{C}_i$, with centers $\textrm{C}_i=(x_i, y_i)$ and radii $r_i$, $i=1, \ldots , N$,
and let $V(\mathscr{C}_i)$ be the Voronoi-Laguerre polygon of the circle $\mathscr{C}_i$.
 For all $k,j=1,2, \ldots, N$,
  $V(\mathscr{C}_k) \cap \mathscr{C}_j \subseteq \mathscr{C}_k $.
\end{theorem}
Less formally, if a point $P$ of the area of interest is covered by at least one sensor, it is certainly covered 
also by the sensor $s_i$ that generates the Voronoi-Laguerre
 polygon $V(\mathscr{C}_i)$ that includes $P$.
 
 \marginpar{NNN:CONTROLLARE CHE NON SIA DEF DUE VOLTE}
 {\color{black}
 \subsection{Characterization of coverage redundancy}
 \label{sec:redundancy}
 We define as {\em redundant} any sensor $s_i \in \mathcal{S}$ such that the sensing circle $\mathscr{C}_i$ is completely
 covered by other sensors, namely $\mathscr{C}_i \subseteq \cup_{s_j \in \mathcal{S}, j \neq i} \mathscr{C}_j $.
 The following corollaries \ref{co:3.1}, \ref{co:3.2} and \ref{th:test} of Theorem \ref{th:lag_coverage} show the criteria 
 to decide whether $s_i$  is redundant.

 \begin{corollary}
 If a sensor $s_i$ does not cover any point of its Voronoi-Laguerre polygon $V(\mathscr{C}_i)$,
 then its sensing circle $\mathscr{C}_i$ is completely covered by other sensors in $\mathcal{S}$.
 Therefore $s_i$ is redundant.
 \label{co:3.1}
 \end{corollary}
 \begin{proof}
  Since by hypothesis $V(\mathscr{C}_i) \cap \mathscr{C}_i = \emptyset$, 
 $\mathscr{C}_i$ contains only points that are external to its polygon. 
 Therefore, if $P\in \mathscr{C}_i$ then $P$ is covered by the generating sensor of the polygon to which it belongs (for Theorem \ref{th:lag_coverage}).
  \end{proof}

  Corollary \ref{co:3.1} affirms that if $s_i$ does not cover its polygon, it can be turned off without affecting coverage.

 \begin{corollary}
 Given a sensor $s_i$ which covers only a portion of its polygon $V(\mathscr{C}_i)$,  let $\ell$ be a circular segment 
 on the intersection between the boundary of $\mathscr{C}_i$ with the polygon $V(\mathscr{C}_i)$.
 All the points on $\ell$ which are not on edges of  $V(\mathscr{C}_i)$ are covered only by $s_i$. 
 \label{co:3.2}
 \end{corollary}
 
 \begin{proof}
  By hypothesis, the region $V(\mathscr{C}_i) \setminus \mathscr{C}_i$  is not covered by the generating sensor of the polygon to which it belongs (that is the sensor $s_i$). Therefore, due to theorem \ref{th:lag_coverage}, it is not covered by any sensor.
  Consider any circular segment $\ell$ on the boundary of  $\mathscr{C}_i$ and inside $V(\mathscr{C}_i)$ 
 (see Fig. \ref{fig:dim} in which $\ell$ is the arc $\wideparen{DF}$) and a point $P$ on $\ell$ but not on the edges of $V(\mathscr{C}_i)$.
 We want to show that $s_i$ is the only sensor which covers $P$.
   Since $P$ is not on the edges of the polygon, it is possible to find a value of $\epsilon$ arbitrarily small,
  such that the $\epsilon$-surrounding of $P$  is internal to $V(\mathscr{C}_i)$.
 The intersection of this $\epsilon$-surrounding with the region  $V(\mathscr{C}_i) \setminus \mathscr{C}_i$ (that in Fig. \ref{fig:dim} is delimited by the segments $\overline{EF}$, $\overline{DE}$ and by the arc $\wideparen{DF}$) is obviously uncovered.
 
 We now proceed by contradiction.  Let us assume that there is another sensor $s_j \in \mathcal{S}$ such that $P$ is also covered by $s_j$.
  Since, by construction, any $\epsilon$-surrounding of $P$ contains an uncovered region,
 the circle $\mathscr{C}_i$ can cover $P$ only with its boundary. Furthermore,   
 since $s_j$ cannot cover points of $V(\mathscr{C}_i) \setminus \mathscr{C}_i$, then $s_j$ must be tangential to 
 $\mathscr{C}_i$ in $P$, and must have a lower  sensing radius  $r_j < r_i$. 
 \marginpar{Here we are implicitly excluding superimposed nodes, for which the Voronoi-Laguerre polygon does not exist}
  However this implies that  $P$ would be crossed by the Voronoi-Laguerre edge formed by $s_i$ and $s_j$, and the portion of $V(\mathscr{C}_i)$ on the opposite side of this edge with respect to $\mathscr{C}_i$ could not belong to $V(\mathscr{C}_i)$, contradicting our construction.
  \end{proof}
  
\begin{figure}[h]
\vspace{1.2cm}
{\color{black}
\centering
{\scalebox{0.4}{\includegraphics[]{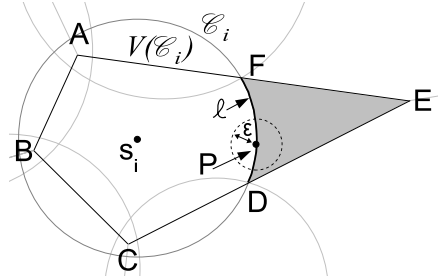}}}
 \caption{\color{black}Voronoi-Laguerre polygon partially covered by its generating sensor.}
\label{fig:dim}
}
\end{figure}

Corollary \ref{co:3.2} states that if $s_i$ only partially covers its polygon, it cannot reduce its sensing radius without affecting coverage.

{\color{black}
\begin{corollary}
\label{th:test}
Let us consider a sensor $s_i$, with sensing circle $\mathscr{C}_{i}$ and Voronoi-Laguerre polygon 
$V(\mathscr{C}_i)$.
Let  $P$ be a point that is covered by $s_i$ and is internal to its polygon, that is $P \in V(\mathscr{C}_{i}) \cap \mathscr{C}_i$.
If $P$ is  covered  by a sensor in $\mathcal{S}$ other than $s_i$, then 
there exists a sensor $s_k \in \mathcal{N}_{\mathcal{S}} (s_i)\cup \mathcal{N}_{\mathcal{S}}^{\emptyset} (s_i)$ 
%
such that $P$ is also covered by $s_k$.
In other words, 
any point of $V(\mathscr{C}_i)$ that is covered by more than one sensor, is certainly covered at least by 
the generating sensor $s_i$ and by one of its Voronoi-Laguerre 
neighbors or a sensor with null polygon.
\end{corollary}
}
\begin{proof}
Let $\mathscr{D}$ be the Voronoi-Laguerre diagram  generated by  $\mathcal{S}$ and
 $\mathscr{D}'$ be the diagram generated by $\mathcal{S'}=\mathcal{S} \setminus \{s_i\}$.
In the diagram $\mathscr{D}$, $P \in V(\mathscr{C}_i)$.
By contrast, in the diagram $\mathscr{D}'$, the sensor $s_i$ is not present. 

Since by the hypothesis, $P$ is covered by a sensor in $\mathcal{S}'$, thanks to Theorem \ref{th:lag_coverage} we
can affirm that $P$ is also covered by the generating sensor $s_k$ of the polygon, such that $P \in  V'(\mathscr{C}_k)$ defined in  $\mathscr{D}'$.
 Obviously, $V'(\mathscr{C}_k) \neq  V(\mathscr{C}_k)$.
 Let us assume, for sake of contradiction, that 
 $s_k \notin {N}_{\mathcal{S}} (s_i)\cup \mathcal{N}_{\mathcal{S}}^{\emptyset} (s_i)$
If the sensor $s_k$ is not a Voronoi-Laguerre neighbor of $s_i$ and it has not a null polygon in  $\mathscr{D}$, 
its polygon in $\mathscr{D}'$ would be the same as in $\mathscr{D}$, because it would be delimited by edges formed by sensors other than $s_i$.
Therefore it would be $V'(\mathscr{C}_k)=V(\mathscr{C}_k)$, which is a contradiction.

\end{proof}

Corollary \ref{th:test} states that in order to decide whether $s_i$ can reduce its radius or be turned off 
it is sufficient to evaluate the coverage of the sensors in $\mathcal{N}_{\mathcal{S}} (s_i)\cup \mathcal{N}_{\mathcal{S}}^{\emptyset} (s_i)$.

 }


\subsection{Reducing the redundancy of sensors with adjustable sensing radius}

%
{\color{black}
The corollaries \ref{co:3.1}, \ref{co:3.2} and \ref{th:test} let us determine whether an adjustable sensor $s_i$ 
can reduce its sensing radius or turn itself off. In particular: (1) if the sensor $s_i$ does not cover any point of its polygon, $s_i$ 
can be turned off (in consequence of Corollary \ref{co:3.1}); (2) if $s_i$ covers its polygon only partially, $s_i$ must stay awake and work with its current radius (in consequence of Corollary \ref{co:3.2});
(3) if $s_i$ covers its polygon completely, it may reduce its sensing radius of an extent that can be determined on the basis of the coverage of its neighbors  (in consequence of Corollary \ref{th:test}).

We now address the third situation more in detail.
If $s_i$ covers its polygon completely, it}
can shrink its sensing radius to the distance between 
$s_i$ and the farthest vertex $f(V(\mathscr{C}_i))$ of its polygon,
 without compromising maximum sensing coverage.

\begin{figure}[h]
{\color{black}
\centering
\begin{tabular}{ccc}
&&\\
&&\\
\hspace{-2cm}
{\scalebox{0.35}{\includegraphics[]{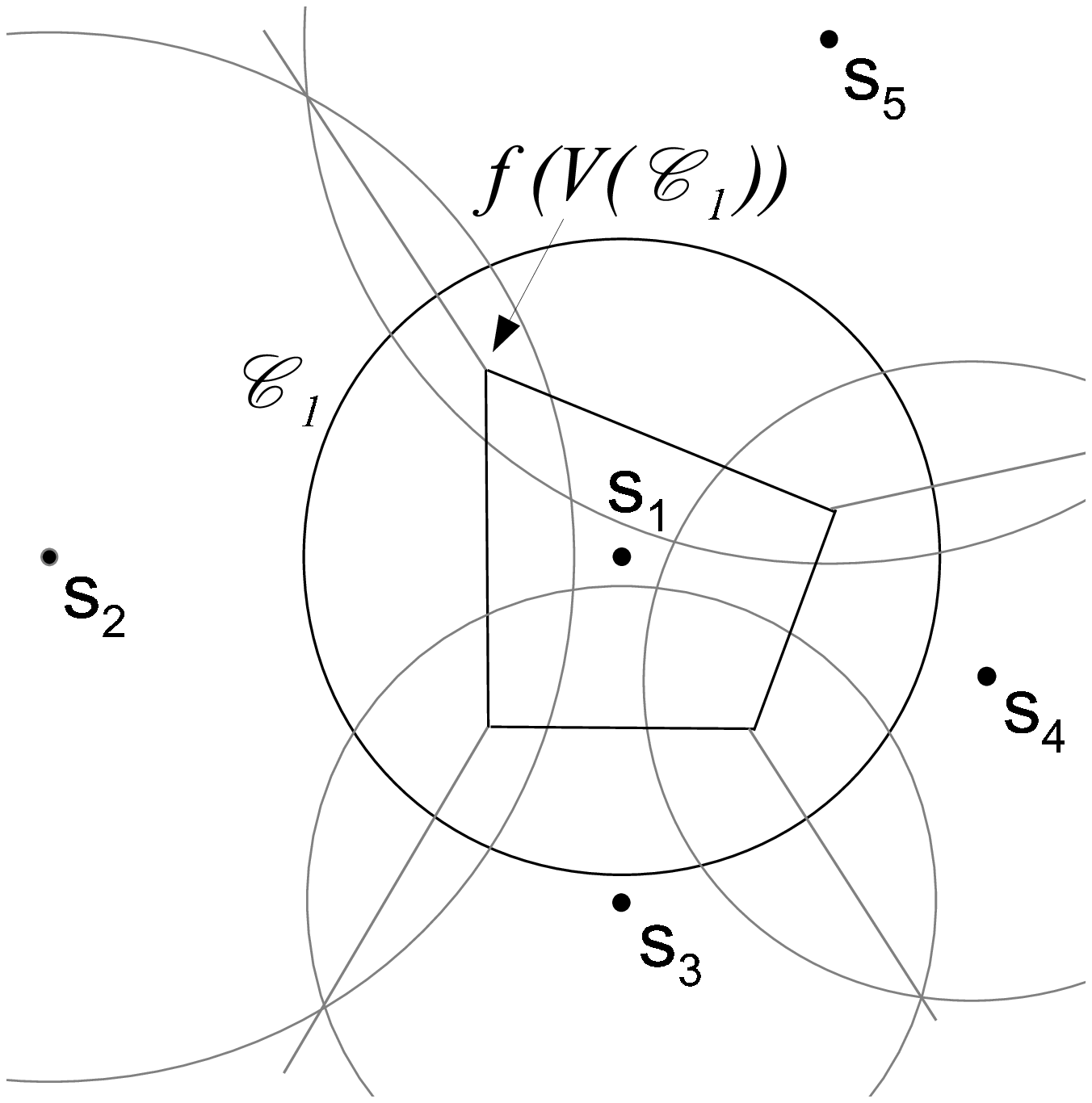}}}
&\hspace{3cm}
&
{\scalebox{0.35}{\includegraphics[]{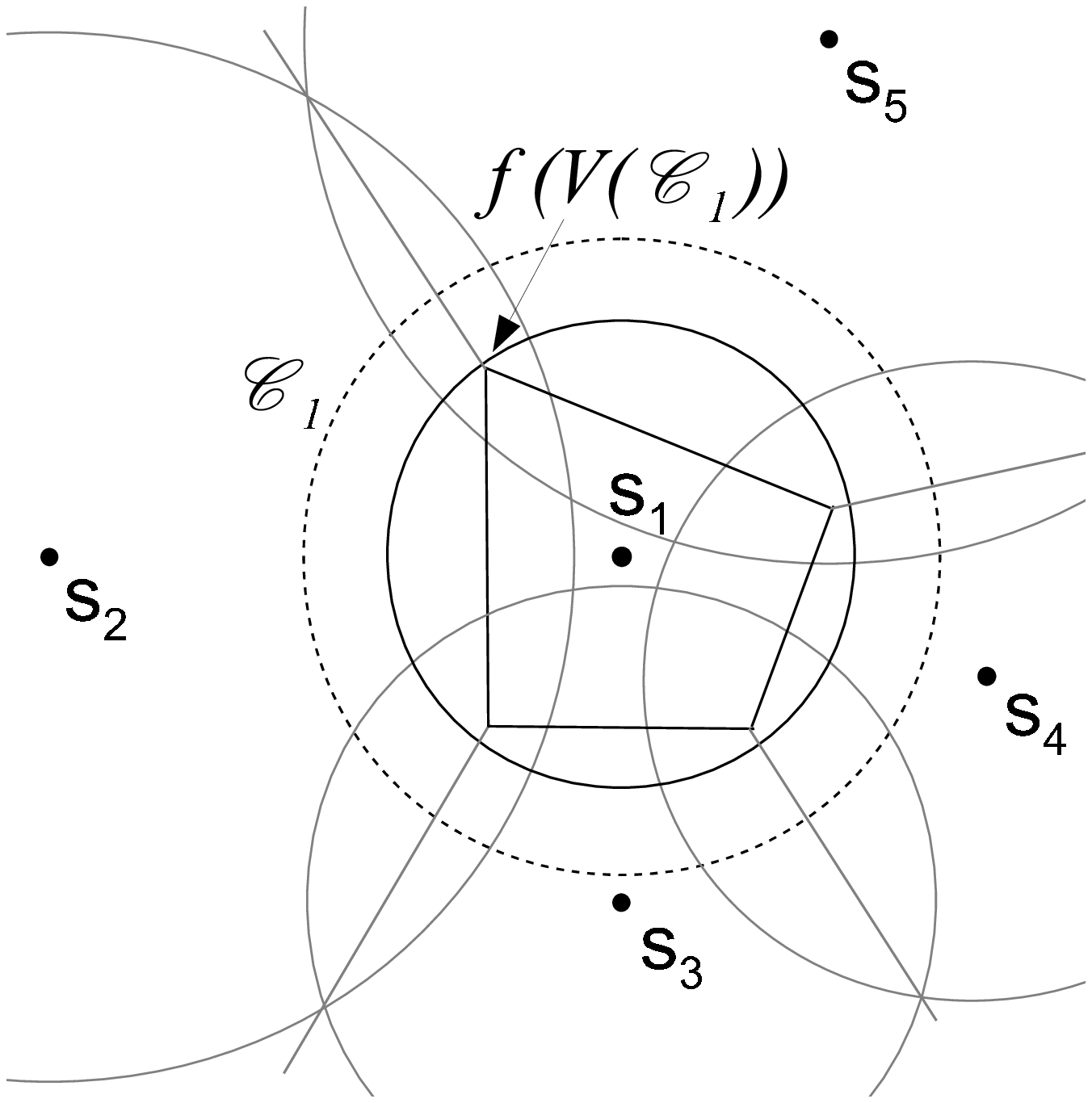}}}
\\
{\footnotesize{(a)}}&\hspace{3cm}&{\footnotesize{(b)}}\\
&&\\
&&\\
&&\\
&&\\
\hspace{-2cm}{\scalebox{0.35}{\includegraphics[]{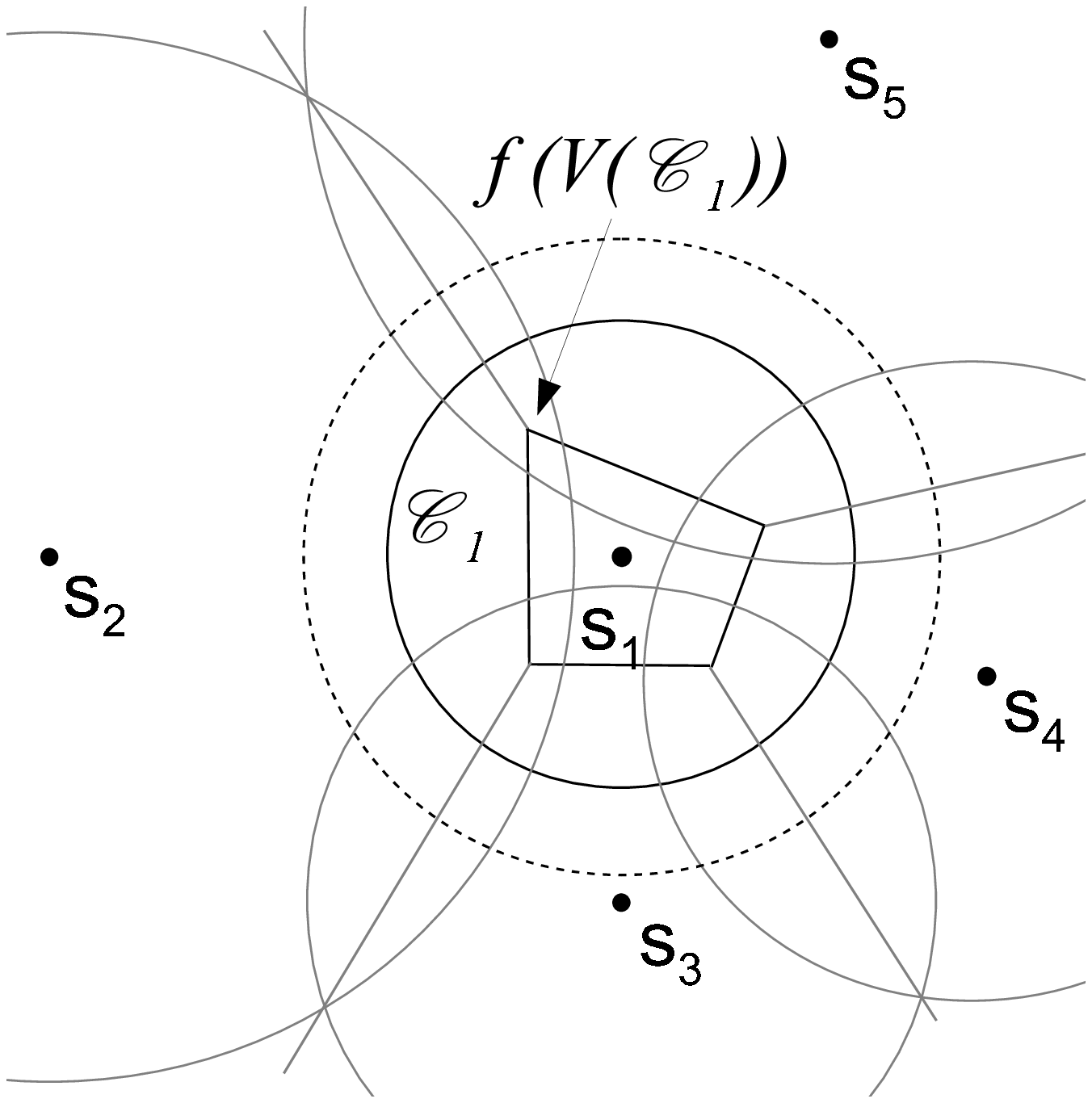}}}
&\hspace{3cm}
&
{\scalebox{0.35}{\includegraphics[]{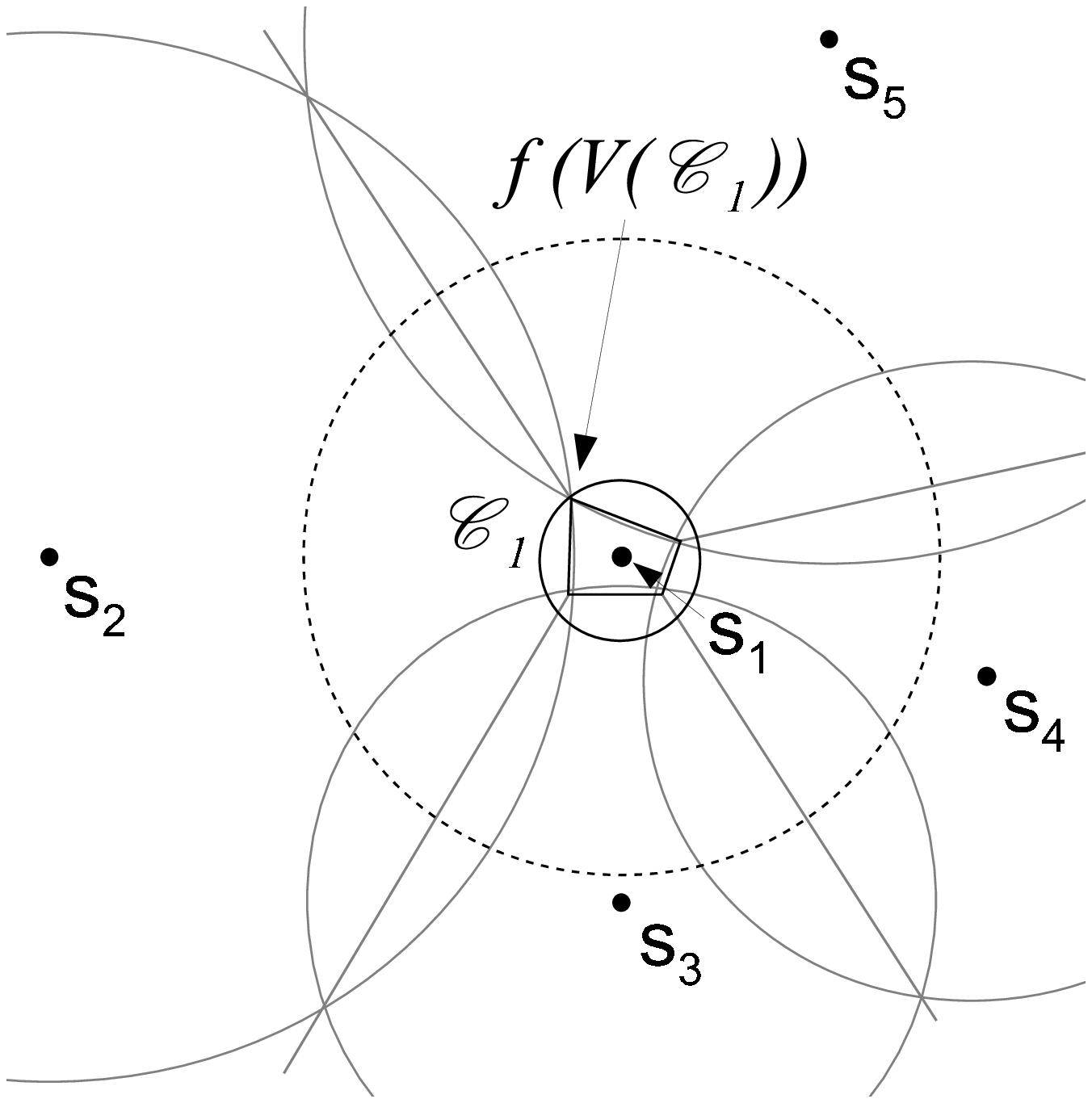}}}
\\
{\footnotesize{(c)}}
&\hspace{3cm}
&
\hspace{0.35cm}{\footnotesize{(d)}}\\
\end{tabular}
 \caption{\color{black}Iterative reduction of the sensing radius of sensor $s_1$ to the farthest vertex of its Voronoi-Laguerre polygon.}
\label{fig:shrink}
}
\end{figure}

As an example of  sensing radius reduction, let us consider the sensor $s_1$  in Figure \ref{fig:shrink}.
In Figure \ref{fig:shrink}(a) the farthest vertex of $V(\mathscr{C}_1)$ is at a distance from $s_1$ which 
is smaller than its radius.
Because of Theorem \ref{th:lag_coverage} we can assert that all the points that are internal to 
$\mathscr{C}_1$ but do not belong to $V(\mathscr{C}_1)$ are covered 
by the sensors generating the Voronoi-Laguerre polygon to which they belong.
 Therefore $s_1$ redundantly covers 
the region within its circle that is external to its polygon and it can reduce its radius to cover {\color{black} no farther than}
 $f(V(\mathscr{C}_1))$, maintaining
full coverage of its responsibility region.
Such  a reduction of the  sensing radius of $s_1$ is shown in Figure \ref{fig:shrink}(b).
Changing the  sensing radius of $s_1$ requires the Voronoi Laguerre polygons of $s_1$ and its Laguerre neighbors to be recomputed, as shown in Fig. \ref{fig:shrink}(c). 
\reversemarginpar
\marginpar{Rimosso: The resulting Voronoi-Laguerre tessellation (shown in Figure \ref{fig:shrink}(c)) is still able to provide full coverage. }
\normalmarginpar
This reduction step can be repeated until the radius of the sensor $s_1$
is such that the farthest vertex of the polygon $V(\mathscr{C}_1)$
is on the circle $\mathscr{C}_1$ and the radius cannot be reduced any more (see Figure \ref{fig:shrink}(d)). 
A convergence proof is given in  Section \ref{sec:properties}, Theorem \ref{le:optimalVariable}.

This repeated reduction of the sensing radius is at the basis of $\alg$, where sensing radii of adjustable sensors 
are reduced until even a single radius reduction would leave a coverage hole.
Note that this process may even lead some sensors to shrink their sensing range to zero (in case of redundant sensors),
which means that such sensors are deactivated.

%


\subsubsection{On a characterization of boundary farthest vertices: Loose and strict farthest vertices}
\label{sec:skippable}
%
\begin{figure}[h]
 
{\color{black}
\centering
\vspace{0.7cm}
\hspace{-2cm}\begin{tabular}{ccc}
{\scalebox{0.4}{\includegraphics[]{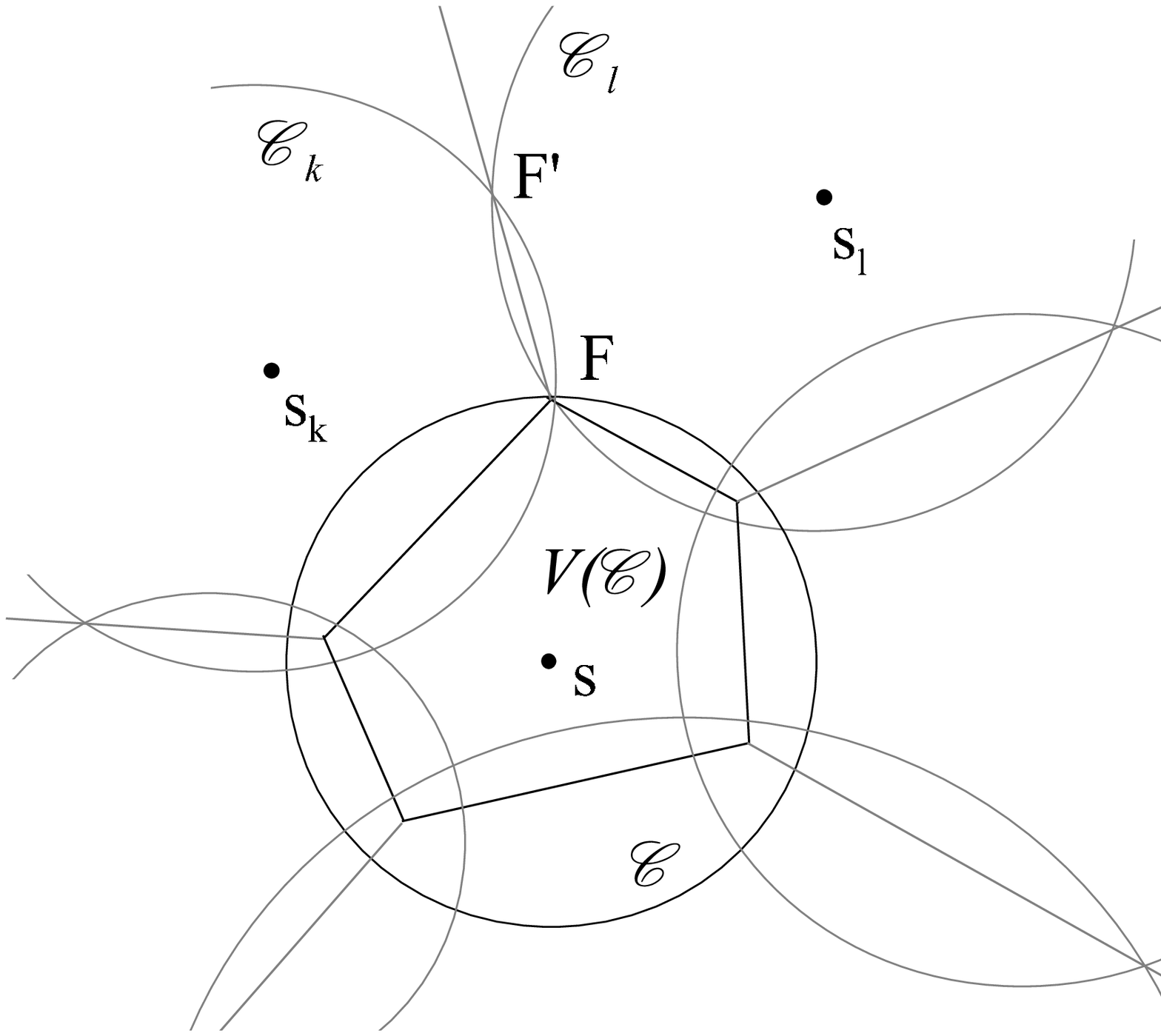}}}
&\hspace{1.5cm} &
{\scalebox{0.4}{\includegraphics[]{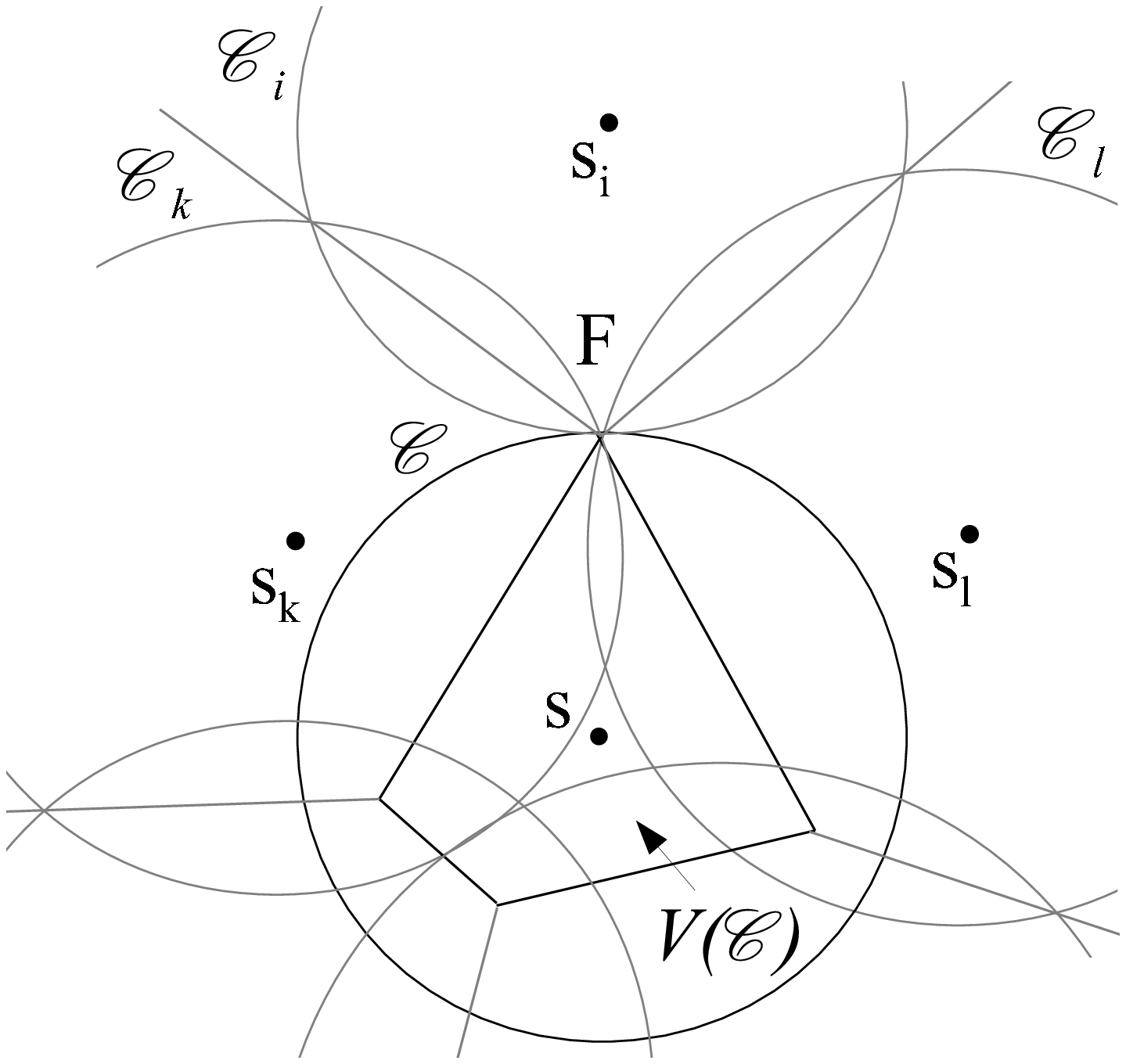}}}
\\ 
{\footnotesize{(a)}}&\hspace{0.8cm} &{\footnotesize{(b)}}
\\
\end{tabular}
 \caption{\color{black}Strict (a) and loose (b)  farthest vertices}
\label{fig:strict_and_loose}
}
\end{figure}

\marginpar{Questa frase mi continuava a dare l'impressione che non fosse raro il fatto di avere farthest che non sono strict ma il fatto di poterlo risolvere}
{\color{black}
$\alg$ typically considers the distance to the  farthest vertex of a Voronoi-Laguerre polygon as a lower bound for the reduction of the sensing radius of the generating sensor.
If the radius is reduced below this threshold, there is a loss of coverage in almost all cases. Nevertheless in some extremely rare configurations\footnote{In the experiments we obtained such a situation 
only by construction.} it is possible to reduce the radius below this distance without any coverage loss, by enforcing an ordering in the radius reduction of neighbor sensors.
}

%


Given  a sensor $s_i$, let $f(V(\mathscr{C}_i))$ be the farthest vertex of its Voronoi-Laguerre polygon $V(\mathscr{C}_i)$.
We call that 
the sensor $s_i$ is called the 
 {\em generating sensor} of the farthest vertex $f(V(\mathscr{C}_i))$ and we call
$f(V(\mathscr{C}_i))$  a {\em boundary farthest} if it lies on the boundary  of $\mathscr{C}_i$.
 
A boundary vertex is the intersection point of three circles and of their three Voronoi-Laguerre axes, 
and therefore is a boundary
 vertex for at least three sensors. In the following we say that
the  boundary farthest vertex of a sensor $s_{i}$ is 
a {\em strict farthest} if the radius of $s_{i}$ cannot be reduced without leaving a coverage hole. 
Otherwise such a vertex is called 
a {\em loose farthest}.
%
 %
 An example of strict and loose boundary farthest vertex 
is given in 
Fig.\ref{fig:strict_and_loose} (a) and (b), respectively. In the example all sensor nodes have reduced their radius 
to their farthest vertex which is therefore a boundary farthest vertex. This is when it makes the difference
whether a farthest vertex is loose or strict.
Let us focus on point $\mathrm{F}$ which is a common boundary 
farthest vertex for the three generating sensors $s$, $s_l$ and $s_k$.
As Fig.\ref{fig:strict_and_loose}}} (b) shows, $\mathrm{F}$ is a loose boundary farthest for sensor $s$, in fact,  $s$ 
can significantly reduce its sensing radius without compromising coverage.
However, a common farthest that is loose for a generating sensor is not necessarily loose for the others.
Point $\mathrm{F}$ is a strict farthest for the three other sensors $s_i$, $s_l$ and $s_k$ which 
cannot  reduce their radius. 

In general, if $s$ is the only generating sensor for which a  
boundary farthest is loose, it can reduce its radius 
without creating any coverage hole: The other generating sensors cannot perform any 
concurrent reduction since their farthest vertex is strict.
In this case, in order to calculate its new radius, 
$s$ has to subtract from its responsibility region $V(\mathscr{C})$ all the areas covered by the other generating 
sensors and guarantee to cover the farthest point of the remaining region $\overline{V}(\mathscr{C}) \triangleq V(\mathscr{C}) \setminus ( \mathscr{C}_k \cup \mathscr{C}_l)$. 
Fig. \ref{fig:loose_radius_reduction} (a) shows how the sensor $s$ seen in Figure \ref{fig:strict_and_loose} (b) 
can reduce its radius to the minimum needed to cover  the
farthest point B of the region ABCD =$\overline{V}(\mathscr{C})$, shaded in the figure.
 After this radius reduction, $s$ needs to recalculate its Voronoi-Laguerre polygon and possibly perform 
a further radius reduction, as in Fig. \ref{fig:loose_radius_reduction} (b). 

Although it is very unlikely to occur, it is theoretically possible for a boundary farthest vertex to be loose 
for two or more generating sensors.
In such a case, a concurrent radius reduction of the two or more sensors having a loose farthest vertex might result in
 a coverage hole.
For this reason we introduce a simple
{\em decision serialization scheme for loose farthest vertices}. This can be easily implemented
by means of  either a back-off policy or a leader election and a leader arbitrated sensor nodes
radius reduction.
 As there are many well established techniques to solve the problem of serializing decisions in a distributed 
computing setting, 
for the sake of simplicity and brevity, we do not address this aspect in the presentation 
of the algorithm.
We refer the reader to the Appendix of this paper for the details of the simple geometrical rules
sensors adopt to determine if their boundary farthest vertex
 is strict or loose. 

%
\begin{figure}[h]
{\color{black}
\centering
\vspace{0.7cm}
\hspace{-2cm}
\begin{tabular}{ccc}
{\scalebox{0.4}{\includegraphics[]{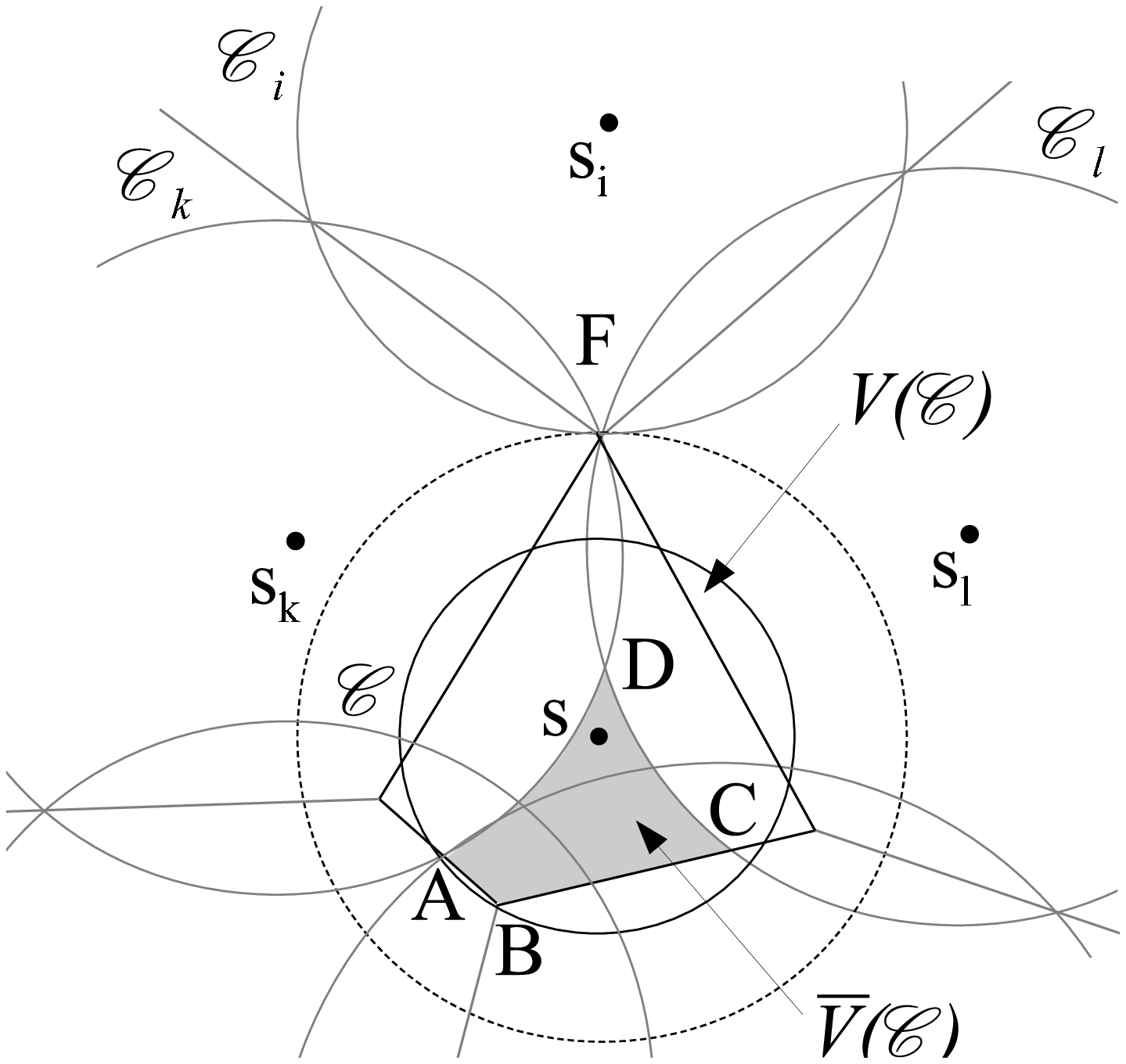}}}&
\hspace{2cm}
&
{\scalebox{0.4}{\includegraphics[]{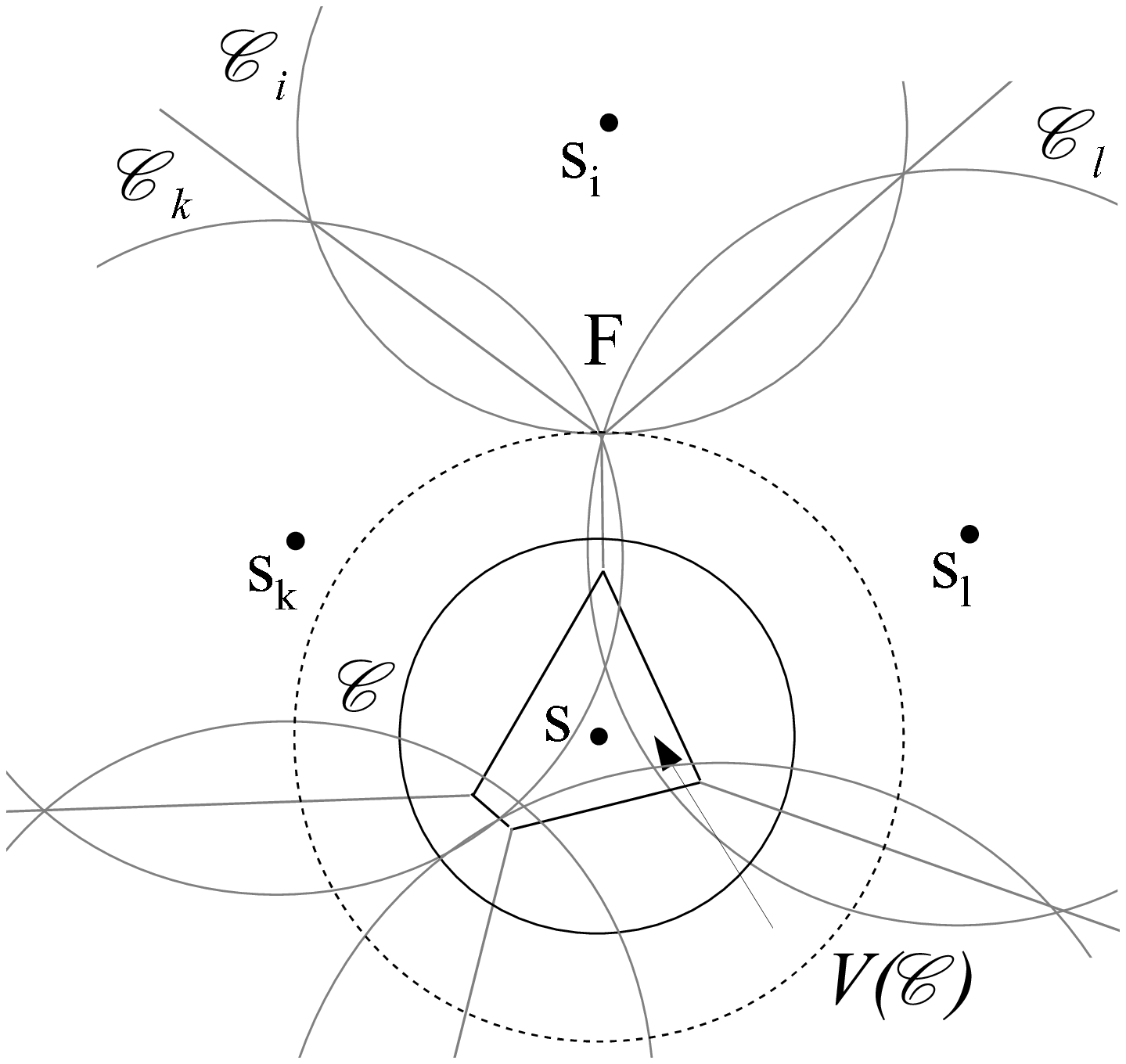}}}
\\ 
{\footnotesize{(a)}} &\hspace{2cm}&{\footnotesize{(b)}}
\\
\end{tabular}
 \caption{\color{black}Reduction of the sensing radius in a situation of loose boundary farthest vertex.}
 \label{fig:loose_radius_reduction}
}
\end{figure}

{\color{black}\subsection{Turning off sensors with fixed sensing radius}}
\label{sec:after_skip}

Not having the capability of tuning the extent of its sensing radius,
 the only way that a node with fixed radius has to save energy is to go to sleep when it is {\em redundant}.
 \marginpar{rimosso: , namely its sensing circle is 
 completely covered by other sensors.}
Therefore, the approach we take for selecting which node with fixed radius should deactivate is based on a greedy algorithm run by each node $s$ that, after a local exchange of information, determines if neighboring nodes can completely cover for $s$, and if $s$ is the ``best'' node for deactivation, i.e., the node that allows us to obtain the most energy conservation.
 
The extent of information needed by a node for deciding whether or not to deactivate can be kept 
significantly low by exploiting the Voronoi-Laguerre tessellation,
{\color{black} in agreement with Corollaries \ref{co:3.1}, \ref{co:3.2} and \ref{th:test}.}
\marginpar{NNN}
{\color{black}
Three cases may occur: (1) the sensing circle $\mathscr{C}$ of $s$ does not cover any point of its Voronoi-Laguerre polygon $V(\mathscr{C})$, (2)
the sensing circle  $\mathscr{C}$ only partially covers $V(\mathscr{C})$, (3) the sensing circle $\mathscr{C}$ completely covers the polygon $V(\mathscr{C})$.

In case (1),  Corollary \ref{co:3.1} states that
$s$ is certainly redundant.
In case (2), Corollary \ref{co:3.2} states that sensor $s$ cannot be turned off.
In case (3)  the sensor $s$ must evaluate the coverage of its Voronoi-Laguerre neighbors and of the sensors intersecting $V(\mathscr{C})$ which have null polygons, and determine its redundancy on the basis of Corollary \ref{th:test}. The mentioned corollaries set the limit to the number of nodes with which a node $s$ needs to exchange information in order 
to decide whether to deactivate or not.
}
%

\section{The algorithm $\alg$}
\label{sec:algo}

$\alg$ is executed in parallel by 
all the sensors of the network. Its execution results in the selection of a subset of sensors to be kept awake while
the others go to sleep, i.e., they are put in a low energy modality or turned off.
$\alg$ also allows a node with adjustable radius that is awake to set its sensing range.
The obtained sensor activation and radius adjustment is used for a time, called \emph{operative time interval}, that lasts until $\alg$ is re-executed.
The operative time interval is not necessarily fixed since $\alg$ execution can be event-driven.%
\footnote{\ An event-driven reconfiguration requires that sensors operating in low power mode can be contacted 
by the sink by means of an interest dissemination. 
Deactivated nodes equipped with a wake-up radio\cite{wakeup}  can be woken up upon need and can 
therefore safely turn off their radio for the whole duration of the operative time interval.
If such extra HW is not available nodes in low power mode must periodically wake up according to  a very low
duty cycle so that changes in the mode of operation of the network can be signaled.}

Each sensor makes the decision about whether to (de)activate and about reducing its radius (if possible) iteratively. 
In order to do so, at each iteration $k$ each node determines its own Voronoi-Laguerre  polygon.
This requires the node to be aware of its one-hop neighbors (nodes it can communicate with directly), their location%
\footnote{\ This information may be obtained through extra hardware such as GPS, if available, or through one of the many 
localization schemes recently proposed.  
} and their sensing radius.
%
%
The iteration is then composed by two phases. During the first phase nodes with fixed radius 
decide whether to go to sleep or not.
In the second phase, the nodes with adjustable radius perform their  radius reduction.
Each node $s_i$ bases its decision on a parameter 
\marginpar{it was: $\alpha^{(k)} _i \in [0,1]$, }
$\alpha^{(k)} _i \in (0,1]$, 
which  depends on the energy gain  that the sensor will achieve by either going to sleep or by reducing its sensing radius.
This parameter is used differently depending on whether a node has a fixed or an adjustable radius.
Specifically, a node $s_{i}$ with fixed radius will go to sleep with probability $\alpha _i^{(k)}$ provided that there are neighboring nodes that are awake and {\color{black}redundantly cover its sensing circle}.
\marginpar{NNN}%
If $s_i$ has an adjustable radius it will reduce it by the fraction $\alpha _i^{(k)}$ of the maximum radius reduction that does 
not 
{\color{black} alter the coverage of its responsibility region}.
%
As we will prove in Section~\ref{sec:properties}, the iterative execution of the two phases 
leads to a network configuration in which there is no redundant fixed sensor and it is not possible to further 
reduce the radius of any adjustable sensor without creating new coverage holes.

\subsection{$\alg$ in details}
\label{ssec:details}

\subsubsection{Initialization}

$\alg$ is described by Algorithms~\ref{alg:selective_activation_FIXED} and~\ref{alg:selective_activation_ACTIVE}, for nodes with fixed and adjustable radius, respectively.
At the start of $\alg$ operations, each sensor sets the iteration counter $k$ and the value of its sensing radius (the maximum value in the case of sensors with adjustable radius). The flag ${\texttt{decision\_made}}$ is set to \texttt{false} indicating that the node is undecided.
The node remains awake and undecided until in one of the iterations it makes a final decision on the value of its sensing radius to be used till a new $\alg$ execution.

Initialization also includes the setting of a timer needed for protocol operations.

\subsubsection{Computing $\alpha _i^{(k)}$}
\label{sec:alpha}

{\color{black}
Consider the $k$-th iteration of $\alg$.
Let $S^{(k)}_\texttt{undecided}
\subseteq
S_A^{(k)}$ be the set of sensors that  have not made their final configuration decision, where 
 $S_A^{(k)}$ is the set of sensors that are still awake.
Similarly, $S^{(k)}_\texttt{decided}=S_A^{(k)}\setminus S^{(k)}_\texttt{undecided}$ is the set of sensors that are still awake and have already made their configuration decision.

Consider $s_i \in S^{(k)}_\texttt{undecided}$. Let 
${\mathscr{L}}^{(k)}(s_i)$  be the subset of $S^{(k)}_\texttt{undecided}$ including $s_i$ and all the undecided sensors that  are either Voronoi-Laguerre neighbors 
of $s_i$
or have a null polygon and their sensing circle intersects $\mathscr{C}_i$: ${\mathscr{L}}^{(k)}(s_i) = S^{(k)}_\texttt{undecided} \cap (\mathcal{N}_{S_A^{(k)}} (s_i)\cup \mathcal{N}_{S_A^{(k)}}^{\emptyset} (s_i) \cup \{ s_i \} ) $.
\marginpar{NNN: nota che ho aggiunto $s_i$ alla def del set ${\mathscr{L}}^{(k)}(s_i)$ altrimenti il valore del gain non \'e tra 0 e 1. Controllare che non fosse gi\'a negli N.}

Let also ${\mathscr{D}}^{(k)}(s_i) = \mathcal{S}_\texttt{decided}^{(k)}    \cap (\mathcal{N}_{S_A^{(k)}} (s_i)\cup \mathcal{N}_{S_A^{(k)}}^{\emptyset} (s_i) ) $
be the subset of the sensors that have already made their decision  and are either 
Voronoi-Laguerre neighbors of $s_i$ or have a null polygon and overlap the sensing circle $\mathscr{C}_i$.
}

The computation of the parameter $\alpha _i^{(k)}$ depends on the comparison {\color{black} between} $s_i$ and the nodes in 
${\mathscr{L}}^{(k)}(s_i)$ with respect to the decrease in energy consumption that is achievable through sensing radius reduction while 
ensuring coverage.
The comparison is motivated by the fact that these nodes are those that still have the chance to reduce their 
sensing radius and consequently their energy expenditure. 
The value of $\alpha _i^{(k)}$ should be higher for a node $s_i$ when choosing it for sensing radius 
reduction or for going to sleep leads to a better performance gain than choosing the other nodes in the neighborhood.

{\color{black}
The criterion we propose to compute $\alpha _i^{(k)}$ is based on the
 \emph{energy gain}, defined as the amount of energy that a sensor can save by reducing its sensing radius 
to the farthest point of the responsibility region
 (in case of sensors with adjustable radius) or by going to sleep (case of sensors with fixed sensing radius).

We recall that  $E_\texttt{sensing}()$ is the energy expenditure per unit time due to sensing, defined  in Equation \ref{eq:e_active_sensing}.
\marginpar{Cambiato questo pezzetto}
{\color{black}The energy gain of sensor $s_i$ in the $k$-th iteration is defined as $\Delta{E}_i^{(k)}= E_\texttt{sensing}(r^\texttt{fixed}_i)$  for sensors with fixed sensing radius. For sensors with adjustable sensing radius, it is either $\Delta{E}_i^{(k)}=E_\texttt{sensing}(r_i^{(k-1)})$ for sensors which have a null or an uncovered polygon, or  $\Delta{E}_i^{(k)}= E_\texttt{sensing}(r_i^{(k-1)})- E_\texttt{sensing}(d_E(s_i,f(\overline{V}^{(k)}(\mathscr{C}_{i}))))$, otherwise. Here 
$\overline{V}^{(k)}(\mathscr{C}_{i})=    V^{(k)}(\mathscr{C}_i) \setminus \cup_{{s_j} \in
{\mathscr{D}}^{(k)}(s_i)  } \mathscr{C}_j$.
}

The energy gain criterion sets the value of $\alpha_i^{(k)}$ as follows:
%
\marginpar{NNN: ho cambiato questa def, mettendo l'$\alpha_\texttt{min}$ perch\'e altrimenti \'e possibile costruire situazioni in cui converge solo il primo nodo della lista e la convergenza degli altri partirebbe solo dopo la fine di questo, che per\'o non accadrebbe in tempo finito (cio\'e nemmeno lo start dei secondi). Immagina una scacchiera dove un nodo con alfa = 0 ha solo vicini con alfa=1 e viceversa. Se quelli con 1 ci mettono infinito, quelli con zero non partono mai. }
\begin{equation}
\alpha_i^{(k)} =
\max \left\{
\frac{ \Delta E_i^{(k)}-\Delta E^{\texttt{min }(k)}_i   }{\Delta E^{\texttt{max }(k)}_i-\Delta E^{\texttt{min }(k)}_i}, \alpha_\texttt{min}\right\},
\label{eq:alfa}\end{equation}
where 
the parameter $\alpha_\texttt{min}$ is an arbitrarily small constant, such that $0<\alpha_\texttt{min} \ll 1$,}
$\Delta{E}^{\texttt{max }(k)}_i= \max_{s_j \in {\mathscr{L}}^{(k)}(s_i)} \Delta E_j^{(k)}$ is the maximum achievable gain in the neighborhood of $s_i$ and $\Delta{E}^{\texttt{min }(k)}_i= \min_{s_j \in {\mathscr{L}}^{(k)}(s_i)} \Delta E_j^{(k)}$ is its minimum value.
{\color{black}
If $\Delta{E}^{\texttt{max }(k)}_i=\Delta{E}^{\texttt{min }(k)}_i$ we consider $\alpha_i^{(k)} =1$.
According to Eq. \ref{eq:alfa}, the more a node $s_i$ allows energy savings the higher is the probability that it is selected for going to sleep if $s_i$ is a fixed sensor, or the higher is the  reduction of sensing radius that is allowed if $s_i$ is an adjustable sensor.
{\color{black}
 This setting of $\alpha_\texttt{min}$ ensures that even the sensor with smallest potential energy gain 
 can make a decision that improves its energy expenditure. 
}

}

{\color{black}
The energy gain criterion has been compared by means of extensive simulations
with several others, including one based on the 
 node residual energy
 and one based on an estimate of the node expected lifetime.
In all the scenarios the energy gain criterion showed superior performance.
Therefore, we will focus only on such a criterion for the remainder of the paper.
}


\subsubsection{$\alg$ for sensors with fixed sensing radius}
\label{sec:fixed_sara}
\marginpar{{\color{black}NNN: mi sembra che in questa sezione non ci sia niente da cambiare riguardo al caso di poligono completamente scoperto, infatti si parla semplicemente di verificare se il nodo � "redundant" non si dice come viene fatta questa verifica.}}
At the beginning of $\alg$ operations, all the sensors with fixed radius are awake and undecided.
Let us consider the $k$-th iterative step of $\alg$ ($k$-th execution of the $\texttt{while}$ cycle in 
Algorithm~\ref{alg:selective_activation_FIXED}).
The set of sensors that are still awake at the $k$-th iteration is referred to as 
$S_A^{(k)}=S^{(k)}_\texttt{fixed} \cup S^{(k)}_\texttt{adjustable}$.

Each undecided sensor $s_i \in  S^{(k)}_\texttt{fixed}$ performs an information exchange with its neighbors 
that are still undecided 
 to gather information regarding their radius and  position\footnote{
 {\color{black} It is not necessary to exchange information with the sensors that have already made their configuration 
decisions. 
Also, the node location is communicated at the start of each SARA execution and only if the node
location has changed. }
 }.
With this information, $s_i$ is able  to construct its Voronoi-Laguerre polygon $V(\mathscr C_{i}^{(k)})$ and to 
determine the set
$\mathcal{N}_{S_A^{(k)}}(s_i)$.

Node $s_i$ then informs its neighbors with which it has a sensing overlap  {\color{black} about the nullity of its polygon}.  
This information allows its neighbors to compute their sets $\mathcal{N}_{S_A^{(k)}}^\emptyset$.
Each node then evaluates  
 its redundancy status (according to Corollary \ref{th:test}).


%
If $s_i$ is not redundant at the $k$-th iteration, it cannot become redundant in any of the successive 
iterations because $\alg$ in each iteration can only reduce the number of sensors that can cover an area.
Therefore, in the case of non redundancy, $s_i$ communicates this to 
the neighbors with a sensing overlap (sending a $\texttt{turn-on}$ message), ends 
the decision phase (setting the $\texttt{decision\_made}$ 
flag to $\texttt{true}$), and stays awake.

If sensor $s_i$ is redundant it communicates {\color{black} its potential energy gain
to the nodes in 
$\mathcal{N}_{S_A^{(k)}}(s_i) \cup \mathcal{N}_{S_A^{(k)}}^\emptyset(s_i)$.

Nodes with a null polygon also send their potential energy gain to all
their neighbors with sensing overlap.
Each node is  then able to construct the set ${\mathscr{L}}^{(k)}(s_i)$ and compute $\alpha _i^{(k)}$.}
{\color{black}
The calculus of $\alpha_i^{(k)}$ is executed by running the function $\texttt{get\_alpha}$ described in 
Algorithm \ref{alg:selective_activation_FUNCTIONS}.
}

{\footnotesize
\begin{algorithm}[!h]
\caption{Algorithm $\alg$ for fixed sensors}
\label{alg:selective_activation_FIXED}
{\footnotesize
\newcommand{\algo}{\textbf{Algorithm }}
\algo $\alg$ executed by node $s_i$ \\
\SetKwBlock{With}{}{}
\hspace{0.5cm}
\parbox{10cm}{
\SetKw{init}{Initialization:\\}
\init
$k=0$\;
Back-off interval = $[0,t_\texttt{max}^\texttt{backoff} ]$\;
{$r_{i}^{(k)}=r_i^\texttt{fixed}$\;}
decision\_made=false\;
Exchange position information with neighbors\;
\vspace{0.3cm}
 {\bf Iterative Voronoi-Laguerre diagram construction:\\}
 \While{$!\texttt{decision\_made}$}
{
Exchange info on radius   with neighbors\; 
Construct the VorLag polygon $V^{(k)}({\mathscr C}_{i})$\; 
Exchange info  
on null polygons\;

 {\color{black} Evaluate redundancy and energy gain\;
\eIf{$s_i$ is not redundant}
{
\tcp{Case of fixed sensors that need to stay awake}
Send \texttt{turn-on} message\;
decision\_made=true\;
Stay awake\;
} 
{
\tcp{Case of redundant fixed sensor}
{\color{black}  Exchange info on energy gain}\; }
Build set $\mathscr{L}^{(k)}(s_i)$\;
$\alpha_i^{(k)} = $\FuncSty{get\_alpha($\mathscr{L}^{(k)}(s_i)$)}\;
Choose a random instant $t_i^* \in [0,  t_\texttt{max}^\texttt{backoff}]$\;
\While{$t<t_i^*$}{
Listen to update messages from the neighborhood\;
}
\eIf{$s_i$ is not redundant anymore}
{
Send  \texttt{turn-on} message\;
decision\_made=true\;
Stay awake; 
}
{

With probability $\alpha^{(k)}_i$
\With{
Send  \texttt{turn-off} message\;
decision\_made=true\;
Go to sleep\;
}
}
} 
$k=k+1$;

} 
} 
} 
\end{algorithm}
}

{\color{black}
{\small
\begin{algorithm}
{\footnotesize
\caption{Function to compute parameter $\alpha _i$}
\label{alg:selective_activation_FUNCTIONS}

\newcommand{\func}{\textbf{Function }}
\func \FuncSty{get\_alpha(${\mathscr{L}}^{(k)}(s_i)$)
}\\
\hspace{0.5cm}
\parbox{10cm}{
{
Set $\Delta{E}^{\texttt{max }(k)}_i= \max_{s_j \in {\mathscr{L}}^{(k)}(s_i)} \Delta E_j^{(k)}$ and \\ 
$ \textbf{ }\textbf{ }\textbf{ }\textbf{ } \Delta{E}^{\texttt{min }(k)}_i= \min_{s_j \in {\mathscr{L}}^{(k)}(s_i)} \Delta E_j^{(k)}$\;
{\color{black}$\alpha_i^{(k)}=\max \left\{ \frac{ \Delta E_i^{(k)} -\Delta E^{\texttt{min }(k)}_i(n)  }{\Delta E^{\texttt{max }(k)}_i-\Delta E^{\texttt{min }(k)}_i} , \alpha_\texttt{min}\right\}$\;}
{\bf return }$\alpha_i^{(k)}$\;
}
}}
\end{algorithm}
}}
Since more than one sensor may decide to turn {\color{black}themselves}\marginpar{was:itself} off at the same iteration, possibly 
leaving coverage holes, we introduce a simple back-off scheme to avoid conflicting decisions.
More precisely, given a back-off interval $t^\texttt{backoff}_\texttt{max}$, each sensor $s_i$ chooses a random instant $t_i^* \in [0, t^\texttt{backoff}_\texttt{max}]$, hereafter called {\em backoff timeout}. It then waits for a time $t_i^*$, during which it considers all the messages received from the nodes in radio proximity that may make $s_i$ redundant.

After the expiration of the backoff timeout $t_i^*$, the sensor $s_i$
verifies if it is still redundant or not.
If it is not redundant anymore, $s_i$ decides to stay awake
and sets the $\texttt{decision\_made}$ flag to $\texttt{true}$.
It then communicates this decision to its neighbors by sending them a $\texttt{turn-on}$ message.
\marginpar{Rimosso: lo fanno solo gli adjustable ed \'e spiegato l\'i.\\
Upon receiving this message, the neighbors of $s_i$ can assume that $\mathscr{C}_i$ is covered.
Therefore, they can subtract this area from the overall AoI when computing the fraction of their own polygon that needs to be covered.
}

If instead $s_i$ is still redundant, it goes to sleep with probability
$\alpha^{(k)}_i$. In the case the node goes to sleep, it sets the $\texttt{decision\_made}$ flag to $\texttt{true}$ and
communicates its decision by sending  a $\texttt{turn-off}$ message.

Notice that a redundant sensor with fixed sensing radius does not necessarily go to sleep  at the first iteration.
Therefore, the execution of a single iteration of the algorithm does not eliminate the existing redundancy completely.
Nevertheless, at each iteration the sensors with higher priority are the ones that more likely will go to sleep. 
The other redundant sensors will eventually  either go to sleep or become non-redundant in one of the 
subsequent iterations depending on the decisions of their neighbors.

\subsubsection{$\alg$ for sensors with adjustable sensing radius}
\label{sec:sara_adjustable}

All sensors with adjustable radius start executing $\alg$ by setting their radii to their maximum value.
They are also all undecided.
As before, we consider the generic $k$-th iteration of $\alg$ ($k$-th execution of the $\texttt{while}$ 
cycle in Algorithm~\ref{alg:selective_activation_ACTIVE}).
%
At each algorithm iteration, the radius reduction decision at node $s_i \in S_\texttt{adjustable}^{(k)}$ is made 
after the back-off phase of its neighbors in $S_\texttt{fixed}^{(k)}$.  %
At the end of such a phase, every sensor $s_i \in S^{(k)}_\texttt{adjustable}$ updates its Voronoi-Laguerre 
polygon $V^{(k)}(\mathscr{C}_i)$, updates its information for computing $\alpha_{i}^{(k)}$
if any of its fixed radius neighbor {\color{black} turned itself}  off during the back-off, 
and determines the sets 
${\mathscr{L}}^{(k)}(s_i)$ and ${\mathscr{D}}^{(k)}(s_i)$.
In this way, the sensor $s_i$ has the necessary information to calculate the maximum radius reduction
that does not create coverage holes. 
Notice that in this calculus, the sensors belonging to the two sets 
$S^{(k)}_\texttt{decided}$ and $S^{(k)}_\texttt{undecided}$ play a different role
since the sensors in $S_\texttt{decided}^{(k)}$ will no longer change their configuration for the current execution of $\alg$, therefore their sensing circles can be considered definitely covered and can be subtracted from the responsibility region of those sensors that still have to make their configuration decision.
This is the reason why the maximum radius reduction for $s_i$ is computed as the one that {\color{black}does not alter}\marginpar{rimosso: guarantees,
dato che c'\'e il caso di poligono tutto scoperto} the coverage of the region $
 \overline{V}^{(k)}(\mathscr{C}_{i})=V^{(k)}(\mathscr{C}_i)\setminus \cup_{{s_j} \in
\mathscr{D}^{(k)}(s_i)
} \mathscr{C}_j$.

\marginpar{Qui ho un po' deviato da quanto concordato, si era detto che nel caso completely uncovered i nodi si possono spegnere immediatamente, qui invece gareggiano. Il motivo \'e che anche altri vicini possono avere il poligono completely uncovered, e quindi potrebbero spegnersi pure loro cambiando la forma del poligono di $s_i$. Alla fine questo processo termina sempre con uno spegnimento di $s_i$ ma questa cosa andrebbe dimostrata, perci\'o non ho affrontato il problema.}

{\color{black}
The 
minimum extent  of $s_i$ sensing radius reduction $\overline{d}_i^{(k)}$ is based on Corollaries \ref{co:3.1}, \ref{co:3.2}, and \ref{th:test}. 
If $s_i$ is not able to cover any point of $\overline{V}^{(k)}(\mathscr{C}_{i}))$ then $\overline{d}_i^{(k)}=0$ (due to Corollary  \ref{co:3.1}). If $s_i$ only partially covers its polygon (
$d_{E}(s_{i},c(\overline{V}^{(k)}(\mathscr{C}_{i})))<r_{i}^{(k)} < d_{E}(s_{i},f(\overline{V}^{(k)}(\mathscr{C}_{i})))  $ where  $c(\overline{V}^{(k)}(\mathscr{C}_{i}))$ is the closest point of 
 $\overline{V}^{(k)}(\mathscr C_{i})$ from $s_i$
)
then $\overline{d}_i^{(k)}=r_i^{(k)}$ (the radius does not change, as determined by Corollary \ref{co:3.2}). Finally, if $s_i$ completely covers its polygon, 
  $\overline{d}_i^{(k)}$ is set to  $d_E(s_i,f(\overline{V}^{(k)}(\mathscr{C}_{i})))$, that is the Euclidean distance between $s_i$ and the farthest point of $\overline{V}^{(k)}(\mathscr C_{i})$.
  }

The sensor $s_i$, whose radius at the $k$-th iteration is $r_i^{(k)}$, will then reduce its radius to an intermediate value in the range $[\overline{d}_i^{(k)},r_i^{(k)}]$, whose position is determined by the priority value $\alpha _i^{(k)}$.
Therefore $s_i$ calculates the new value of its radius $r_i^{(k+1)}$ as
$ r_i^{(k+1)}= r_i^{(k)} -  \alpha_i^{(k)} \cdot (  r_i^{(k)} -   \overline{d}_i^{(k)}  )      $.

 %
 Each sensor belonging to $S_\texttt{adjustable}$ that reduces its radius affects the potential decisions of its Voronoi-Laguerre neighbors, so the process is iterated until no further reductions are possible, because either 
a strict farthest vertex is on the boundary of the sensing circle, or the Voronoi-Laguerre polygon of the sensor gradually became null, and the sensor is put to sleep.

 {\footnotesize
 \begin{algorithm}[!h]
 \caption{Algorithm $\alg$ for adjustable sensors}
 \label{alg:selective_activation_ACTIVE} 
\ls{1}
{\footnotesize
 \newcommand{\algo}{\textbf{Algorithm }}
 \algo $\alg$ executed by node $s_i$ \\
\parbox{10cm}{
 \tcp{before starting the next operative time interval,
 the sensor $s_i$ works with the radius  \\ determined at the previous execution of $\alg$}
\SetKw{init}{Initialization:\\}
 \init
 $k=0$\;
 Back-off interval = $[0,t_\texttt{max}^\texttt{backoff} ]$\;
 {$r_{i}^{(k)}=r_i^\texttt{max}$\;}
 \texttt{decision\_made}=\texttt{false}\;
 Exchange position information with neighbors\;
 \vspace{0.3cm}
 
  {\bf Iterative Voronoi-Laguerre diagram construction:\\}
  \While{$!\texttt{decision\_made}$}
  {
 Exchange info on radius with neighbors\;
   Construct the VorLag polygon $V^{(k)}(\mathscr C_{i})$\;
   Exchange redundancy/polygon nullity information messages and potential energy gain\;
   \While{$t<t_\texttt{max}^\texttt{backoff}$}{
   listen to update messages from the fixed nodes in the neighborhood\;
   }
   Update the VorLag polygon $V^{(k)}(\mathscr C_{i})$\;
 Build sets $\mathscr{L}^{(k)}(s_i)$ and ${\mathscr{D}}^{(k)}(s_i)$\;
%
  Let $\overline{V}^{(k)}(\mathscr{C}_{i})=
  V(\mathscr{C}_{i}) \setminus \cup_{{s_j} \in
{\mathscr{D}}^{(k)}(s_i)} \mathscr{C}_j $\;
{\color{black}
  Let $f(\overline{V}^{(k)}(\mathscr{C}_{i}))$ be the farthest point of 
 $\overline{V}^{(k)}(\mathscr C_{i})$ from $s_i$\;
 Let $c(\overline{V}^{(k)}(\mathscr{C}_{i}))$ be the closest point of 
 $\overline{V}^{(k)}(\mathscr C_{i})$ from $s_i$\;

\eIf{  $(d_{E}(s_{i},c(\overline{V}^{(k)}(\mathscr{C}_{i})))<r_{i}^{(k)} < d_{E}(s_{i},f(\overline{V}^{(k)}(\mathscr{C}_{i})))  ) 
\vee\\ (f(\overline{V}^{(k)}(\mathscr{C}_{i}))$
 is a strict farthest ) $\vee (r_i^{(k)} = 0))$}
 {
 \tcp{reached minimum radius}
 $\texttt{decision\_made}=\texttt{true}$\;}
 { 
\eIf{$r_{i}^{(k)}< d_{E}(s_{i},c(\overline{V}^{(k)}(\mathscr{C}_{i})))$}
  {
 \tcp{completely uncovered polygon}
$ \overline{d}_i^{(k)}=0$\;
  }
{ 
 
  $\overline{d}_i^{(k)}=d_{E}(s_{i},f(\overline{V}^{(k)}(\mathscr C_{i})))$\;
  }
 $\alpha_i = get\_alpha(\mathscr{L}(s_i))$\;
$r_i^{(k+1)}= r_i^{(k)} -  \alpha_i^{(k)} (  r_i^{(k)} -   \overline{d}_i^{(k)}  ) $\;
 $k=k+1$\;
 
 }
  }
 }
 }
 
 \eIf{$r_{i}^{(k)}=0$}
  {
 \tcp{null {\color{black} or completely uncovered} polygon}
  go to sleep\;
}
{
Adjust the sensing radius to $r_{i}^{(k)}$\;
}
} 
 \end{algorithm}
}

\section{Properties of $\alg$} \label{sec:properties}

{\color{black}
The execution of $\alg$  on  a set of sensors $\mathcal{S}$ leads to a final configuration 
that will be hereby called {\em cover set}.
In the following we will shortly denote with $S_\texttt{SARA}$  such a cover set, where $S_\texttt{SARA}$ is a set of awake sensors with their radius configuration decided by $\alg$.
The following theorem 
shows that the cover set calculated by $\alg$ provides the same coverage as the starting configuration  
(the one where all sensors are active at maximum radius).
%
%
\marginpar{COVERAGE EQUIVALENCE}
\begin{theorem}\label{th:correctness}
{\bf{\color{black}(Coverage equivalence)}} Consider a set of adjustable and fixed sensors $\mathcal{S}=\mathcal{S}_\texttt{adjustable} \cup \mathcal{S}_\texttt{fixed}$. 
Let  $\mathscr{A} \subseteq AoI$ be the area that the sensors in $\mathcal{S}$ are able to cover if they are all active and the adjustable sensors work at
their maximum radius.
Let $S_\texttt{SARA}$ be the cover set calculated by  $\alg$. The coverage extension of $S_\texttt{SARA}$ is equal to $\mathscr{A}$.
\end{theorem}

\begin{proof}
Let us denote with $S^{(k)}_\texttt{SARA}$ the cover set determined by $\alg$ at the $k$-th iteration, with $S^{(0)}_\texttt{SARA}=S$.
Let us also denote with $\mathscr{A}^{(k)}\subseteq AoI$ the portion of the $AoI$ that is covered by $S^{(k)}_\texttt{SARA}$, 
therefore 
$$\mathscr{A}^{(k)}=\cup_{s_j \in S^{(k)}_\texttt{SARA}} \mathscr{C}^{(k)}_j.$$

The Voronoi-Laguerre diagram of $S^{(k)}_\texttt{SARA}$ creates a partition of the $AoI$.
Therefore,
in order to prove that the coverage extension does not decrease after the algorithm execution, it is enough to
prove that, at each iteration, the coverage of each polygon is preserved, that is:

\begin{equation}
V(\mathscr{C}^{(k)}_i)   \cap \mathscr{A}^{(k)}
  \subseteq 
   \mathscr{A}^{(k+1)}, \forall s_i \in S^{(k)}_\texttt{SARA}.
\label{eq:coverage_preserved}
\end{equation}

Regarding fixed sensors, $\alg$ allows  them to go to sleep one at a time and only if their polygon is already covered by other sensors, so if one of them decides to go to sleep the coverage of its polygon does not decrease, thus guaranteeing that Eq. \ref{eq:coverage_preserved} is trivially verified for fixed sensors.

For what concerns the case of adjustable sensors, let us consider any sensor $s_i$ still active in the $k$-th iteration.
Theorem
\ref{th:lag_coverage} affirms that the covered area of $V({\mathscr C}^{(k)}_i)$ is
all covered by $s_i$. 
This means that, for any $s_i \in S^{(k)}_\texttt{SARA}$ and for any iteration $k$:
 $$V(\mathscr{C}^{(k)}_i)   \cap \mathscr{A}^{(k)} \stackeq{ Th. \ref{th:lag_coverage}} V(\mathscr{C}^{(k)}_i)   \cap  \mathscr{C}^{(k)}_i.$$
 Therefore in order to prove Eq. \ref{eq:coverage_preserved}, it is sufficient to prove that 
 
 \begin{equation}
V(\mathscr{C}^{(k)}_i)   \cap  \mathscr{C}^{(k)}_i
  \subseteq 
   \mathscr{A}^{(k+1)}, \forall s_i \in S^{(k)}_\texttt{SARA} \cap S_\texttt{adjustable}.
\label{eq:coverage_preserved_bis}
\end{equation}

{\color{black}
 $\alg$ reduces the radius of an adjustable sensor $s_i$ to a value such that the coverage of the region 
$\overline{V}(\mathscr{C}_i^{(k)})=V(\mathscr{C}_i^{(k)}) \setminus
(\cup_{s_j \in  \mathscr{D}_i^{(k)}      } \mathscr{C}_j^{(k)})$ is not altered.}

By the definition of $\mathscr{D}_i^{(k)}$, sensors belonging to $\mathscr{D}_i^{(k)}$ are such that their sensing circles do not change in the following iterations, therefore if 
 $s_j  \in \mathscr{D}_i^{(k)}$ then $\mathscr{C}_j^{(k)}=\mathscr{C}_j^{(k+1)}$.

Let us consider a further partition of $V(\mathscr{C}_i^{(k)})$ in the following two subsets: 

$V_{i,k}^1\triangleq  \overline{V}(\mathscr{C}_i^{(k)})=V(\mathscr{C}_i^{(k)}) \setminus
(\cup_{s_j \in  \mathscr{D}_i^{(k)}      } \mathscr{C}_j^{(k)})$, and $V_{i,k}^2 \triangleq  V(\mathscr{C}_i^{(k)} )\setminus  \overline{V}(\mathscr{C}_i^{(k)}) = 
V(\mathscr{C}_i^{(k)} )\cap
( \cup_{s_j \in  \mathscr{D}_i^{(k)}      } \mathscr{C}_j^{(k)})$.

We will now prove
that Eq. \ref{eq:coverage_preserved_bis} is verified  by separately  considering the two subsets $V_{i,k}^1$ and $V_{i,k}^2$.
Let us first consider $V_{i,k}^1$.
$$V_{i,k}^1 \cap \mathscr{C}_i^{(k)} \stackeq{ {\texttt{SARA}}} V_{i,k}^1 \cap \mathscr{C}_i^{(k+1)}  \subseteq \mathscr{C}_i^{(k+1)} \subseteq  \mathscr{A}^{(k+1)}.$$


We now show that the same property holds for $V_{i,k}^2$.

$$V_{i,k}^2 \cap \mathscr{C}_i^{(k)} \subseteq   V_{i,k}^2  = V(\mathscr{C}_i^{(k)}) \cap (\cup_{s_j \in  \mathscr{D}_i^{(k)} } \mathscr{C}_j^{(k)}        )   \stackeq{ {\texttt{Def of } \mathscr{D}_i^{(k)}}}$$
$$
= V(\mathscr{C}_i^{(k)}) \cap (\cup_{s_j \in  \mathscr{D}_i^{(k)} } \mathscr{C}_j^{(k+1)} ) \subseteq \cup_{s_j \in  \mathscr{D}_i^{(k)} } \mathscr{C}_j^{(k+1)} \subseteq \mathscr{A}^{(k+1)}.$$

Since $V( \mathscr{C}_i ^{(k)}) \cap {\mathscr{C}_i^{(k)}} = (V_{i,k}^1 \cup V_{i,k}^2) \cap {\mathscr{C}_i^{(k)}}  \subseteq  \mathscr{A}^{(k+1)}$,
Eq.  \ref{eq:coverage_preserved} is verified.


%
\end{proof}

\reversemarginpar

\normalmarginpar

\marginpar{CONVERGENZA SOLO ADJUSTABLE }
\begin{theorem}\label{le:optimalVariable}
{\bf (Convergence in the case of adjustable sensors)}
Given a set $\mathcal{S}=\mathcal{S}_\texttt{adjustable}$ of only adjustable sensors, under the execution of  $\alg$,
each sensor will converge to a final configuration decision.
\end{theorem}

 \begin{proof}
 
 Consider the adjustable sensor $s_i \in \mathcal{S}$, positioned in $\mathrm{C}_i$.
  Let $r_i^{(k)}$ be its sensing radius at the  $k$-th iteration of $\alg$, and let 
 $\mathscr{C}^{(k)}_i$  and $V(\mathscr{C}^{(k)}_i)$ be its sensing circle and its Voronoi-Laguerre polygon, respectively.
 We distinguish three cases:
 (1)  $V(\mathscr{C}_i^{(k)})$ is completely covered (notice that this case includes the situation of null polygons which can be considered a degeneration of non null polygons), (2) $V(\mathscr{C}_i^{(k)})$ is only partially covered
and (3) $V(\mathscr{C}_i^{(k)})$ is not covered  (neither by $s_i$ nor by any other sensor, due to Theorem \ref{th:lag_coverage}).

{\em Convergence in case (1).}
Theorem \ref{th:lag_coverage} ensures that $V(\mathscr{C}_i^{(k)})$ is completely covered by $s_i$.
Since $\alg$ preserves coverage (for Theorem \ref{th:correctness}), the new polygon and its farthest point  will also be covered by $s_i$ at any successive iteration of $\alg$.
We recall that $\overline{V}(\mathscr{C}^{(k)}_i)=V(\mathscr{C}_i^{(k)} )\setminus
( \cup_{s_j \in  \mathscr{D}_i^{(k)}      } \mathscr{C}_j^{(k)})$. We define 
$\overline{d}_i^{(k)}=d_E(s_i,f(\overline{V}(\mathscr{C}^{(k)}_i)))$ and $d_i^{(k)}=d_E(s_i,f(V(\mathscr{C}^{(k)}_i)))$. As $\overline{V}(\mathscr{C}^{(k)}_i) \subseteq V(\mathscr{C}_i^{(k)} ) \subseteq \mathscr{C}_i^{(k)}$
the following holds:
\begin{equation}
0 \leq \overline{d}^{(k)}_i \leq d^{(k)}_i \leq r_i^{(k)}.
\label{eq:comparison2}
\end{equation}

Since $r_i^{(k)}$ is strictly decreasing and non-negative, when $k \rightarrow \infty$, it converges to a value $R_i\geq 0$.
$\alg$ sets the radius of $s_i$ for the next iteration as: $r_i^{(k+1)}=r_i^{(k)}-\alpha_i^{(k)} \cdot (r^{(k)}_i-\overline{d}^{(k)}_i )$, where  
$\alpha_i^k \in (0,1]$.
It follows that 
$R_i=R_i - \lim_{k \rightarrow \infty}  \alpha_i^k  \cdot (R_i -  \lim_{k \rightarrow \infty}   \overline{d}^{(k)}_i)$.
As $\alpha_i^k > \alpha_\texttt{min}$ is strictly positive and lower than 1, then 
$  \lim_{k \rightarrow \infty}   \overline{d}^{(k)}_i=     R_i$.

The convergence of $\lim_{k \rightarrow \infty}  d^{(k)}_i$ follows, due to Equation \ref{eq:comparison2}, by applying the comparison criterion.
This means that the radius of $s_i$ converges to  the minimum value  to cover the farthest vertex of its polygon, 
which is a {\em boundary farthest configuration}.

If such a boundary  farthest vertex  is strict, then $s_i$ terminates its execution of $\alg$.

Otherwise, the adoption of the {\em serialization scheme for loose farthest vertices} discussed in 
Section \ref{sec:skippable} ensures that  all the sensors with loose vertices will perform their additional radius reduction one at a time. 
After this radius reduction, $s_i$ will never generate again a loose farthest with 
the same neighbors (as this would require an increase in the sensing range of at least one sensor, which is not allowed by $\alg$).
Since there is a finite number of neighbor sensors that can generate a loose farthest with $s_i$, then $s_i$ will eventually reach a strict farthest situation and will exit. \marginpar{NNN: Qui siamo un po' superficiali nella spiegazione del fatto che non si pu\'o avere di nuovo un farthest con gli stessi vicini, secondo voi a un rev. pignolo non gli viene il dubbio che se tutti si riducono si pu\'o verificare di nuovo? Magari se lo disegna e non capisce che sta disegnando un buco di copertura... boh, ditemi voi. }

{\em Convergence in case (2).}
In this case, as the coverage of the polygon is only partial, the sensor  cannot reduce its radius (due to Corollary \ref{co:3.2}) and $\alg$ immediately terminates.

{\em Convergence in case (3).}
Consider $k=0$.
In case (3)  $V(\mathscr{C}^{(0)}_i) \cap \mathscr{C}^{(0)}_i =\emptyset$.  At the successive iterations, the polygon of $s_i$ can only be altered by the radius reductions performed by  the neighbors of $s_i$ and $s_i$ itself. 

As the polygon $V(\mathscr{C}^{(0)}_i)$ is not covered, the polygons of the Voronoi-Laguerre neighbors of $s_i$ are either partially covered or completely uncovered, because they share an edge with $V(\mathscr{C}^{(0)}_i)$. 
A radius reduction of a neighbor with completely uncovered polygon may result in an extension of the polygon of $s_i$ with new uncovered zones. By contrast, 
the neighbors of $s_i$ which partially cover their polygons will not change their radius. 
Therefore,  for any iteration $k > 0$, $V(\mathscr{C}^{(k)}_i) \cap \mathscr{C}^{(k)}_i = \emptyset$, that is a polygon which is initially uncovered will remain uncovered.

Hence, for a sensor $s_i$ being in case (3), $\overline{d}_i^{(k)}=0$, $\forall k \geq 0$.
This implies that $r_i^{(k+1)}=(1-\alpha_i^{(k)})\cdot r_i^{(k)} \leq (1-\alpha_\texttt{min})\cdot r_i^{(k)}$, $\forall k \geq 0$.
Therefore
$\lim_{k \rightarrow \infty} r_i^{(k)} \leq \lim_{k \rightarrow \infty}  (1-\alpha_\texttt{min})^k
\cdot r_i^{(0)}=0$, proving that the sensor $s_i$ converges to a final configuration in which it will be switched off\footnote{Notice that, although $s_i$ knows from the beginning that its decision will be to turn off, it still executes the algorithm iteratively in order not to alter the decision priority established by the energy gain criterion.}.

 \end{proof}
}
 \marginpar{TERMINAZIONE DEI FISSI }
 {\color{black}
 \begin{theorem}\label{le:optimalFixed}
 {\bf (Termination in the case of fixed sensors)}
 Given a set $\mathcal{S}=\mathcal{S}_\texttt{fixed}$ of only fixed, $\alg$ puts to sleep all the redundant sensors in a finite time. \end{theorem}
 \begin{proof}
 At the $k$-th  iteration of $\alg$, every fixed sensor determines whether it is redundant or not.
 If it is not redundant it immediately ends its execution with the decision to stay awake.
 If instead   it is redundant
 it turns itself off with probability $\alpha _i$ (see Algorithm  \ref{alg:selective_activation_FIXED}).
 At every iteration $k$ of the algorithm, there is at least one sensor $s_i$ (namely the one with maximum value of $\Delta{E}_i$) whose value of $\alpha_i^{(k)}$ 
 is equal to 1 and therefore has probability 1 to go to sleep.
 It follows that at each iteration at least one redundant sensor turns itself off (although in practice many sensors go to sleep at each iteration, as shown in Section \ref{sec:results}). Hence, in a finite number of steps all redundant fixed sensors will go to sleep.
 \end{proof}

 }

{\color{black}
\marginpar{CONVERGENZA NEL CASO MIXED }
 \begin{theorem} ({\bf Convergence of $\alg$ in the general scenario})\label{le:convergence}
 Given a set $\mathcal{S}=\mathcal{S}_\texttt{adjustable} \cup \mathcal{S}_\texttt{fixed}$ of both adjustable and fixed   sensors,  under the execution of $\alg$, each sensor converge to a final configuration decision.
 \label{th:conv_mixed}
 \end{theorem}
 
 \begin{proof}
  The convergence of $\alg$
  easily descends from the Theorems~\ref{le:optimalVariable} and~\ref{le:optimalFixed}.
  
 It has to be noted that although the presence of fixed sensors does not alter the convergence property of the adjustable sensors, the opposite is not true.
 In fact, the presence of adjustable sensors in the mix alters the behavior of the fixed sensors as it is no longer guaranteed that at every iteration $k$ there will be a redundant fixed sensor
  that will turn itself off.
  Although it is still true that there will be at least one sensor $s_i^{(k)}$ in $\mathcal{S}$ with $\alpha_i^{(k)}=1$, this sensor may belong to the adjustable class.
  Therefore, the convergence speed of the fixed class is slowed down by the presence of the adjustable sensors\footnote{This is because we want the two classes of sensors to reduce their radius in parallel without
favoring a given class. If, due to a particular operative setting, one of the
two classes should have a higher priority in making configuration decisions, this can be handled
by redefining accordingly the priority parameter $\alpha_i^{(k)}$.} .
 For this reason this theorem only affirms the convergence and not the termination of $\alg$ in the mixed scenario, as in the case of only adjustable sensors. 
  \end{proof}
  }
  {\color{black} 

%
%
{\color{black}
\marginpar{Su questo pezzo aveva messo mano Tom, lo lascio nei commenti e metto quello indicato da te.
Ho modificato qualche parola e ho aggiunto il commento che lasciavi in sospeso.
Dettagli dopo non ne troveranno molti di pi\'u dato che Tom ci ha fatto eliminare gli exp sul $K$.}


Theorem \ref{le:convergence} states the convergence of $\alg$ in the mixed scenario.
The question
is how to ensure that convergence does not take too long: The
adjustable sensors might 
reduce their radius of
an infinitesimal step at each iteration.
In order to ensure the theoretical termination of the algorithm in a finite number of 
steps we can set an upper limit $K$ on the number of iterations ({\em faster termination
condition}).  Despite convergence might theoretically take quite long
time we have observed that no more than 20 iterations are sufficient to achieve termination of the $95\%$ of sensors.
Setting a value of $K$ as low as $20$ has a negligible impact on the performance of $\alg$, but has the advantage to ensure a very fast termination of the algorithm execution.

}

 \marginpar{Rimosso
 il corollario che dice che se si fissa il nr di iterazioni allora il nr di iterazioni � finito :-)}

The following Lemma \ref{lemma:farthestWithoutEpsilon} analyzes the property of the cover set obtained after the execution of $\alg$ 
focusing in particular on the polygons generated by the adjustable sensors.

\reversemarginpar
\marginpar{NELLA CONFIG FINALE GLI ADJUSTABLE SONO TUTTI AL BOUNDARY FARTHEST}
\normalmarginpar
\begin{lemma}
\label{lemma:farthestWithoutEpsilon}
{\bf(Properties of the cover set)} Consider a mixed set  of adjustable and fixed sensors  $\mathcal{S}=\mathcal{S}_\texttt{adjustable}\cup\mathcal{S}_\texttt{fixed}$. Consider the cover set $S_\texttt{SARA}$ obtained after the execution of $\alg$ on $\mathcal{S}$. 
If $s_i \in S_\texttt{SARA} \cap \mathcal{S}_\texttt{adjustable}$ either $s_i$ partially covers its polygon $V({\mathscr C}_i)$, or  its farthest vertex $f(V({\mathscr C}_i))$ 
is a strict boundary farthest vertex. 
\end{lemma}

\reversemarginpar
\begin{proof}

Let $s_i$ exit $\alg$ at time $T_{s_i}$ with its radius set to $r_{i}^{k} > 0$\footnote{Notice that the case $r_i^k = 0$ is exluded because $s_i$ belongs to the cover set $S_\texttt{SARA}$}.
According to Algorithm
\ref{alg:selective_activation_ACTIVE}
$s_i$ terminated $\alg$ execution either because its polygon is not completely covered
or because it has reached a strict boundary farthest configuration.
We now show that changes in the sensing coverage of other nodes $s_j$
which occur at a time $T>T_{s_i}$ cannot change this property. As this is obvious for sensors which partially cover their polygons, let us consider the case of $s_i$ completely covering its polygon.

Two types of events can
occur after $T_{s_i}$ which affect sensor $s_i$ responsibility region:
1) other adjustable sensors $s_j$ reduce their radius,
2) fixed or adjustable sensors are turned off.
Both these events may result in an increase of sensor $s_i$ responsibility region.
However, since $s_i$ radius cannot change ($s_i$ has exited),
since the reduction of other nodes radius preserves coverage
(Theorem \ref{th:correctness})
and since  if a point P is covered it is covered by the node to which responsibility
region it belongs (theorem \ref{th:lag_coverage})
it derives that $s_i$ responsibility region stays within the circle centered in $s_i$
and with radius equal to $r_{i}^{k}$.
Therefore, each boundary farthest point at time $T_{s_i}$ is still a boundary farthest
at the end of $\alg$ execution.
\end{proof}

  According to $\alg$, each sensor pursues an individual utility that is 
  to reduce its power consumption and at the same time to do its best  to cover the AoI.
  In terms of this utility function, 
  the cover set $S_\texttt{SARA}$ obtained by $\alg$ starting from $\mathcal{S}$, is Pareto optimal.
  In fact, it is not possible to increase the utility of a single sensor (i.e., by reducing the sensing range of an adjustable sensor or turning off a fixed one) without decreasing the utility (i.e., increasing the sensing range of an adjustable sensor or turning on a fixed one that was previously sleeping) of at least another device in the network.

 \marginpar{PARETO OPTIMALITY}
 \begin{theorem}
 \label{th:optimalWithoutEpsilon}
 {\bf (Pareto optimality)} Given a set $\mathcal{S}=\mathcal{S}_\texttt{adjustable} \cup \mathcal{S}_\texttt{fixed}$ of  sensors, after the 
 execution of $\alg$  (without the faster termination condition), the produced cover set  $S_\texttt{SARA}$ is Pareto optimal.
 \end{theorem}
 
 \begin{proof}
  In order to prove the Pareto optimality of $\alg$ we need to show that there is no action that could improve the utility of a single sensor, i.e. a sensing radius reduction or the deactivation of a device, without reducing the coverage achieved by $S_\texttt{SARA}$.
  
  This property is true for fixed sensors, since all 
  redundant fixed sensors will eventually turn themselves off according to the back-off scheme provided by $\alg$. This trivially derives from Theorems \ref{le:optimalFixed} and \ref{th:conv_mixed}.
  
  In the case of adjustable sensors, consider $s_i \in \mathcal{S}_\texttt{adjustable} $,  Theorem \ref{le:optimalVariable} states that
 under the execution of $\alg$ $s_i$ will eventually reach a final configuration decision, while 
 Lemma \ref{lemma:farthestWithoutEpsilon}
gives a characterization of the final solution, affirming that 
if $s_i$ completely covers its polygon, $s_i$ is in a strict boundary farthest vertex configuration
whereas if $s_i$ covers its polygon only partially, Corollary \ref{co:3.2} proves that in this case $s_i$ cannot reduce its radius without affecting coverage. 

%
%
%
%
%
%

 \end{proof}

 Pareto optimality is a {\em necessary condition for global optimality}.
 Unfortunately, the Pareto optimality of the cover set  does not have implications in terms of quality of the solution to the lifetime problem, as there are infinite Pareto optimal solutions.
 Nevertheless,  by adopting an energy-aware policy, $\alg$ is able to choose a cover set 
 among all the possible Pareto-optimal ones, which reduces the energy consumption of the network and  prolongs its lifetime,
 as experimentally shown in Section \ref{sec:results}.


\section{Two recently proposed selective activation and radius adaptation algorithms}
\label{sec:two_approaches}

To the best of our knowledge there is no prior work in the literature that addresses the problem of selective activation and  sensing radius adaptation in a general applicative scenario combining fixed sensors and sensors
endowed with variable sensing capabilities.
Moreover, previous works rarely consider device heterogeneity. For these reasons,  we compare $\alg$ to the Distributed Lifetime Maximization ($\dlm$) algorithm \cite{Kaskebar2009} which is designed to work with 
fixed radius sensor and to the Variable Radii Connected Sensor Cover ($\gup$) algorithm \cite{Zou2009}
 which is designed to work only with devices that can adjust their sensing radius.
The choice of these two algorithms is motivated by the performance analysis carried out by the same authors which shows that $\dlm$ and $\gup$ achieve better performance to previous schemes proposed in the same class.

In this section we give a short description of $\dlm$ and $\gup$ and of our extensions to adapt them for a 
general scenario.
We also discuss why they do not provide Pareto optimal solutions.

%


$\dlm$ 
addresses the problem of activating a subset of sensors so that each point of the AoI is
monitored by at least $k$ sensors
\footnote{For the sake of simplicity, in this paper we do not address the problem of $k$-coverage. Hence in all our experiments, detailed in Section \ref{sec:results}, we assume that all the algorithms  work with $k=1$.}.
$\dlm$ considers the case of heterogeneous sensors with fixed sensing radii.
The authors call  {\em intersection point} any point where two 
sensing circles intersect with each other and observe that 
if each intersection point is $k$-covered, then the whole AoI is $k$-covered. 
$\dlm$ is a round based algorithm. At each round, maximum coverage is obtained by iteratively waking up sensors according to an ordered list of nodes that are in radio proximity.
 The list is sorted on the basis of the energy consumed by the nodes and of the number of intersection points that they can cover. Such a list provides the priority order for the iterative waking up of the sensors in a neighborhood. At each iteration, the sensors whose sensing range is already $k$-covered by other already awake  sensors are removed from the list (they will not wake up themselves).
We refer to \cite{Kaskebar2009} for the details of the algorithm.

We extend  $\dlm$ to the case of sensors with adjustable sensing radii 
by considering the devices with variable radii as if they were fixed.
This means that each sensor, independently of the class to which it belongs, will either
wake up (i.e., operate at maximum transmission radius) or go to sleep.
As $\dlm$ is not designed to deal with variable radii devices, this variant is introduced only to show that to apply $\dlm$ to a more general setting requires non trivial changes. 

\medskip

$\gup$ explicitly addresses
the problem of $k$-covering the AoI with sensors with adjustable radii (both transmission and sensing radii).

$\gup$ makes use of Voronoi diagrams to determine which sensors are completely redundant. It then reduces  the radius of each sensor to the minimum necessary to cover the farthest point of its Voronoi polygon. For each redundant sensor $s$, $\gup$ calculates the energy benefit obtained by putting it to sleep. This benefit is compared to the additional energy expenditure that the neighbors of $s$  would incur to enlarge their radius with respect to their minimum setting (i.e. the one needed to cover their Voronoi polygon) so as to cover the Voronoi polygon of $s$ on its behalf.
We refer the reader to \cite{Zou2009} for more details on $\gup$.

We extend the use of $\gup$ to the case of sensors with fixed radii. In the case of fixed sensors $\gup$ 
only operates the waking up/putting to sleep decisions, while the rules to reduce sensor radius
are disabled. 
The purpose of this variant is to show how trivial extensions of $\gup$  
perform in a more general scenario than the one for which it is designed.

 
 \begin{figure}
\begin{center}
\begin{tabular}{ccc}
{\includegraphics[width=0.3\textwidth]{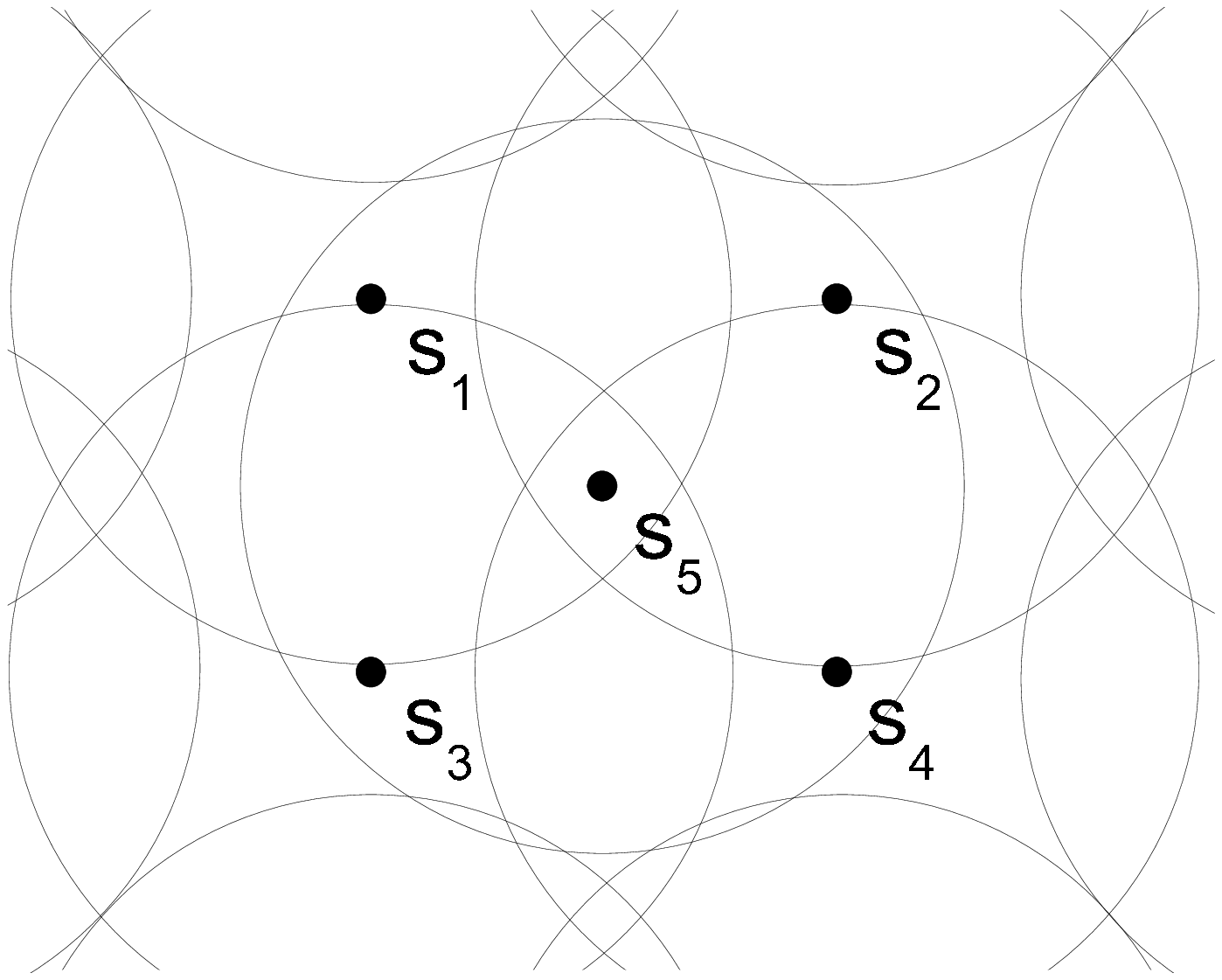}} & 
{\includegraphics[width=0.3\textwidth]{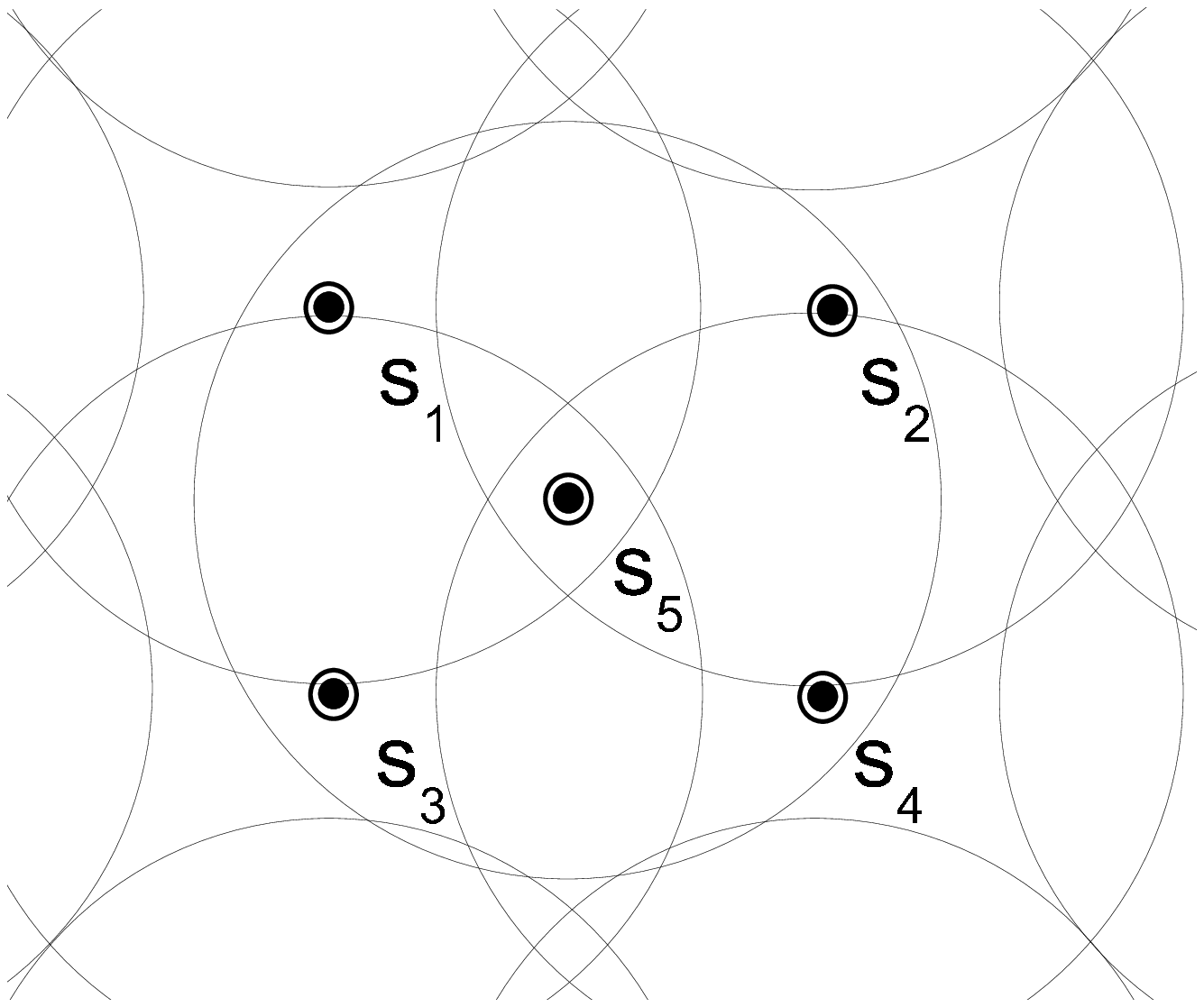}} & 
{\includegraphics[width=0.3\textwidth]{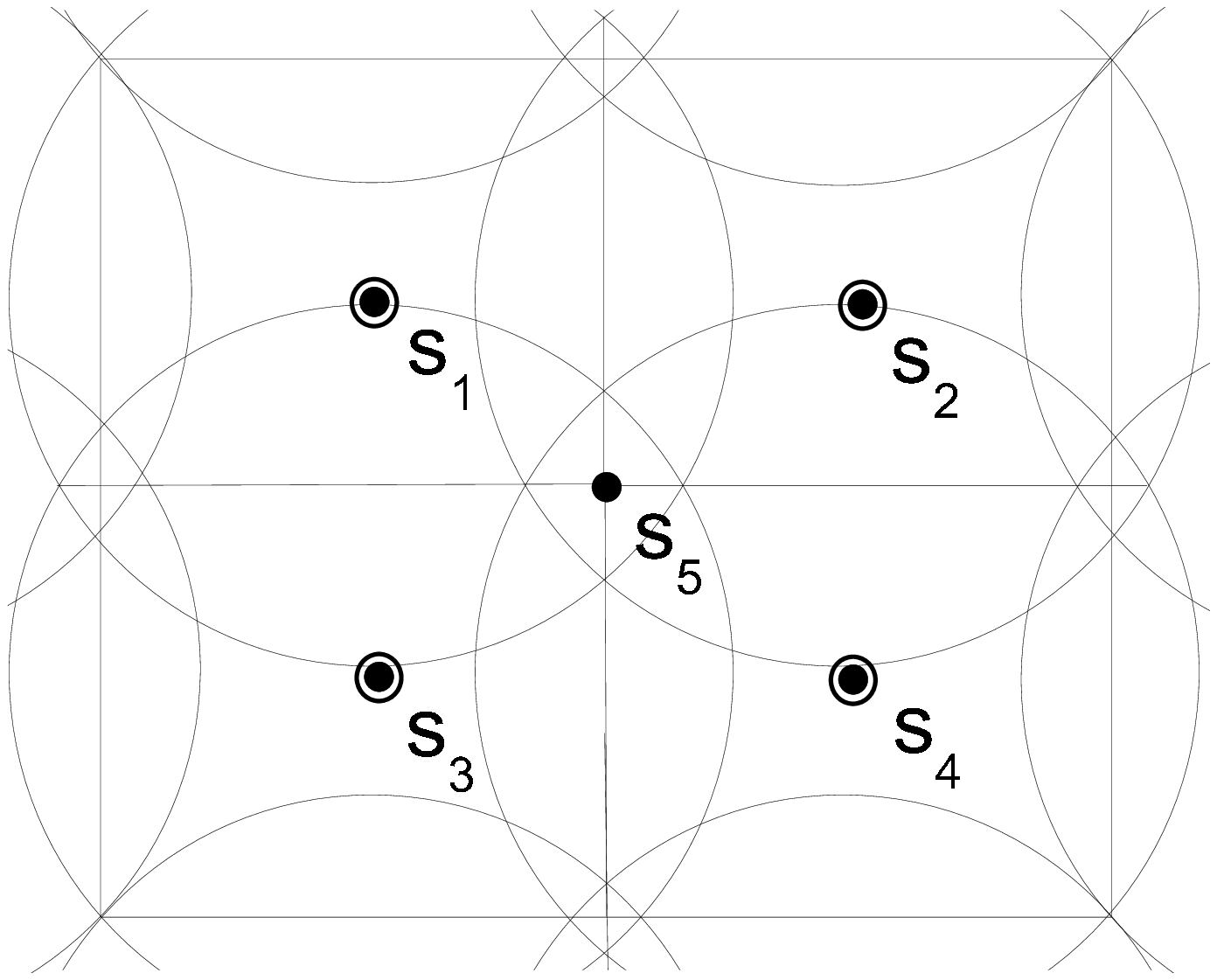}} \\
(a) & (b) & (c) 
\end{tabular}
\end{center}
 \caption{About Pareto optimality. Initial configuration (a). Selective activation with $\dlm$ (b) and $\alg$ (c). The nodes with double circle are awake, while the other ones are sleeping.
}
\label{fig:redundancies_2}
 \end{figure}

Unlike our approach, both $\dlm$ and $\gup$ do not meet the necessary condition for optimality discussed in 
Section \ref{sec:properties}.
This is explained in Figures \ref{fig:redundancies_2} and \ref{fig:redundancies_3}.

Figure \ref{fig:redundancies_2}(a) represents an initial configuration with fixed sensors.
Observe that sensors $s_1$, $s_2$, $s_3$ and $s_4$
must be awake to ensure a  complete coverage of the AoI,
  as they cover portions of the AoI that cannot
be covered by any other sensor in the network.
According to $\dlm$, if the energy available to sensor $s_5$ is sufficiently high,
 $s_5$ can be the first sensor
to be woken up in its neighborhood. Once awake, it stays awake
despite the waking up of the other four sensors makes $s_5$ unnecessary (see
Figure \ref{fig:redundancies_2}(b)).

Under the same initial setting $\alg$ would not activate $s_5$,
as the backoff policy ensures the sleeping of all redundant sensors. This is shown in Figure \ref{fig:redundancies_2}(c).


\begin{figure}
\begin{center}
\begin{tabular}{ccc}
{\includegraphics[width=0.3\textwidth]{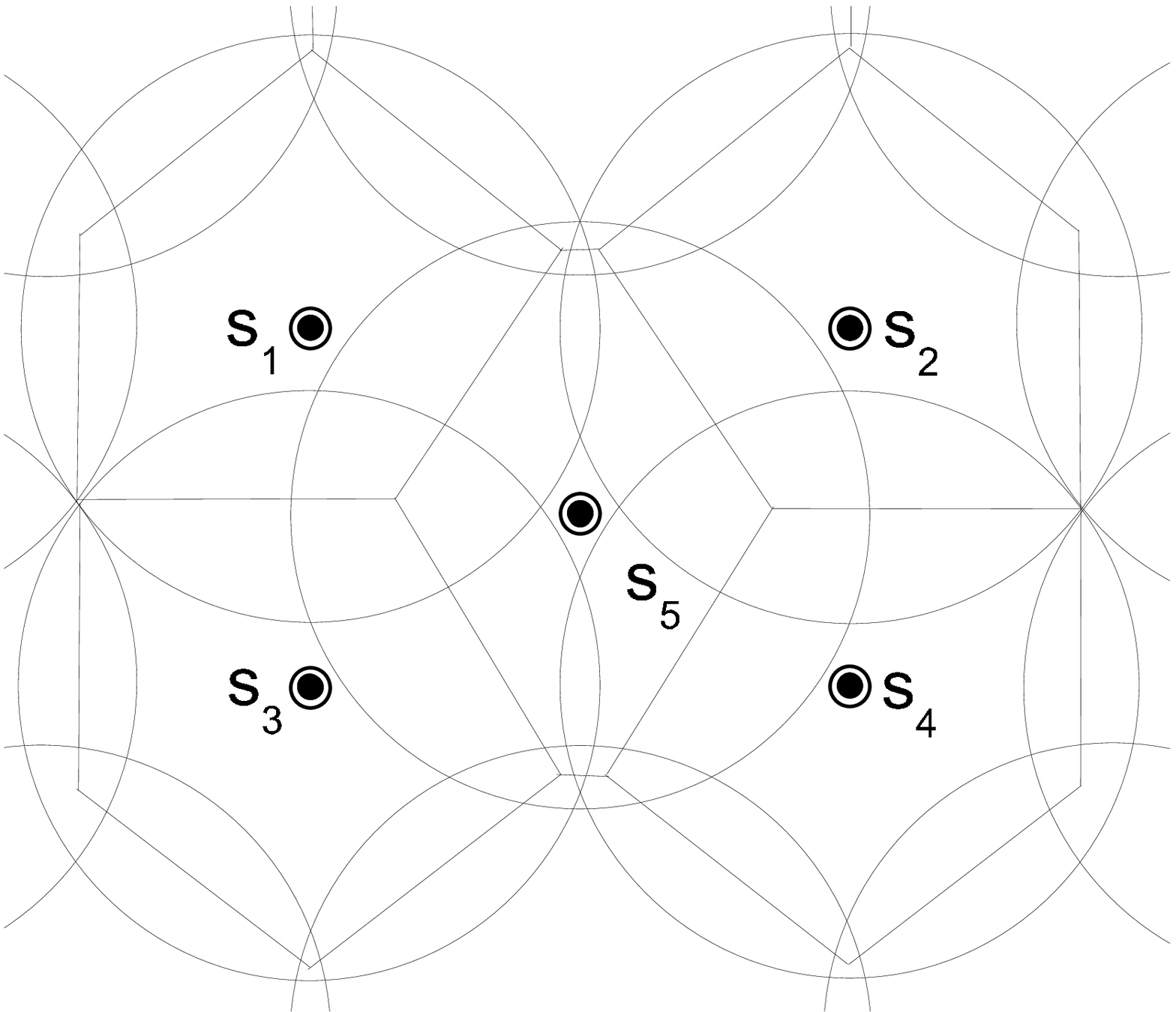}}  & 
{\includegraphics[width=0.3\textwidth]{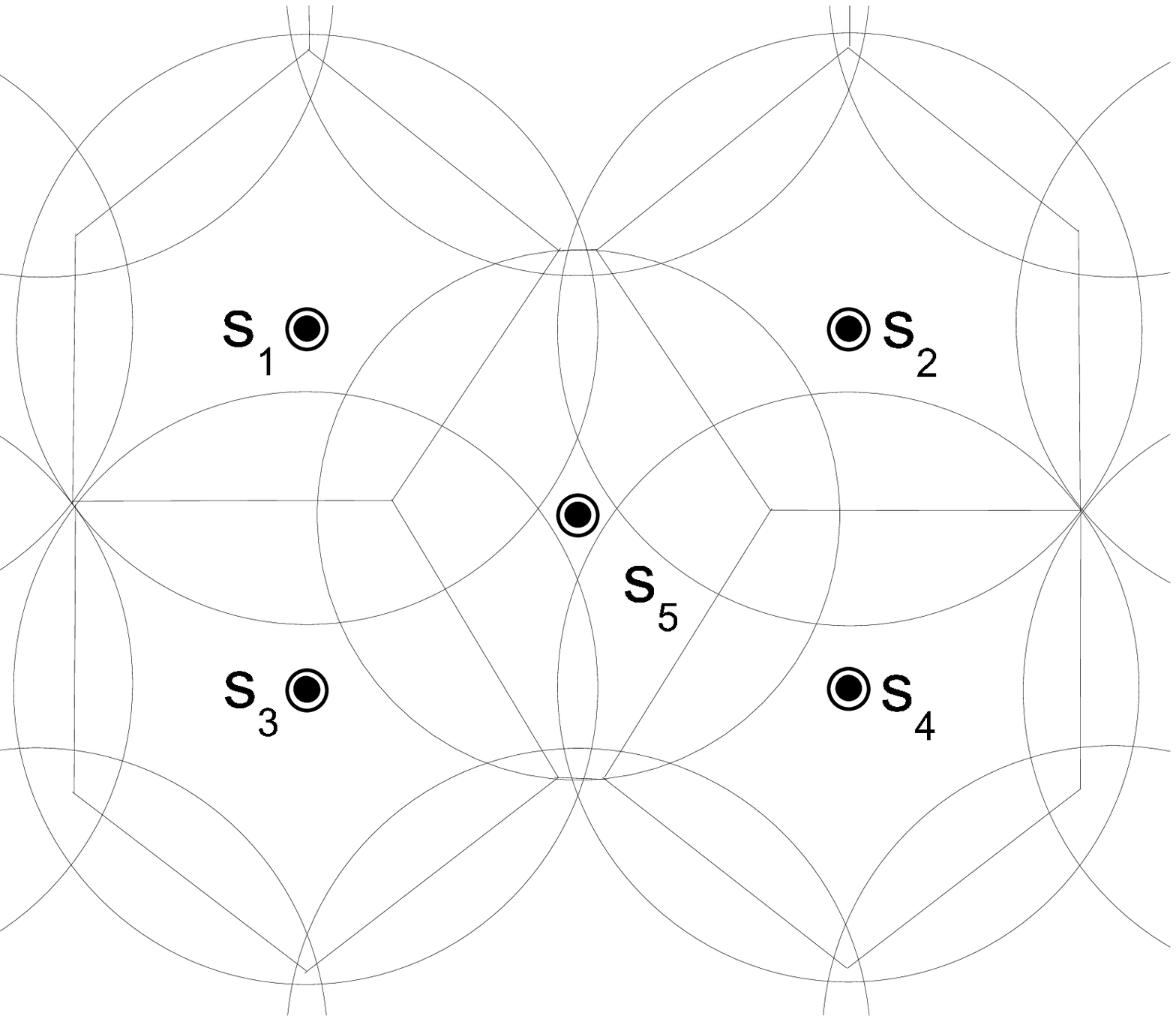}} &
{\includegraphics[width=0.3\textwidth]{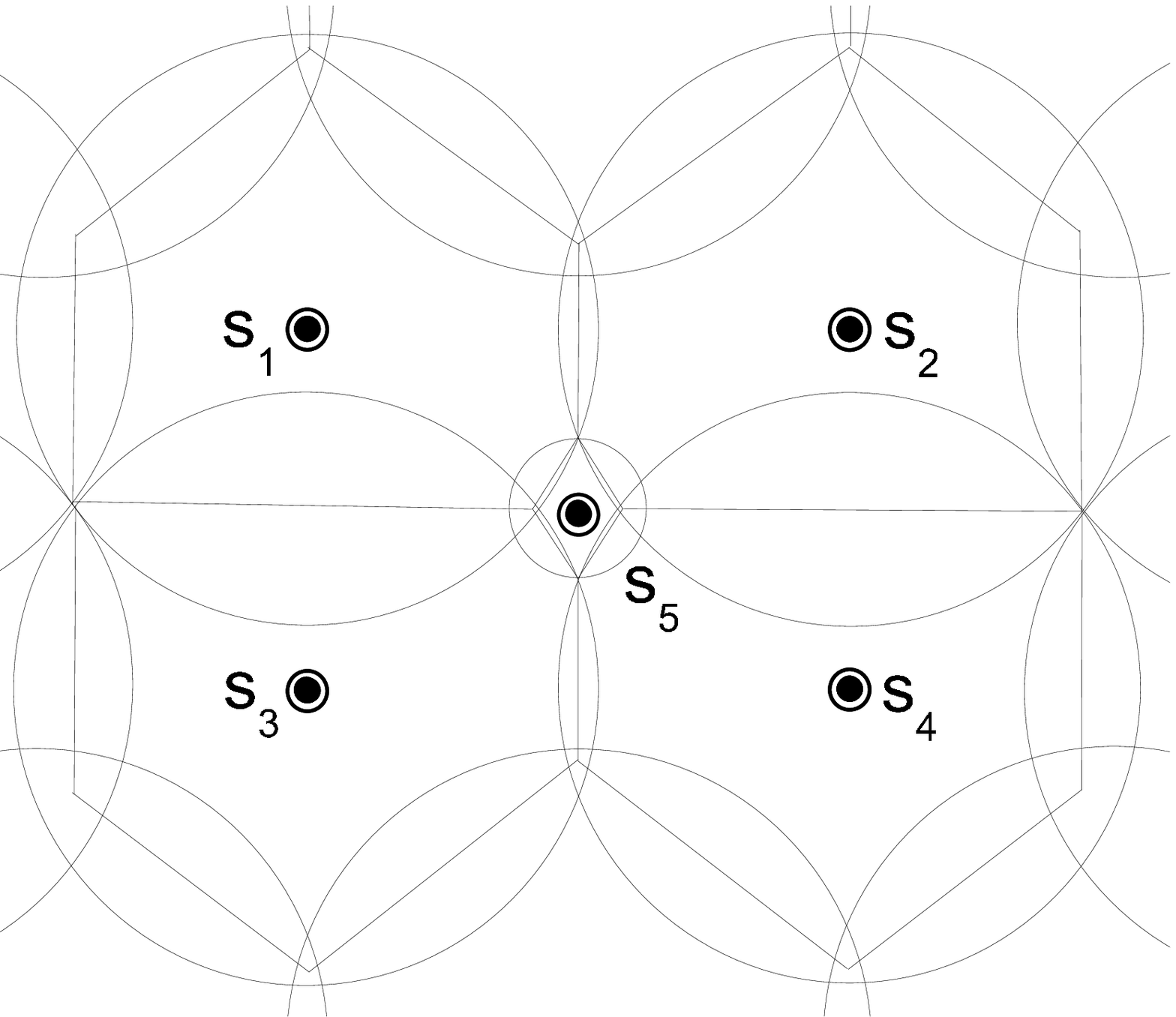}} \\
(a) & (b) & (c) 
\end{tabular}
\end{center}
 \caption{About Pareto optimality. Initial configuration (a). Selective activation with $\gup$ (b) and $\alg$ (c).
}
\label{fig:redundancies_3}
 \end{figure}

Figure \ref{fig:redundancies_3} displays a scenario with adjustable sensors.
Figure  \ref{fig:redundancies_3}(a) shows the initial configuration where all sensors are awake and work at their maximum radius. The figure also highlights the Voronoi
diagram of the considered sensors. In this example all sensors ($s_1$, $s_2$, $s_3$, $s_4$ and $s_5$) are 
needed to achieve full coverege.
Sensors  $s_1$, $s_2$, $s_3$ and $s_4$ cannot reduce their radius as their uniquely covered zone reaches 
the boundary of their sensing circle. Sensor $s_5$, instead, can significanlty reduce its radius without 
affecting coverage.

According to $\gup$ each sensor sets its radius to the distance from it to the farthest vertex of its Voronoi polygon. 
Therefore, $s_5$ reduces its radius as shown in Figure \ref{fig:redundancies_3}(b). Since no sensor  can be 
put to sleep, this is the final configuration achieved by $\gup$. Nevertheless, sensor $s_5$ can still significantly 
reduce its radius. By iteratively adjusting the radius of $s_5$,  $\alg$ reaches a Pareto optimal configuration, 
where the radius of $s_5$ is set to the minimum value that does not leave a coverage hole, as shown in Figure 
\ref{fig:redundancies_3}(c).

We conclude this subsection by underlying that if $\dlm$ and $\gup$ are not properly extended as discussed above,
$\gup$ cannot be used in the case of non adjustable radii and, vice-versa, $\dlm$ cannot be applied to the 
case of variable radii.
Our algorithm, instead, works in both the operative settings. Moreover, our algorithm is also
able to work in a mixed scenario  characterized by both sensors with adjustable and fixed radii, even in the
presence of heterogeneous devices, showing an impressive versatility.
In the performance evaluation section we will also show that $\alg$ achieves significant performance
improvements over the other two schemes in all operative settings.

We summarize the features of the three schemes in Table
\ref{tab:appli}.

\begin{table}[ht]
\centering
\begin{tabular}{|c|c|c|c|c|c|c|} \hline
& \multicolumn{2}{|c|}{Fixed type} & \multicolumn{2}{|c|}{Adjustable type} & \multicolumn{2}{|c|}{Both types}\\ \hline
& Hom. & Het. & Hom. & Het. & Hom. & Het.\\ \hline
$\dlm$ & Y & Y & N & N & N & N\\ \hline
$\gup$ & N & N & Y & N & N & N\\ \hline
$\alg$ & Y & Y & Y & Y & Y & Y\\ \hline
\end{tabular}
\caption{Scenarios where the considered algorithms are applicable.}
\label{tab:appli}
\end{table}

To give a fair performance comparison, in Section \ref{sec:results}
we  compare $\alg$ to $\dlm$ and $\gup$ in their restrictive operative settings and then we  extend their 
use to the   general applicative scenario where devices belong to both the two classes of sensors with fixed and 
adjustable radii and are heterogeneous in their
sensing capabilities.


\section{Experimental results} \label{sec:results}

\subsection{Experimental setting}

In all the experiments we use the following setting. The AoI is a square shaped region of 80m $\times$  80m. 
 We adopt the Telos \cite{Telos} communication cost model. Concerning the sensing model of sensors with adjustable radius
we consider the cost of six Maxbotix sonar devices \cite{sonar} with different orientations, working at 2Hz. 
We adopted the cubic law of energy cost ($c=3$ in Equation \ref{eq:e_active_sensing}) 
with respect to the sensing radius.

According to these models, each sensor has a transmission range of 30m. The battery capacity is 1840 mAh and  sensors are 
endowed with an initial energy that is uniformly distributed in the interval (0, 1840mAh].
The length of the operative time interval between two successive executions of the algorithm $\alg$ and $\dlm$ is set to 
24h which is equal to 1.5\% of the total time a sensor can remain awake. Notice that the algorithm $\gup$, as defined in 
\cite{Zou2009}, reconfigures the network every time
a sensor has exhausted its available energy.

Regarding the setting of the sensing radius, it varies from one applicative scenario to  the other.
We consider all the applicative scenarios described in Table \ref{tab:appli}.

The algorithms were implemented by using the Wireless module of the OPNET modeler software \cite{opnet}.

\subsection{Choice of the reduction criterion } 
\label{sec:alpha_evaluation}

Before we give the comparative performance evaluation between $\alg$ and other previous works, we show an extract of the many experiments that motivated our choice in the formulation of the priority decision parameter $\alpha$ described in Subsection \ref{sec:alpha}.

Notice that all the properties of $\alg$ that we demonstrated in Section \ref{sec:properties} hold no matter which is the formulation of the parameter $\alpha_i^{(k)}$.

In particular, we proved that, independently of the particular choice for the setting of $\alpha_i^{(k)}$, the algorithm $\alg$ guarantees that the solution will be Pareto-optimal,
therefore no sensor will be able to decrease its energy consumption by turning itself off or reducing its sensing radius, without requiring other sensors to increase their energy expenditure to compensate the coverage loss deriving from the decision of the first sensor.

The order in which the sensors operate their decisions is determined by the particular choice for the formulation of the decision priority $\alpha_i^{(k)}$ and has a direct impact on the 
solution, i.e. different (all Pareto-optimal) solutions are obtained executing the algorithm by giving sensors different priorities.
Nevertheless, it is clear that the network lifetime of different Pareto optimal solutions can vary significantly.

Since the setting of  $\alpha_i^{(k)}$ influences the policy decisions at a local level only, it is not completely intuitive to determine the effects of such local decisions on
a global performance metric such as the network lifetime.
Therefore we considered several possible formulation of the decision priority parameter.

First we can consider the {\em residual energy} of the devices, and gave higher priority to sensors with 
lower residual energy in making turning-off or radius reduction decisions, therefore 
$\alpha_i^{(k) \texttt{residual\_energy} }=  \frac{\max_{j \in \mathscr{L}^{(k)}}(s_i) E^{(n)}_\texttt{available}(s_j)-E^{(n)}_\texttt{available}(s_i)}{ \max_{j \in \mathscr{L}^{(k)}}(s_i) E^{(n)}_\texttt{available}(s_j)-\min_{j \in \mathscr{L}^{(k)}}(s_i) E^{(n)}_\texttt{available}(s_j)}$.

Nevertheless this criterion revealed that by turning off the sensors with lower energy
(or significantly reducing their responsibility region) would cause other sensors (those with large residual energy) to consume an arbitrary large amount of energy, possibly making the alive sensors end up with much lower energy than the turned off sensor would have had
if it were kept awake.
This is the reason why this criterion performs worse than others as we show in Figure \ref{fig:alpha}.

Then, we can consider  a formulation of the parameter $\alpha_i^{(k)}$ which results in a priority setting based on the {\em expected residual lifetime} of individual sensors.

The expected residual lifetime (in number of equally sized operative time intervals) $\hat{L}_i(k)$ of the sensor $s_i$ at the $k$-th iteration of the algorithm  is calculated as the ratio between the
currently available energy $E_\texttt{available}(s_i)$ and the energy consumption per operative time interval with the currently calculated radius, i.e. $\hat{L}_i(k)=\frac{E_\texttt{available}(s_i)}{E_\texttt{active\_sensing}(r_i^{(k)})}$.

The residual lifetime criterion consists therefore in setting the value of $\alpha_i^{(k)}$ as 
$$\alpha_i^{(k) \texttt{residual\_lifetime}}=\frac{L^{(k)}_\texttt{max}(n)-\hat{L}^{(k)}_i(n)}{L^{(k)}_\texttt{max}(n)-L^{(k)}_\texttt{min}(n)}$$

Although this setting of $\alpha_i^{(k)}$ is superior to the previous ones in all the considered scenarios, it still tends to favor the improvement of the lifetime of single sensors with respect to the utility of the global network.
 In particular there are still some situations in which 
the algorithm gives higher priority to turning off some sensors with smaller residual lifetime at the expense of sensors with larger residual lifetime that in this way are forced to work longer, spending more energy than all the others in the neighborhood.

We experimented other formulations of $\alpha_i^{(k)}$ that we do not discuss in this paper for the sake of brevity. 
We experimentally obtained the best results by setting $\alpha_i^{(k)}$ as we described in Subsection \ref{sec:alpha}, giving higher priority to the decisions that 
lead to a better energy gain.

In Figures \ref{fig:alpha}(a) and (b) we considered 
a scenario with 900 homogeneous sensors with adjustable radius, whose sensing range was allowed to vary in the interval $[2m, 6m]$.
Figure 
\ref{fig:alpha} (a) shows a comparison of the residual energy 
obtained by the three criteria, while Figure \ref{fig:alpha}(b) shows how
the criterion that we chose guarantees a longer lifetime than the other ones.

\subsection{Impact of the faster termination condition}

In this subsection we analyze the impact of the {\em faster termination condition} introduced in Section \ref{sec:properties}, namely of the configuration of the maximum number of iterations $K$ that $\alg$ is allowed to execute at the beginning of each operative time interval.
 We recall that the algorithm $\alg$ is guaranteed to converge, but theoretically it may do so in an infinite number of smaller and smaller steps. Although in the experiments we never encountered a scenario
where $\alg$ did not terminate in a finite number of steps, we introduced the condition for faster termination, by imposing an upper bound $K$ on 
the number of algorithm iterations at each operative time interval.

The experiments shown in Figure  \ref{fig:sensitivity_K} (a-d)
are made in a scenario with 900 sensors with adjustable sensing radii ranging from 2m to 6m.
These experiments highlight that even by setting $K$ to a small value
(e.g. 20) the algorithm $\alg$ shows the same performance of the unbounded case in terms of active, sleeping, and dead nodes (Figure \ref{fig:sensitivity_K} (a-b-c)) and of network coverage (Figure \ref{fig:sensitivity_K}(d)).


\begin{figure}
\begin{minipage}[b]{0.32\textwidth} 
\centering
\begin{tabular}{c}
\hspace{-0.5cm}\includegraphics[width = 0.7\textwidth,angle=-90]{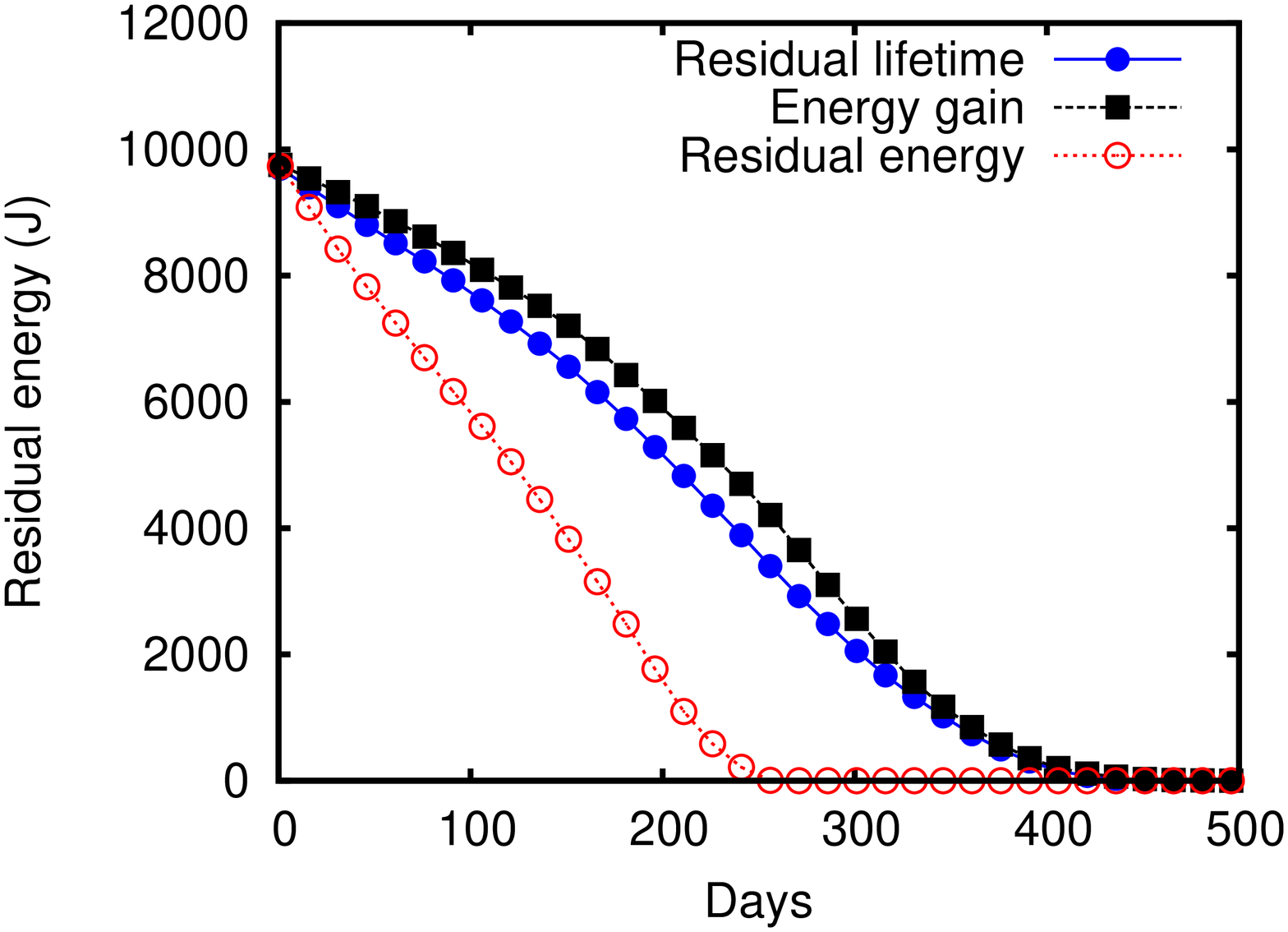}
\\
\hspace{-0.5cm}{\footnotesize{(a)}}\\
\hspace{-0.5cm}\includegraphics[width = 0.7\textwidth, angle=-90]{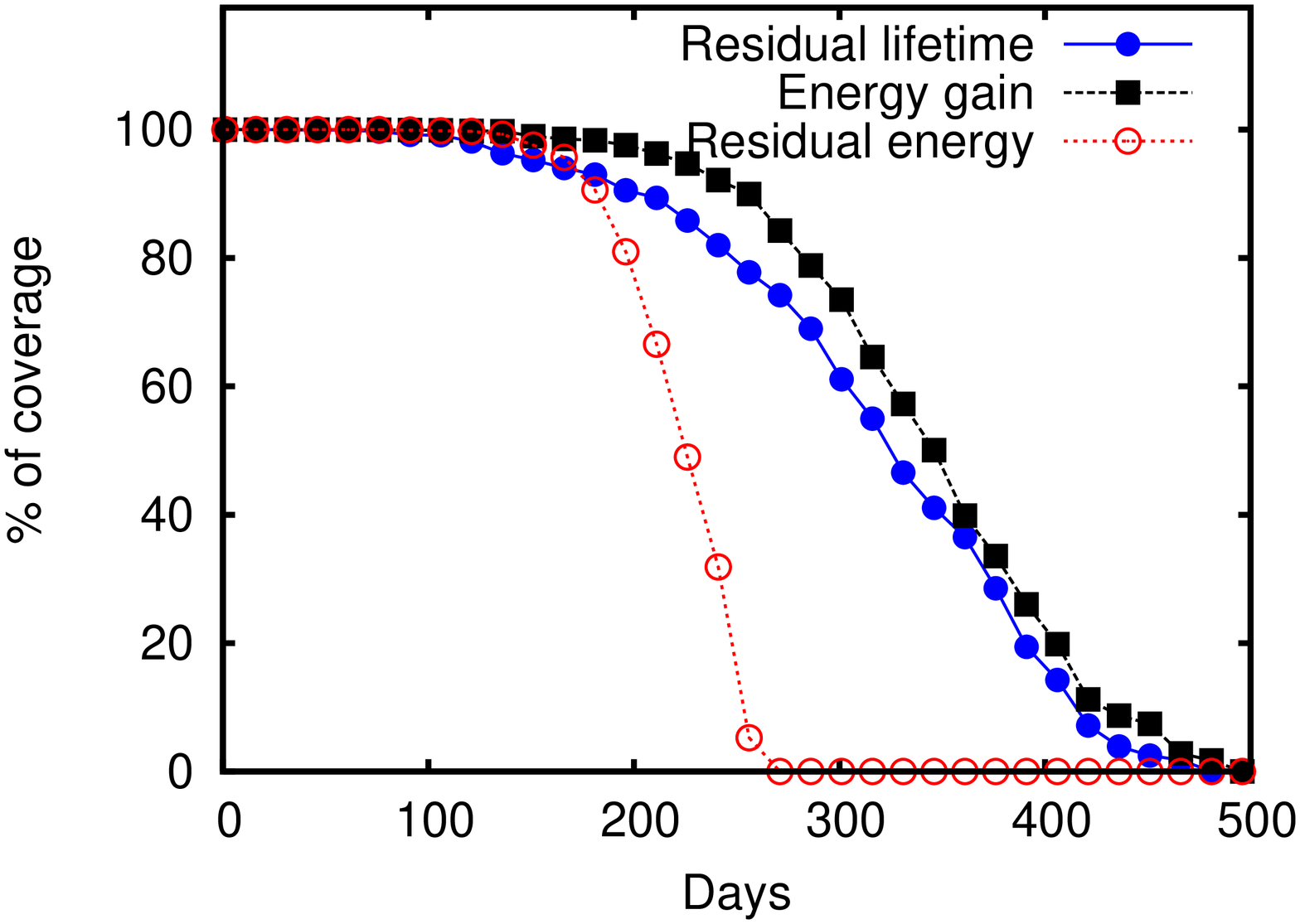}\\
\hspace{-0.5cm}{\footnotesize{(b)}}
\end{tabular}
 \caption{Average residual energy (a) and coverage of $\alg$ under different formulations of the decision priority parameter. }
\label{fig:alpha}
\end{minipage}
\hspace{0.5cm} 
\begin{minipage}[b]{0.62\textwidth}
\begin{tabular}{cc}
\hspace{-0.5cm}\includegraphics[width = 0.36\textwidth,angle=-90]{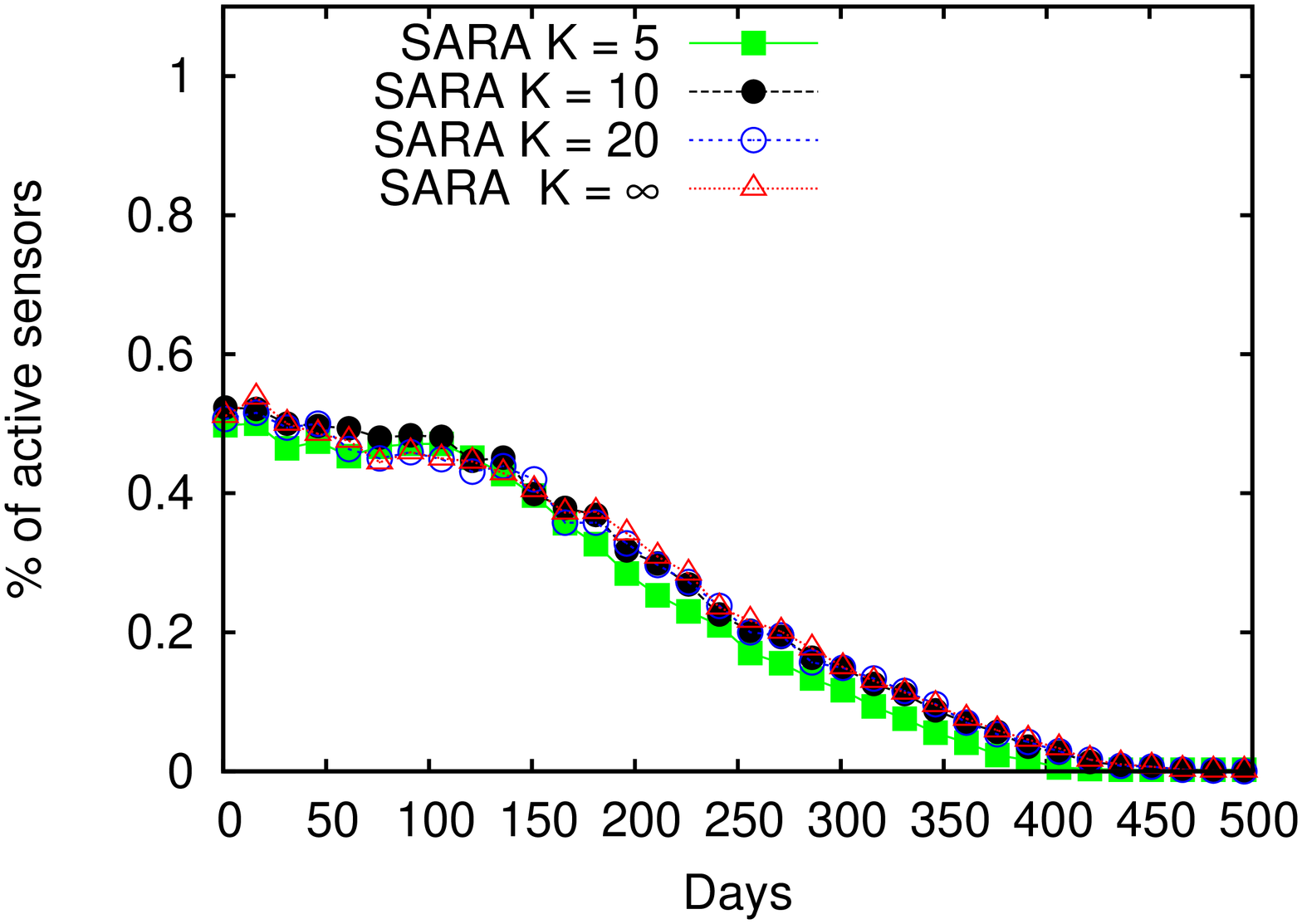}
&
\hspace{-0.5cm}\includegraphics[width = 0.36\textwidth,angle=-90]{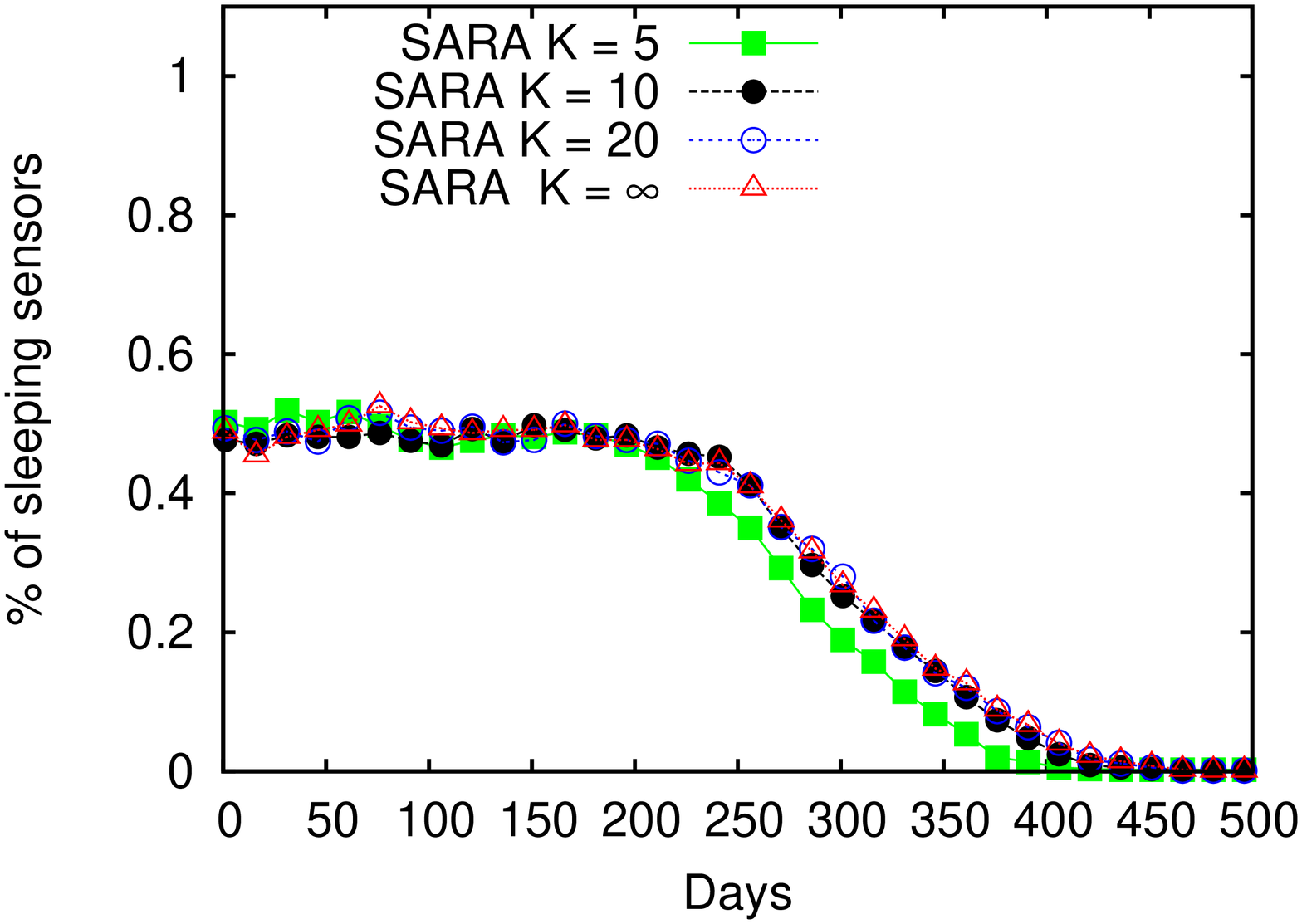}
\\
\hspace{-0.5cm}{\footnotesize{(a)}}&{\footnotesize{(b)}}\\
\hspace{-0.5cm}\includegraphics[width = 0.36\textwidth,angle=-90]{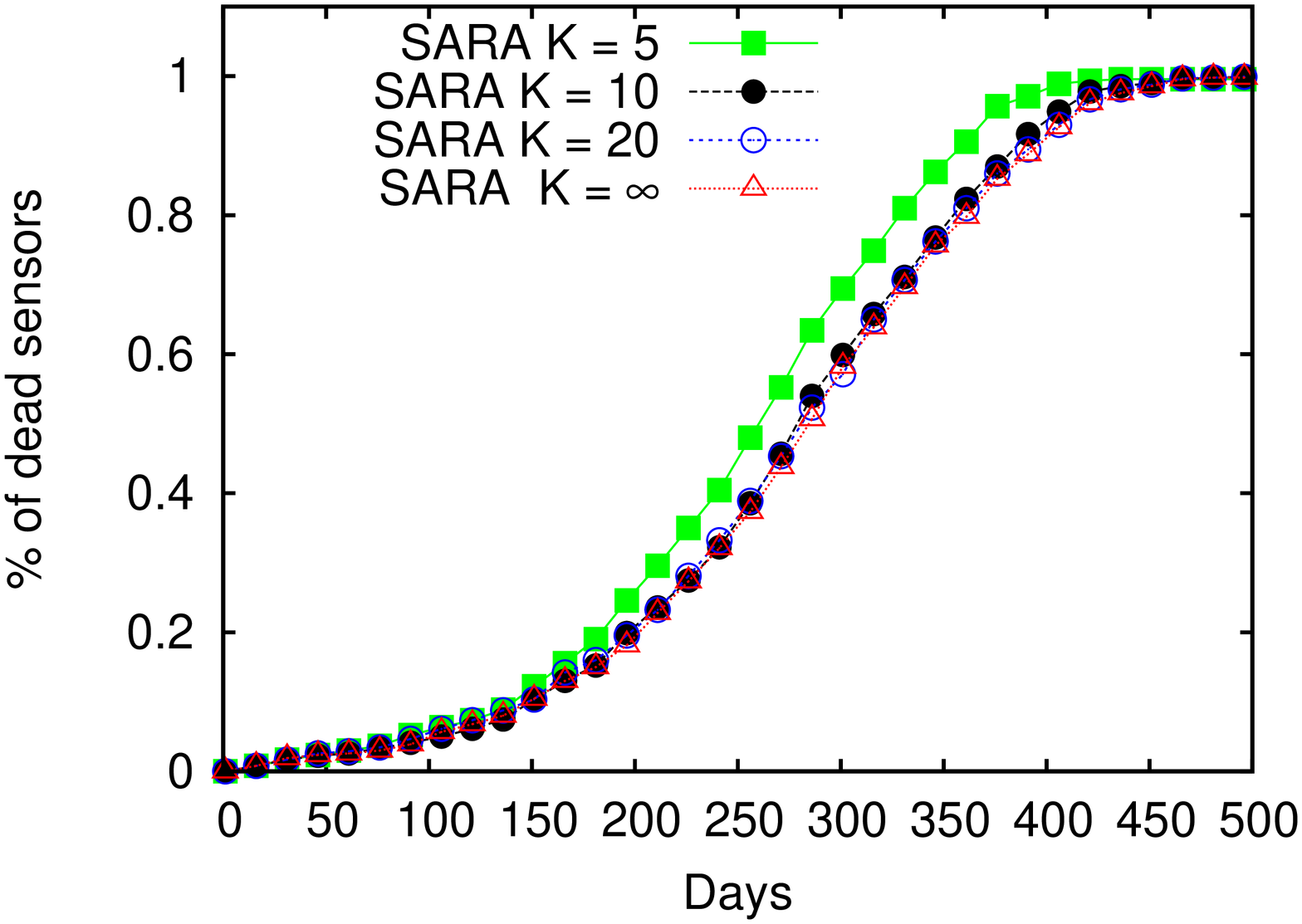}
&
\hspace{-0.5cm}\includegraphics[width = 0.36\textwidth, angle=-90]{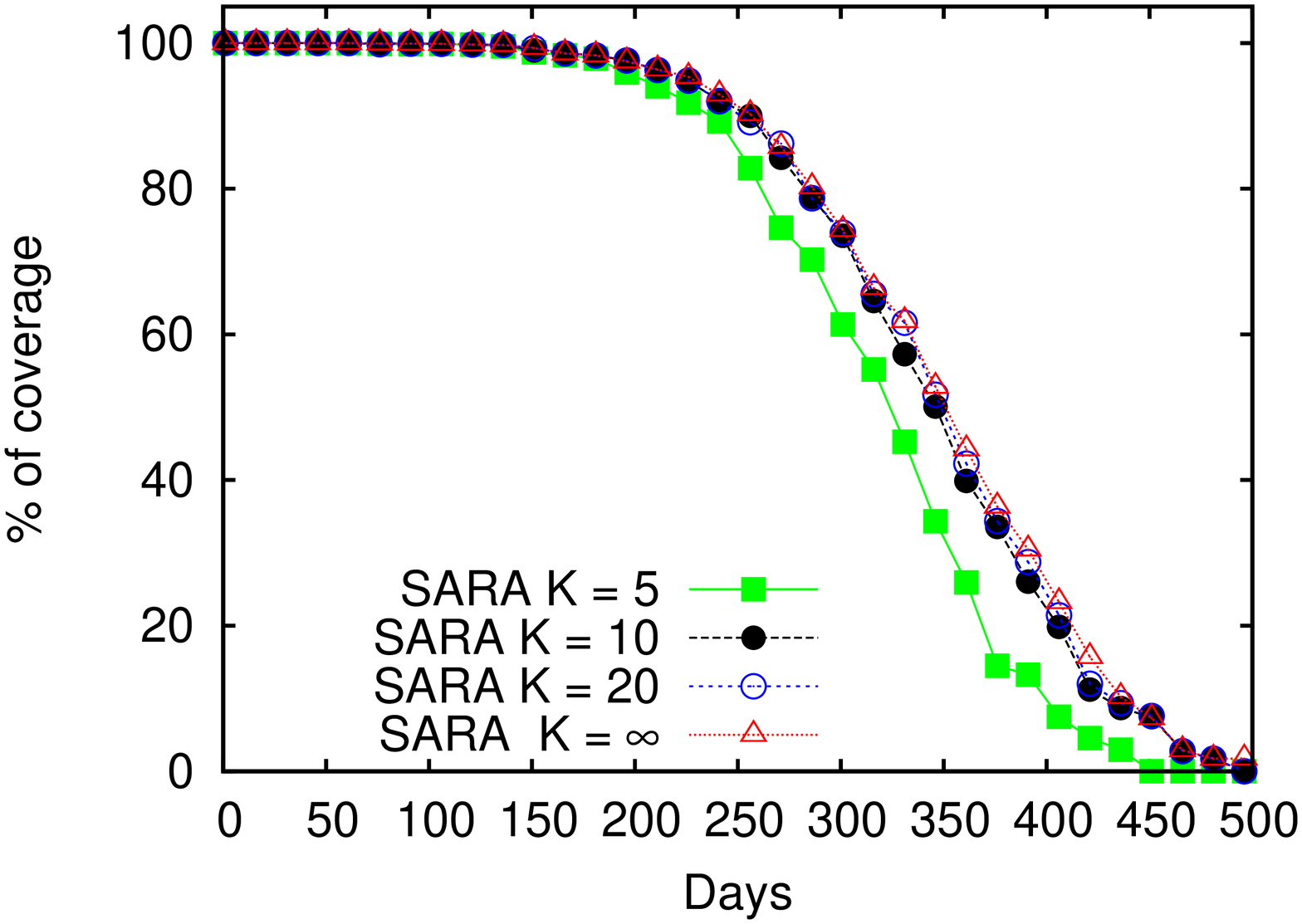}
\\
\hspace{-0.5cm}{\footnotesize{(c)}}&{\footnotesize{(d)}}
\end{tabular}
 \caption{Performance of $\alg$ under different settings of the maximum number of iterations $K$ (faster termination condition). Percentage of active sensors (a), sleeping sensors (b), dead sensors (c), percentage of coverage of the AoI (d). }
\label{fig:sensitivity_K}
\end{minipage}
\end{figure}


%
%


\subsection{Adjustable sensors: Homogeneous setting}
\label{sec:homo_var}

This section is devoted to a comparative analysis of the performance of $\alg$, $\dlm$ and $\gup$ in a scenario with only sensors with equal capabilities to adjust their sensing range.
As in the experiments of the previous sections, we considered 900 sensors whose range varied in the interval $[2m,6m]$.
It should be noted that this scenario is not the most general, but it is the one for which $\gup$ was specifically designed.
Therefore, it is reasonable to expect that $\gup$ show a good performance.
Nevertheless, the experiments highlight that even in this case,
the algorithm $\alg$ performs better.
Indeed thanks to the device homogeneity, the algorithm $\gup$
 is able to work in this scenario without creating coverage holes, but it is not able to fully 
exploit the adaptability of the sensing range as $\alg$ does, thanks to the use of Voronoi diagrams in the Laguerre geometry. 

In the following experiments we consider the modified version of $\dlm$
described in Section \ref{sec:two_approaches} in order to apply it to the scenario with adjustable sensing radii.
As we have already argued, this modified version is introduced in these experiments to highlight that this algorithm
cannot be trivially extended to the general scenario without a significant loss in performance.

Figure \ref{fig:time_VAR_HOM}(a) shows how the coverage achieved by the three algorithms decreases with time.
The loss in coverage with $\dlm$ is much faster than with $\gup$ and $\alg$, evidencing its inapplicability to this operative scenario.
In this case $\alg$ performs better than $\gup$. 
For instance, in correspondence to day 350, $\alg$ is still capable to cover about twice the 
extension of the area covered by $\gup$.
This evidences the capability of $\alg$ to prolong the network lifetime, when this is formulated as the time within which the network is still capable to cover
a given percentage $x$ of the AoI, independently of the value of $x$.

Figure \ref{fig:time_VAR_HOM}(b) and (c) show how the percentage of active and sleeping sensors varies with time. These percentages are calculated with respect to the whole set of available sensors. 
It should be noted that although $\dlm$ activates a very small percentage of the available sensors, it is penalized by the fact that the radius of the active sensors cannot be modulated by the algorithm (Figure \ref{fig:time_VAR_HOM}(e)). Thence the energy consumption per sensor is very high, as demonstrated by Figure 
 \ref{fig:time_VAR_HOM}(f) which shows how small is the residual energy under $\dlm$ after few operative time intervals, resulting in a very high percentage of dead sensors 
 \footnote{Hereby we refer to {\em dead sensor} as to devices which have exhausted their available energy.}
Notice that under $\dlm$ the number of active sensors shows a peak after about 50 operative intervals for a twofold reason.
On the one hand the greedy nature of the algorithm $\dlm$ results in the activation of redundant sensors, as it gives higher priority to the activation of sensors which have consumed less energy in the previous operative intervals.  
On the other hand, while at the first intervals, the algorithm is able to select the best suitable sensors to cover the AoI, after the death of several sensors some regions can only be covered by sensors which cause larger overlaps.

The algorithms $\alg$ and $\gup$  are able to modulate the sensing radius of the active sensors so as to reduce the coverage overlaps and save energy. Therefore, with respect to $\dlm$, more sensors are activated (Figure \ref{fig:time_VAR_HOM}(b)) working with lower sensing radius (Figure \ref{fig:time_VAR_HOM}(e)). 
This permits to the algorithms $\gup$ and $\alg$ to preserve more energy than $\dlm$ (\ref{fig:time_VAR_HOM}(f)). Notice that,  $\alg$ activates a higher number of sensors with smaller radius than $\gup$, thus being able to prolong the network lifetime by reducing  the amount of consumed energy.

\begin{figure}[h]
\centering
\begin{tabular}{ccc}
\hspace{-0.5cm}\includegraphics[width = 0.22\textwidth, angle=-90]{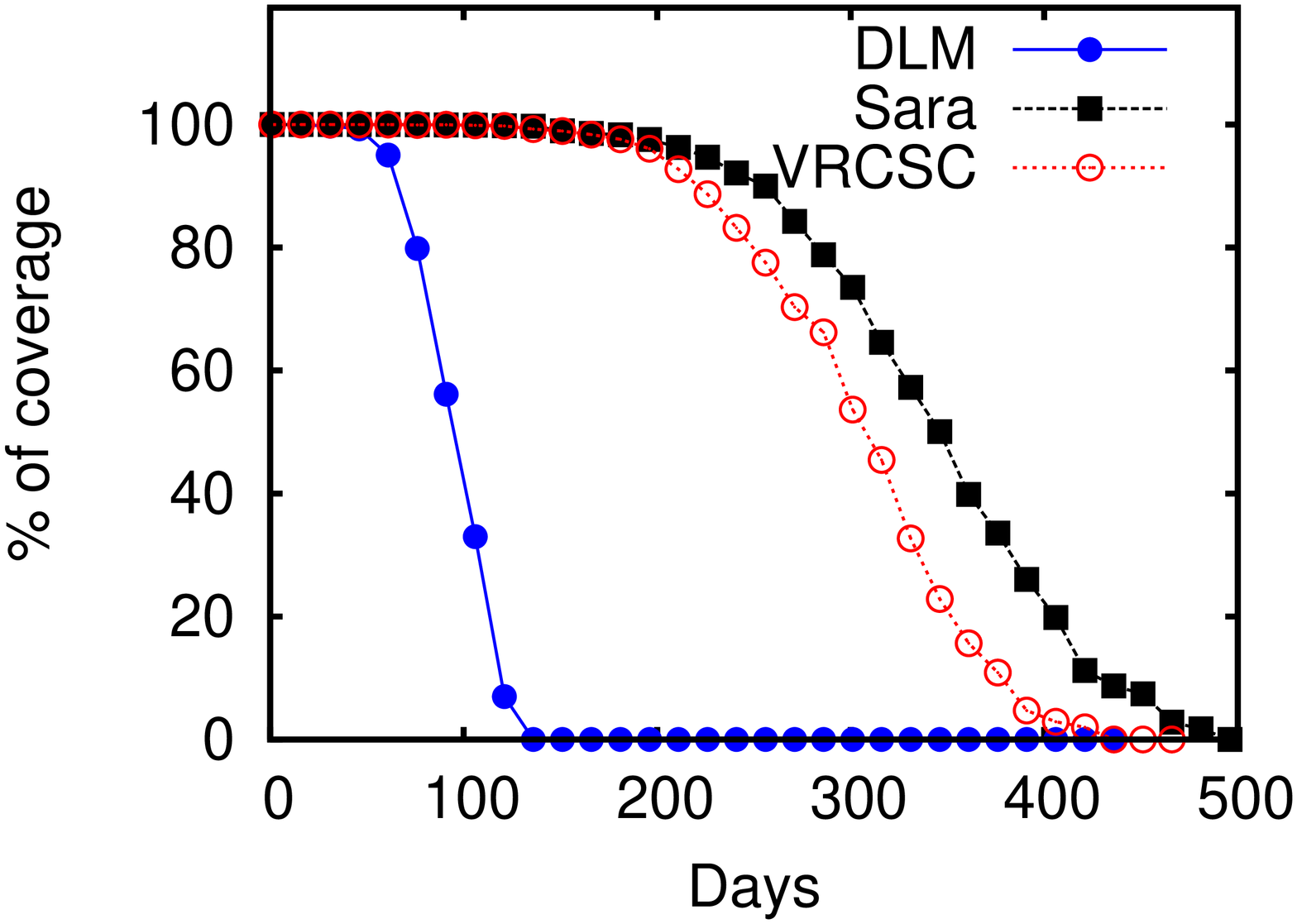} &
\includegraphics[width = 0.22\textwidth,angle=-90]{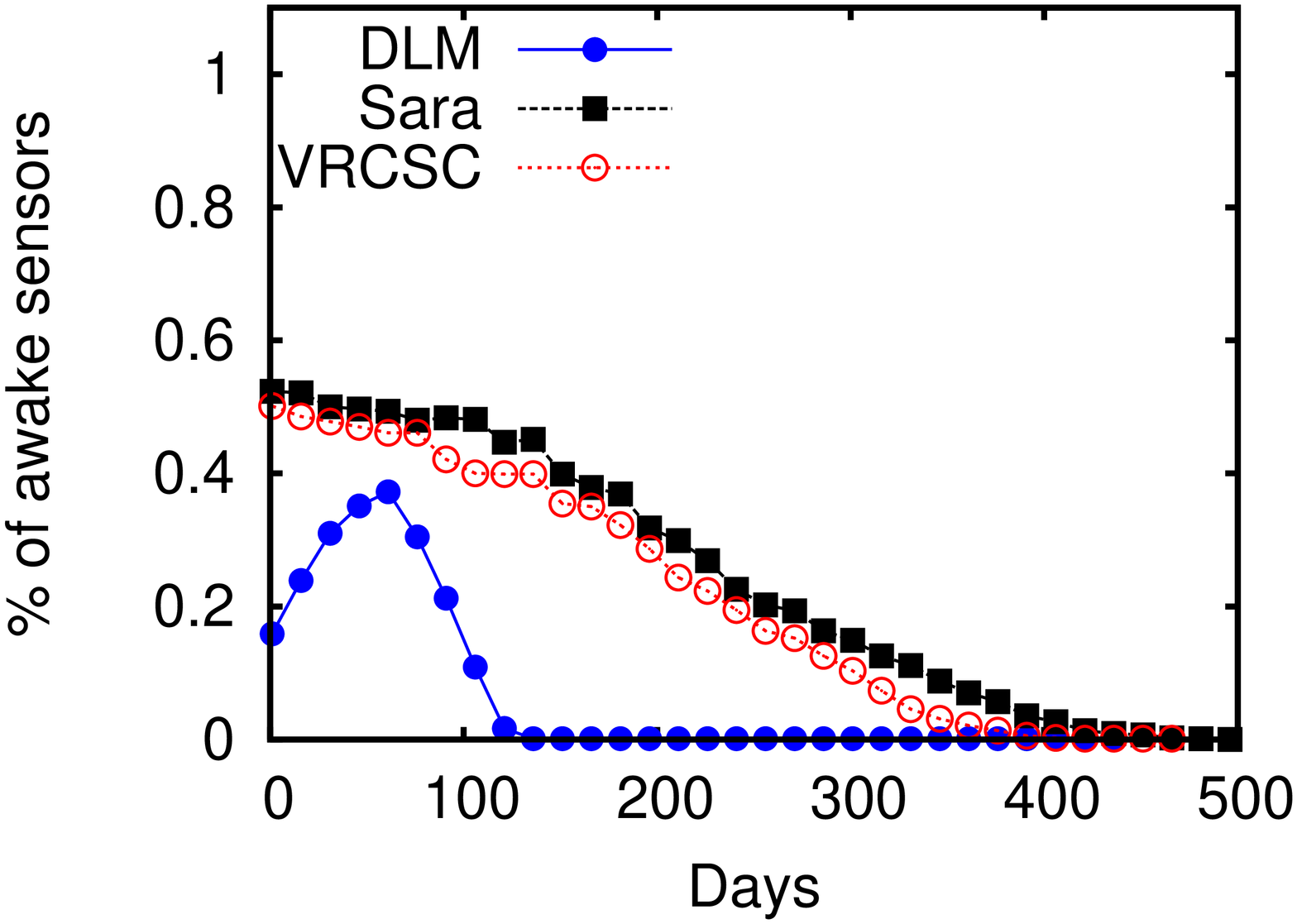}&
\includegraphics[width = 0.22\textwidth,angle=-90]{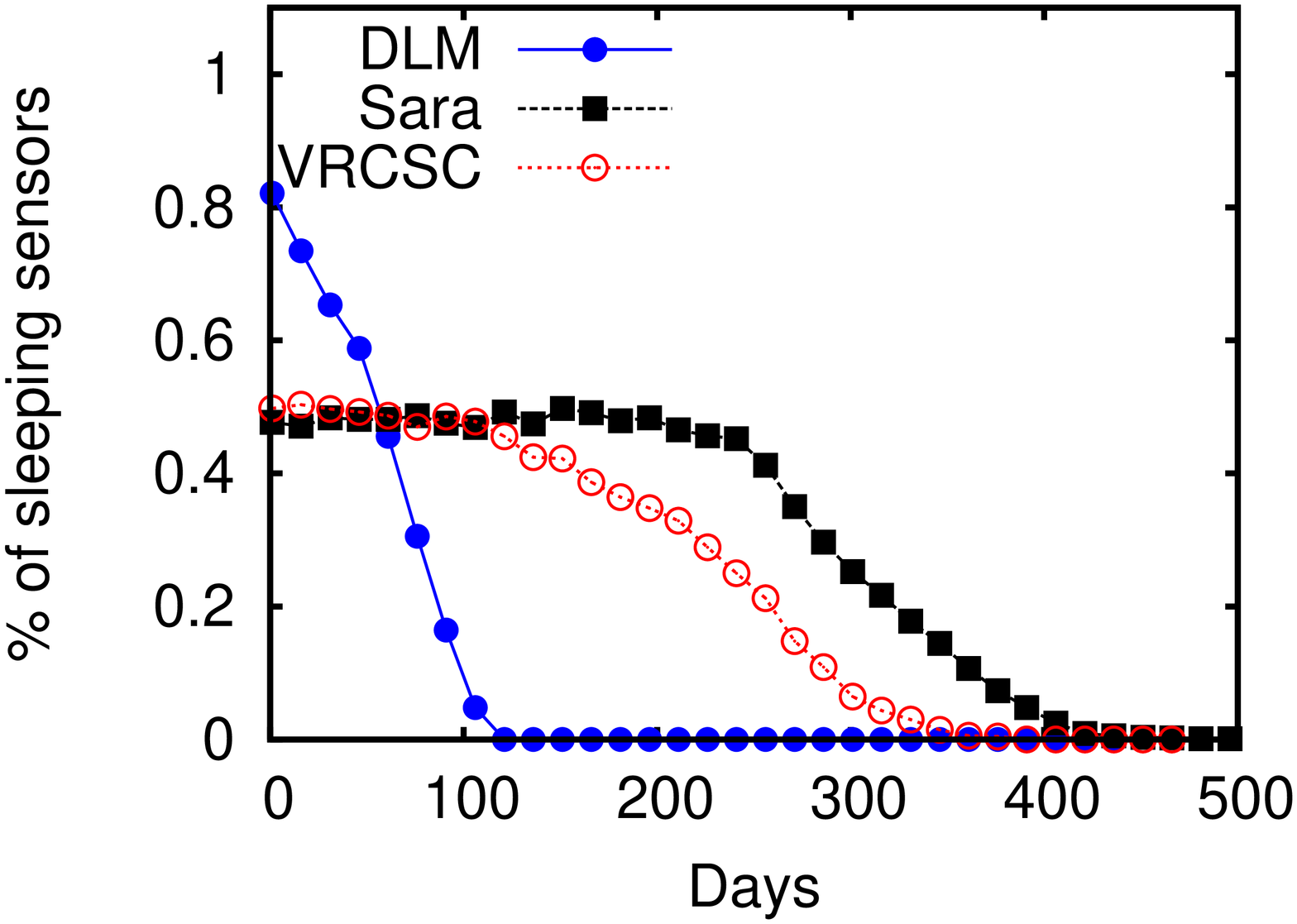}\\
\hspace{-0.5cm}{\footnotesize{(a)}}&{\footnotesize{(b)}}&{\footnotesize{(c)}}\\
\hspace{-0.5cm}\includegraphics[width = 0.22\textwidth,angle=-90]{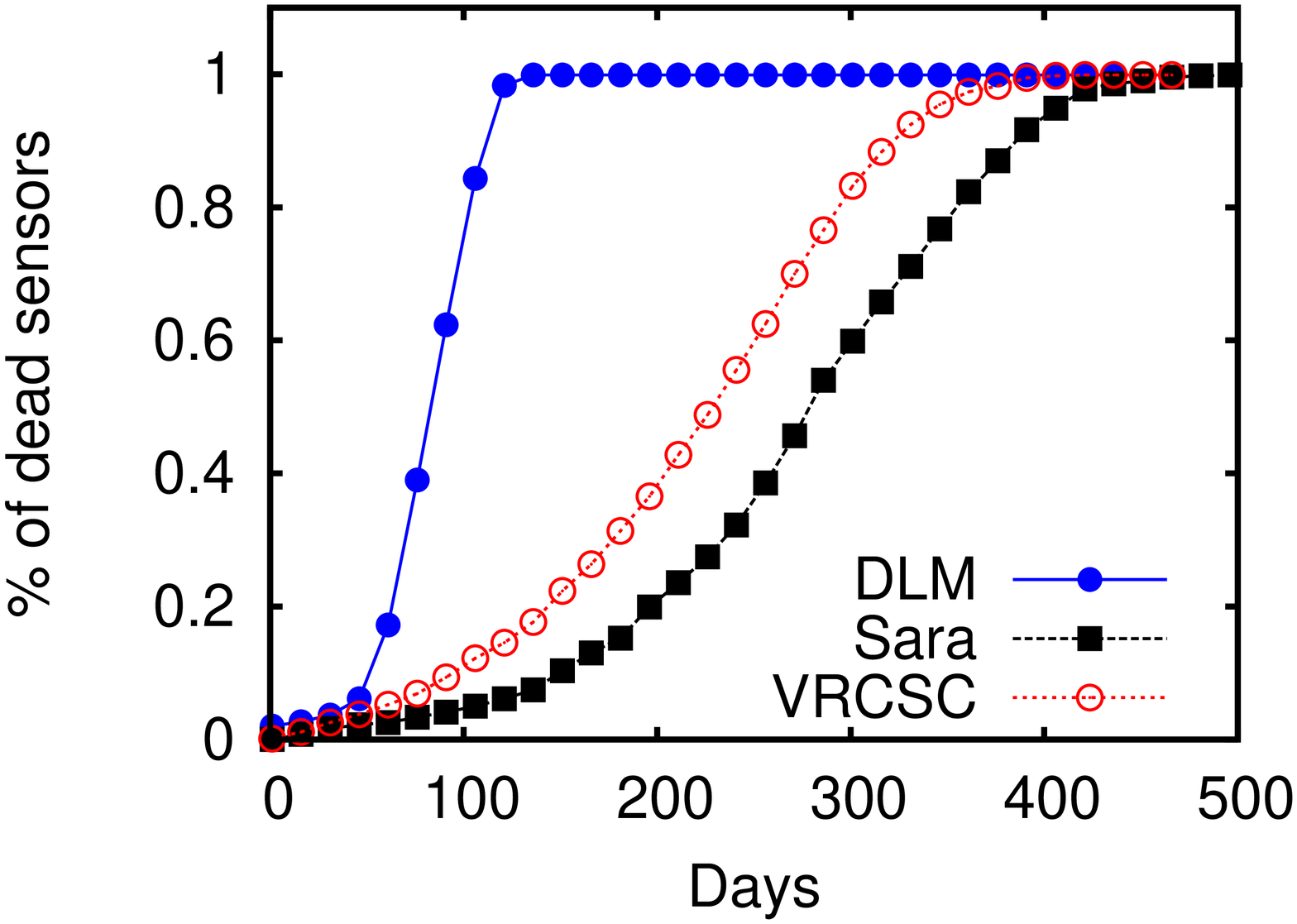}
&
\includegraphics[width = 0.22\textwidth,angle=-90]{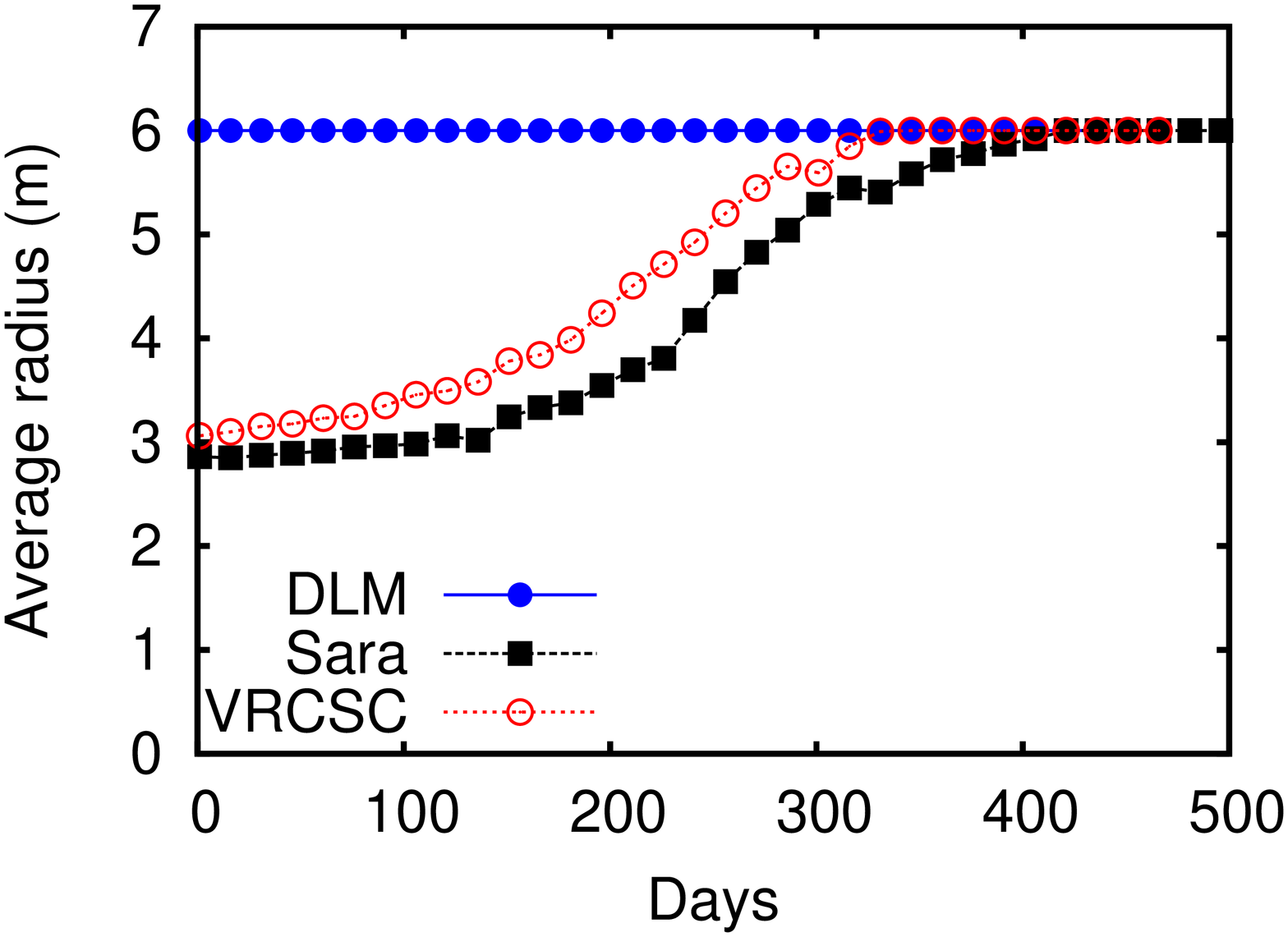}
&
\includegraphics[width = 0.22\textwidth,angle=-90]{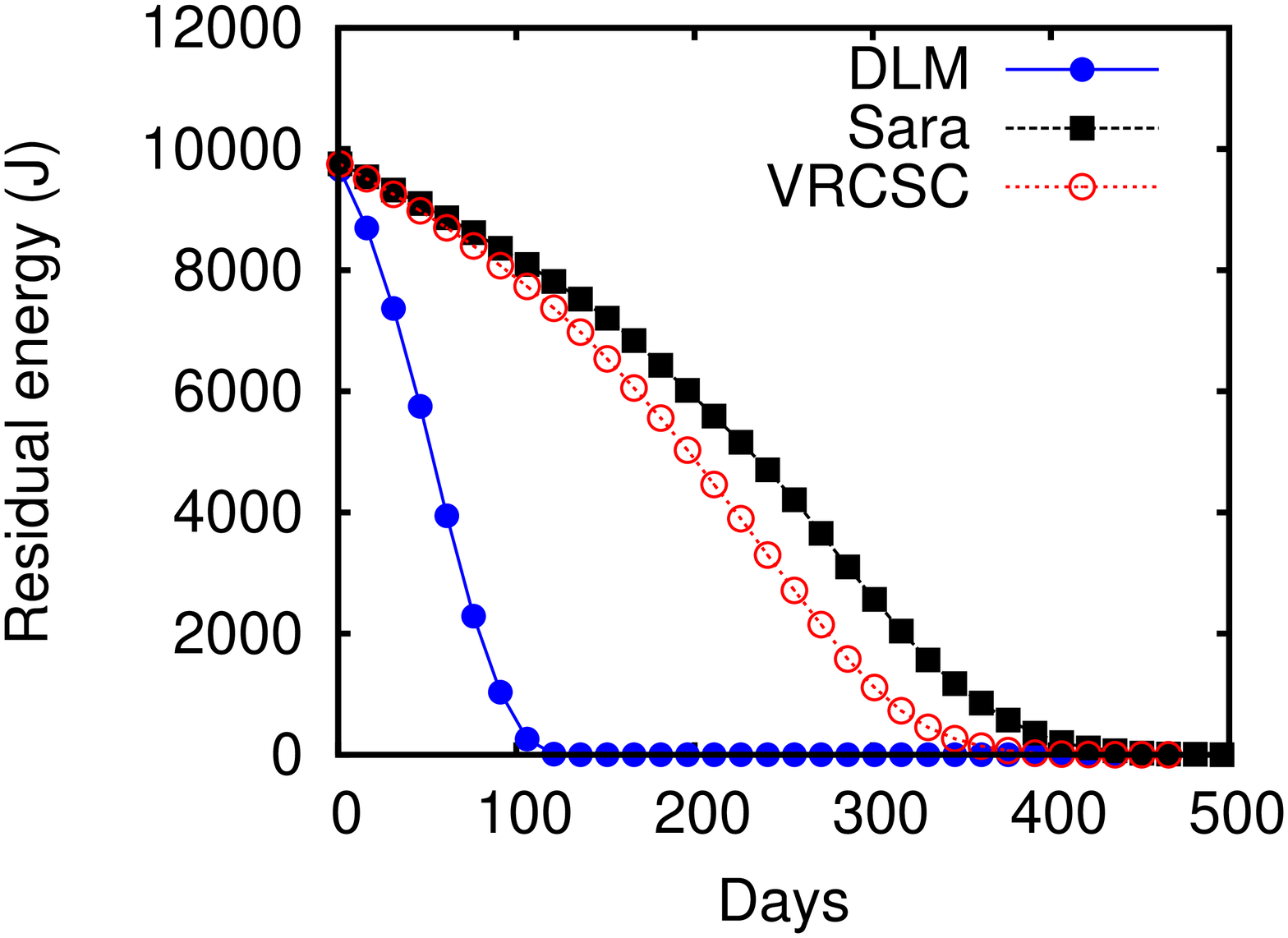}
\\
\hspace{-0.5cm}{\footnotesize{(d)}}&{\footnotesize{(e)}}&{\footnotesize{(f)}}\\
\end{tabular}
 \caption{Adjustable sensors: homogeneous setting. Comparative analysis of $\alg$, $\dlm$ and $\gup$. Percentage of coverage (a), active sensors (b), sleeping sensors (c), sensors with no residual energy (d).
 Average radius of the awake sensors (e) and average residual energy per sensor (f). Scenario with 900 equally equipped sensors.}
\label{fig:time_VAR_HOM}
\end{figure}

We now evaluate the benefits of the three algorithms in terms of lifetime improvements.
Figure \ref{fig:nr_VAR_HOM}(a) shows  the time when the algorithms are no longer able to achieve a coverage higher than the 80$\%$ of the area of interest by varying the number of available sensors.

\begin{figure}
\begin{center}
\begin{tabular}[c]{ccc}
\hspace{-0.5cm}
  \includegraphics[width = 0.22\textwidth,angle=-90]{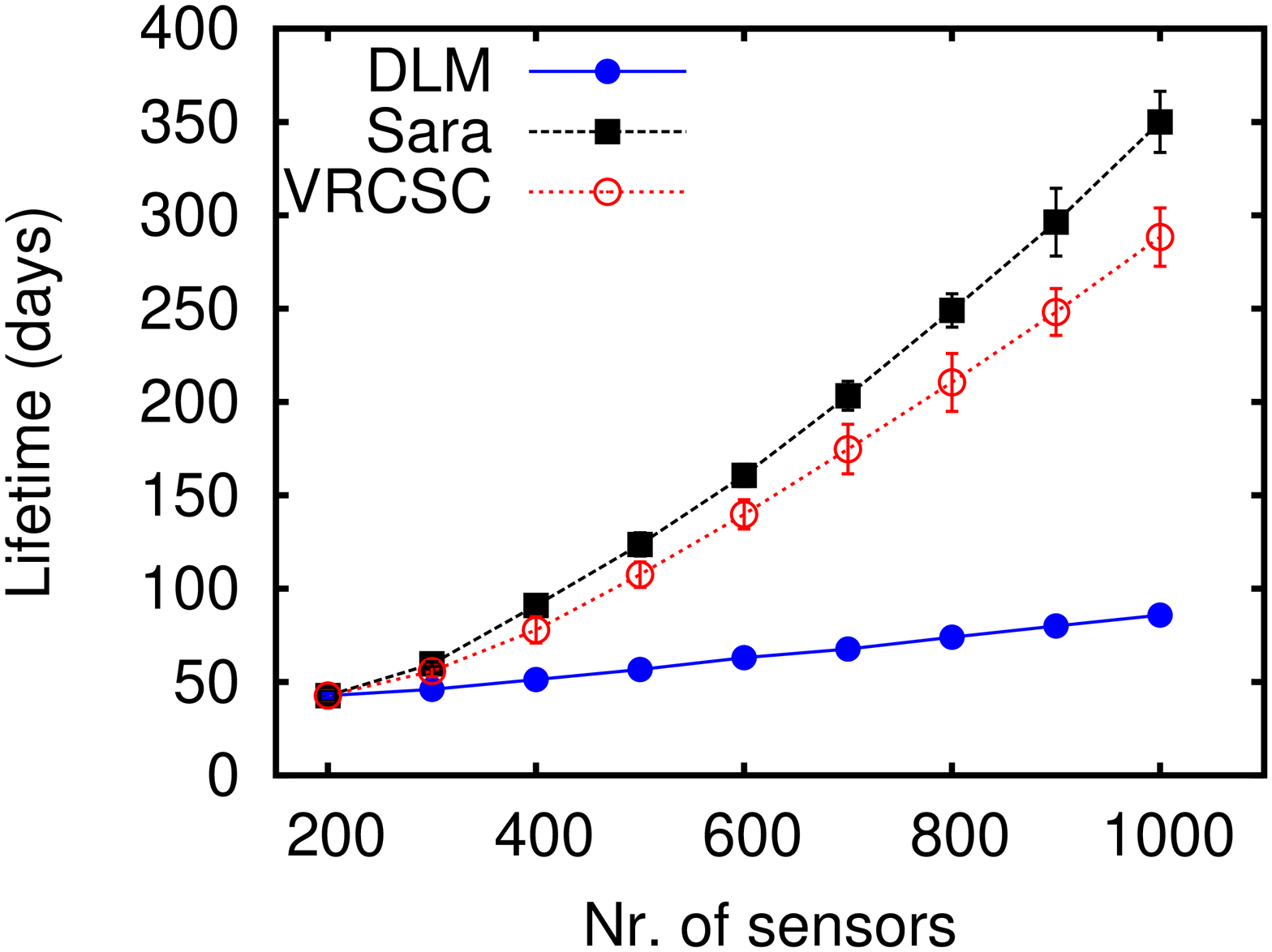}
  &{\includegraphics[width=0.22\textwidth,angle=-90]{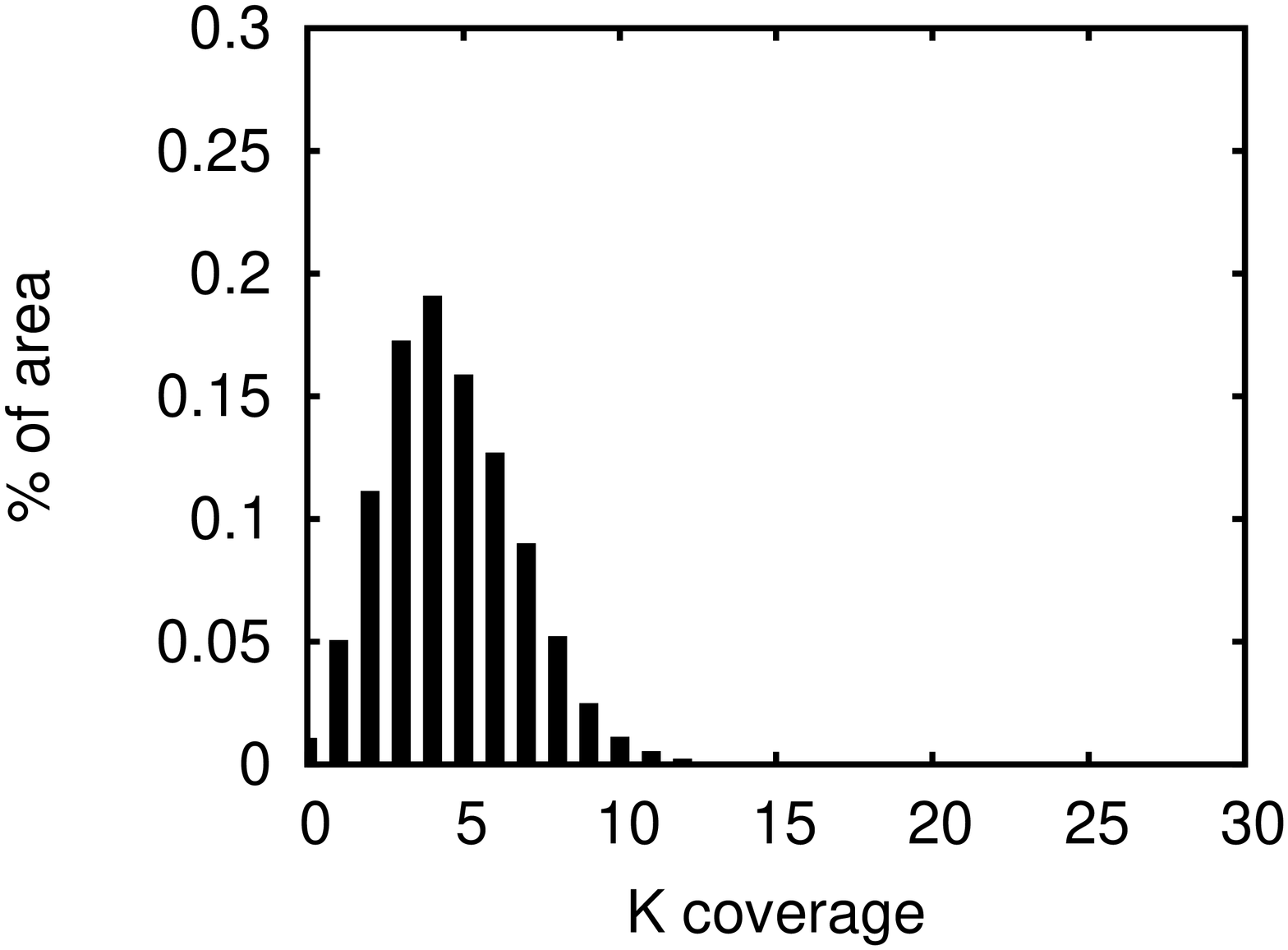}}
  &  {\includegraphics[width=0.22\textwidth,angle=-90]{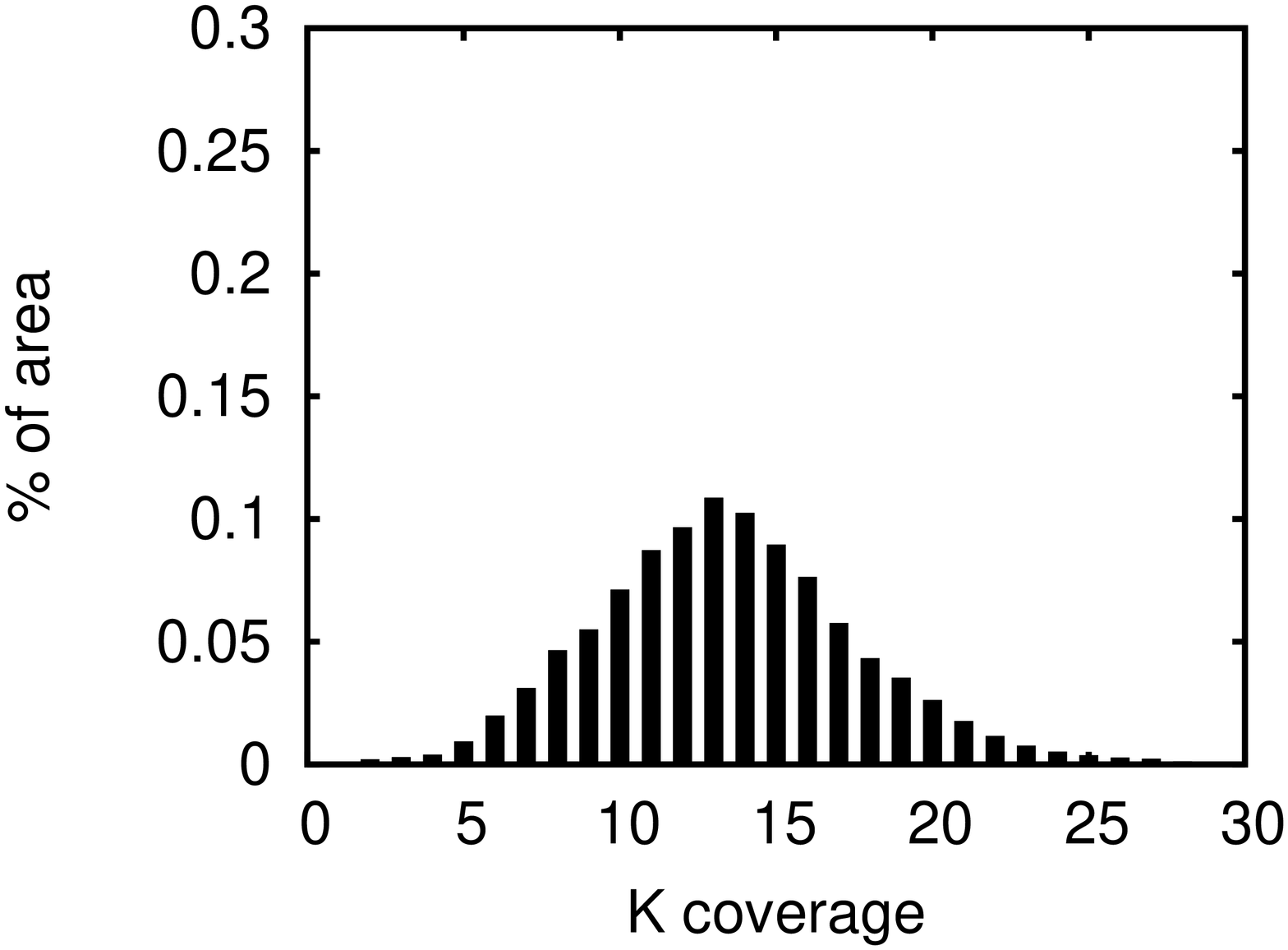}}
    \\
{\footnotesize{(a)}}&{\footnotesize{(b)}}&{\footnotesize{(c)}}
\end{tabular}
\end{center}
 \caption{{Adjustable sensors: homogeneous setting. Lifetime (a) and distribution of the sensors over an area of interest of 80m x 80m with a 
random deployment of 300 sensors (b) and 900 equally equipped  sensors (c).}}
\label{fig:nr_VAR_HOM}
 \end{figure}

 Figure \ref{fig:nr_VAR_HOM}(b) shows  the distribution of the sensors over the AoI  in the operative scenario with 300 sensors.
 This figure evidences that, despite the general redundancy, 
 a significant portion of the area of interest is either uncovered or covered by only few sensors.
 These sensors deplete their energy faster than others no matter which algorithm is in use.
This implies that with a tolerance of only 20$\%$ (due to the definition of lifetime as the time at which more than 20\% of coverage is lost) 
the three algorithms cannot do much to improve the network 
lifetime. 

For this reason, in this applicative scenario,
the lifetime of $\dlm$, $\gup$ and $\alg$ is about the same, as seen in Figure \ref{fig:nr_VAR_HOM}(a) when the number of available sensors is 300.
By contrast, Figure \ref{fig:nr_VAR_HOM}(c) shows the distribution of the sensors over the AoI  in the operative scenario with 900 sensors.
It shows that due to this higher density, there is more room for the algorithms to improve the network lifetime
by means of selective activation and radius reduction, as also evidence by Figure  \ref{fig:nr_VAR_HOM}(a)  when the number of available sensors is 900.
 The algorithm $\alg$ outperforms the other two by achieving a longer lifetime
being able to perform a more efficient activation policy. In particular, although this scenario is the most favorable to the algorithm $\gup$, $\alg$ is able to always achieve a longer lifetime. For instance, when the number of sensors is 1000, the algorithm $\alg$ achieves and increase of 20\% in the network lifetime with respect to $\gup$ (350 days for $\alg$ versus 290 days for $\gup$).

\begin{figure}
\begin{center}
\begin{tabular}[c]{ccc}
\hspace{-0.5cm}
  \includegraphics[width = 0.22\textwidth,angle=-90]{resizable/homogeneous/lifetime/lifetime_80}
  &{\includegraphics[width=0.22\textwidth,angle=-90]{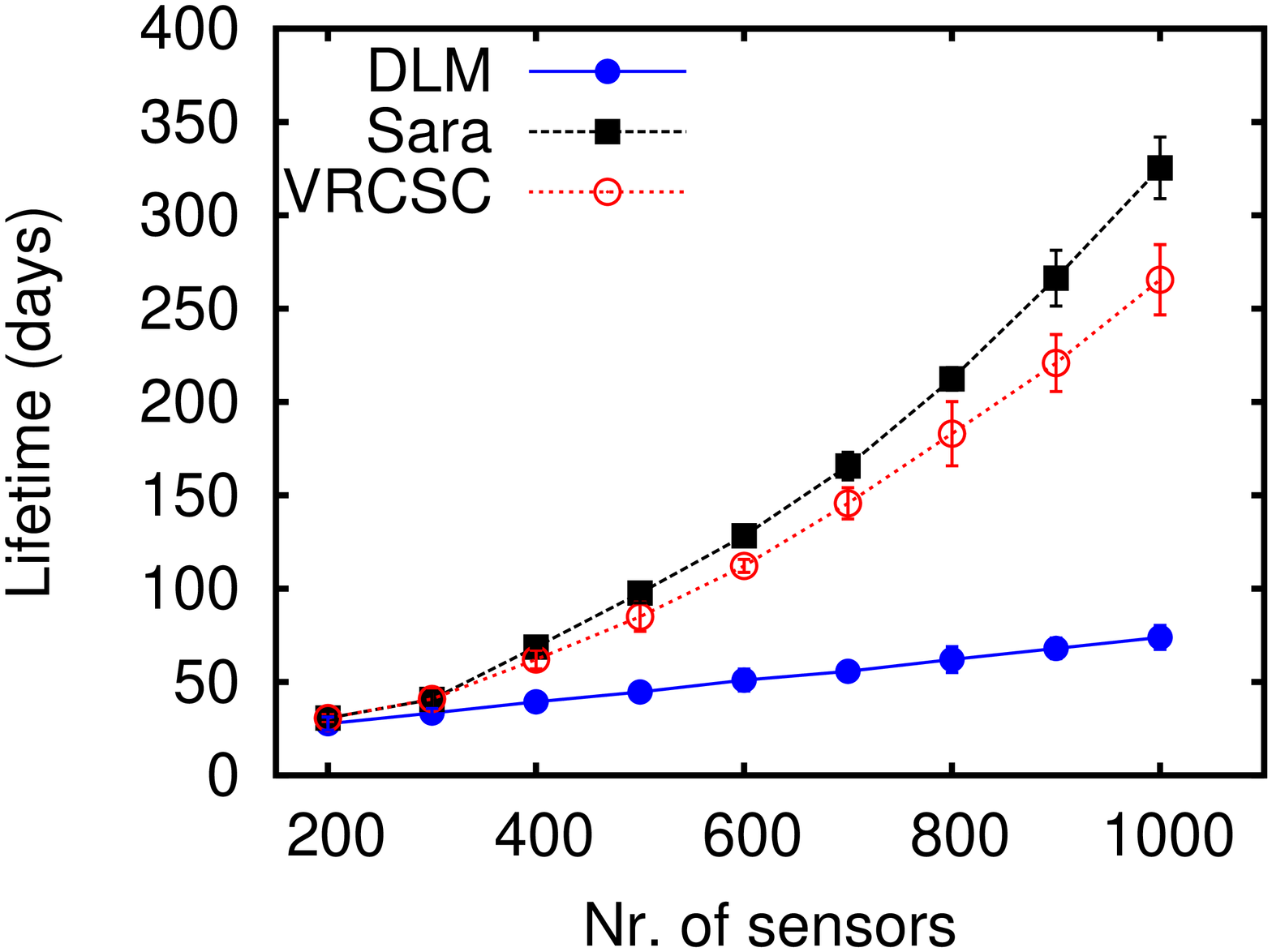}}
  &  {\includegraphics[width=0.22\textwidth,angle=-90]{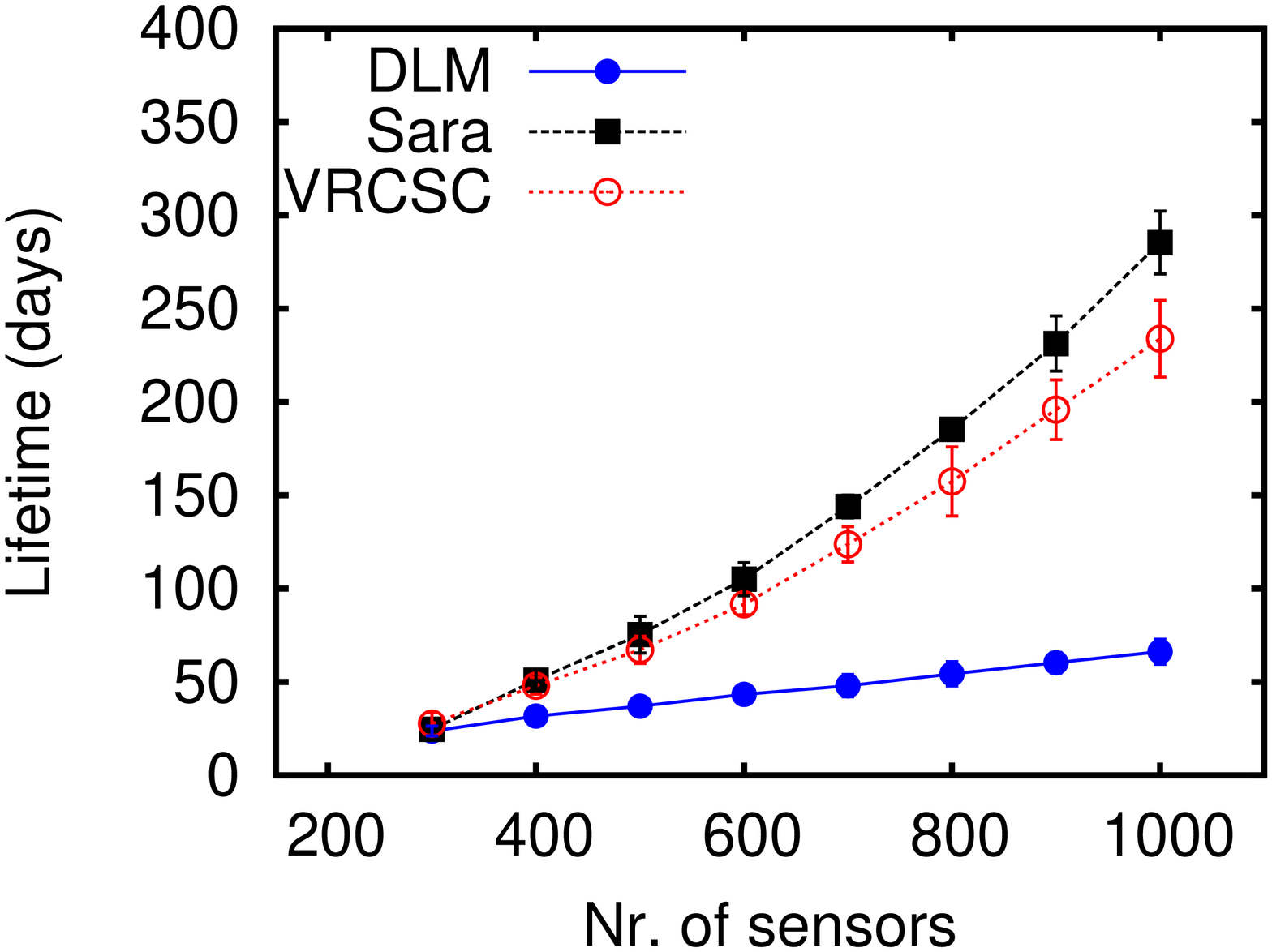}}
    \\
{\footnotesize{(a)}}&{\footnotesize{(b)}}&{\footnotesize{(c)}}
\end{tabular}
\end{center}
 \caption{{Adjustable sensors: homogeneous setting. Lifetime achieved by the three algorithms expressed as the time after which the algorithm is no longer capable to 
 cover more that 80\% (a), 90\% (b) and 95\% (c) of the AoI.}}
\label{fig:lifetimeVAR_HOM}
 \end{figure}
 
 In Figure \ref{fig:lifetimeVAR_HOM}
 we compare the algorithms in terms of network lifetime by increasing the number of deployed sensors. We consider the time at which the coverage of the AoI goes  below the 80\% (a), 90\% (b) and 95\%(c).
Notice that, even if our algorithm does not specifically target a particular notion of lifetime, it outperforms the other two also under other possible definitions of lifetime.

\subsection{Adjustable sensors: heterogeneous setting}
In this section we analyze a scenario with only sensors with adjustable radius. Differently from the setting of the previous experiments,
we now consider the case of sensors unequally equipped.
In particular the 50\% of the available sensors is capable to adjust its radius up to a value of 6m, whereas the remaining  50\% can only reach a
sensing radius 3m long.

\begin{figure}
\begin{center}
\begin{tabular}[c]{ccc}
\hspace{-0.5cm}
  \includegraphics[width = 0.22\textwidth,angle=-90]{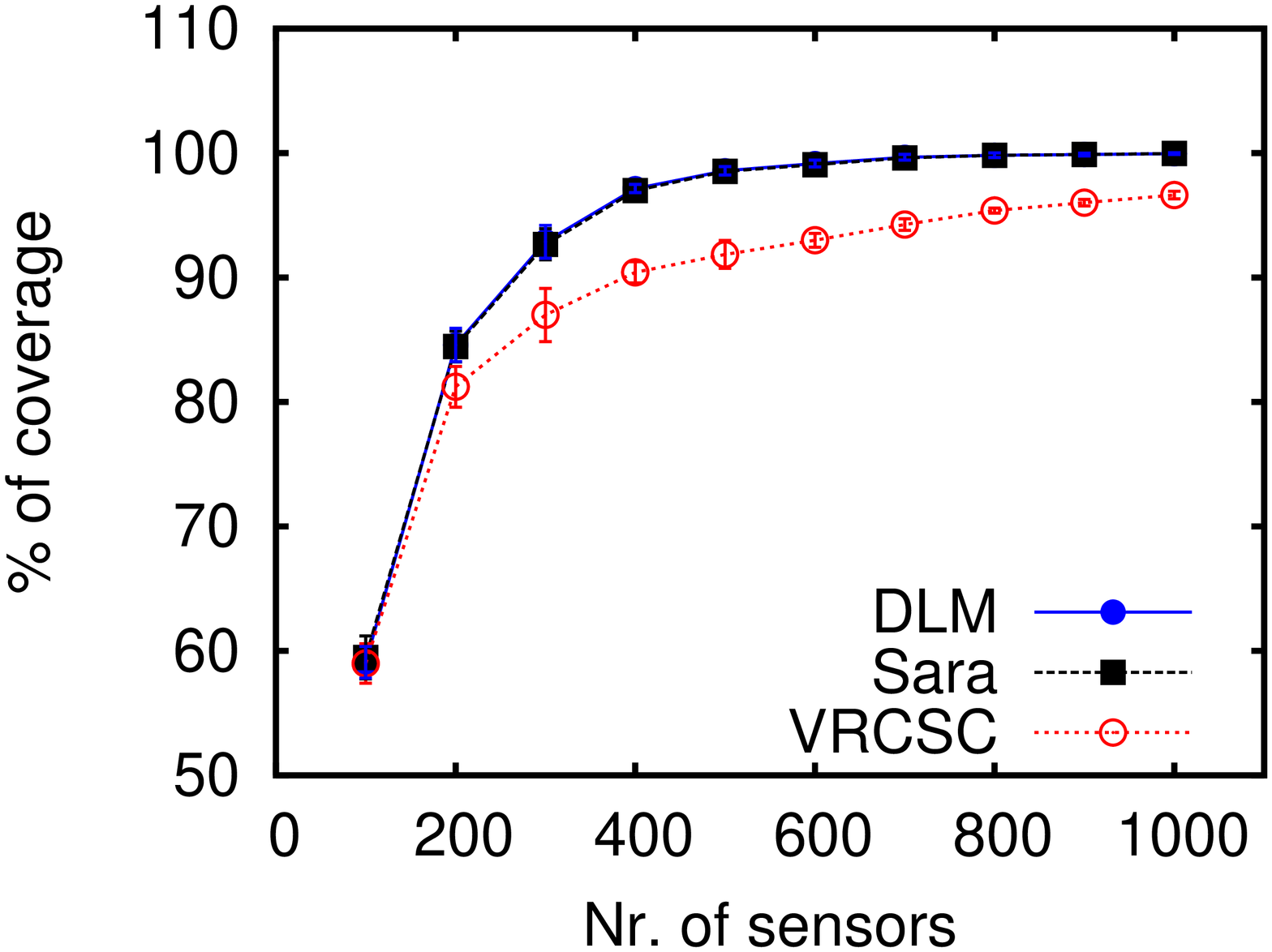}
    &\includegraphics[width = 0.22\textwidth,angle=-90]{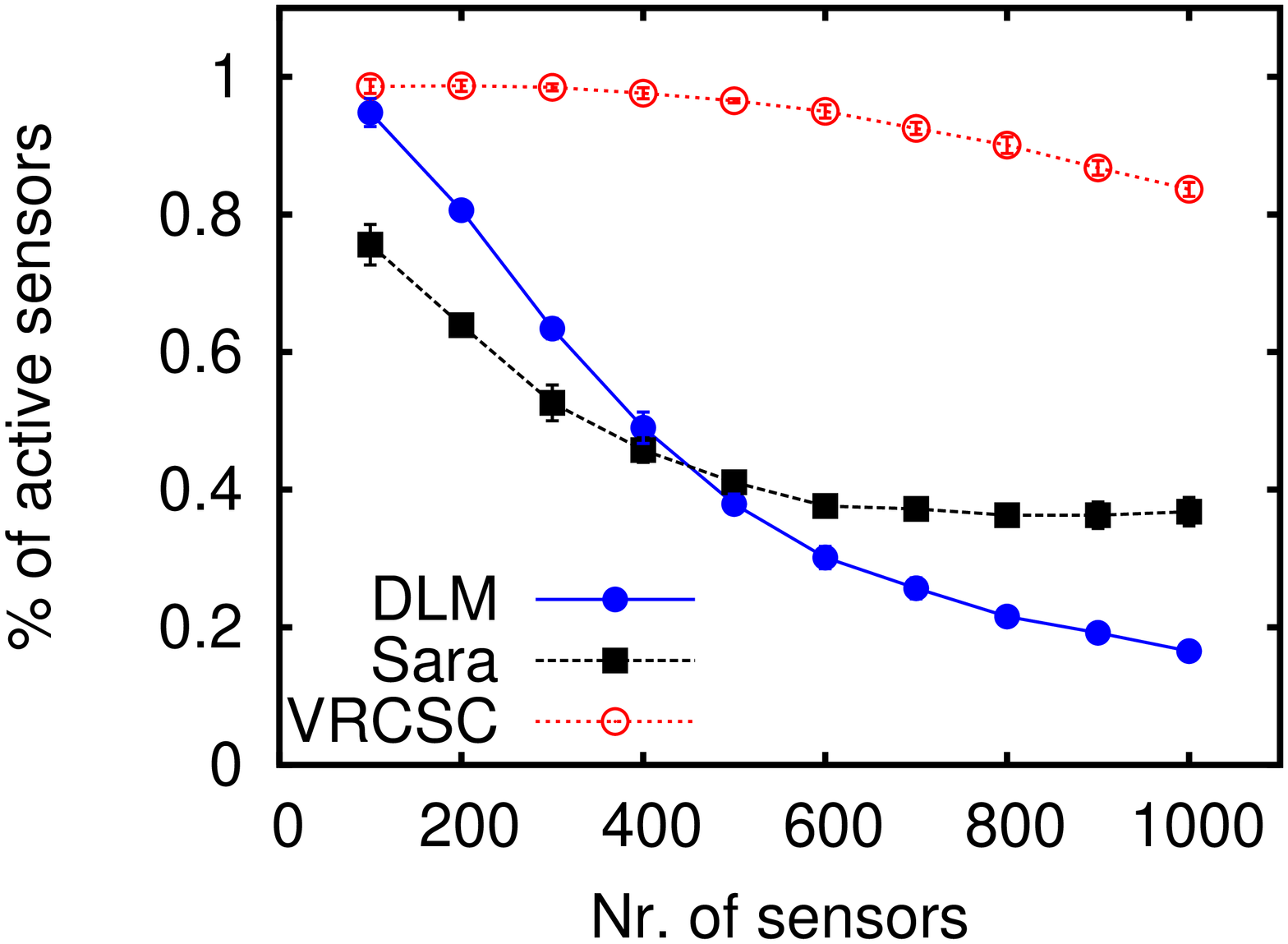}
  &  \includegraphics[width = 0.22\textwidth,angle=-90]{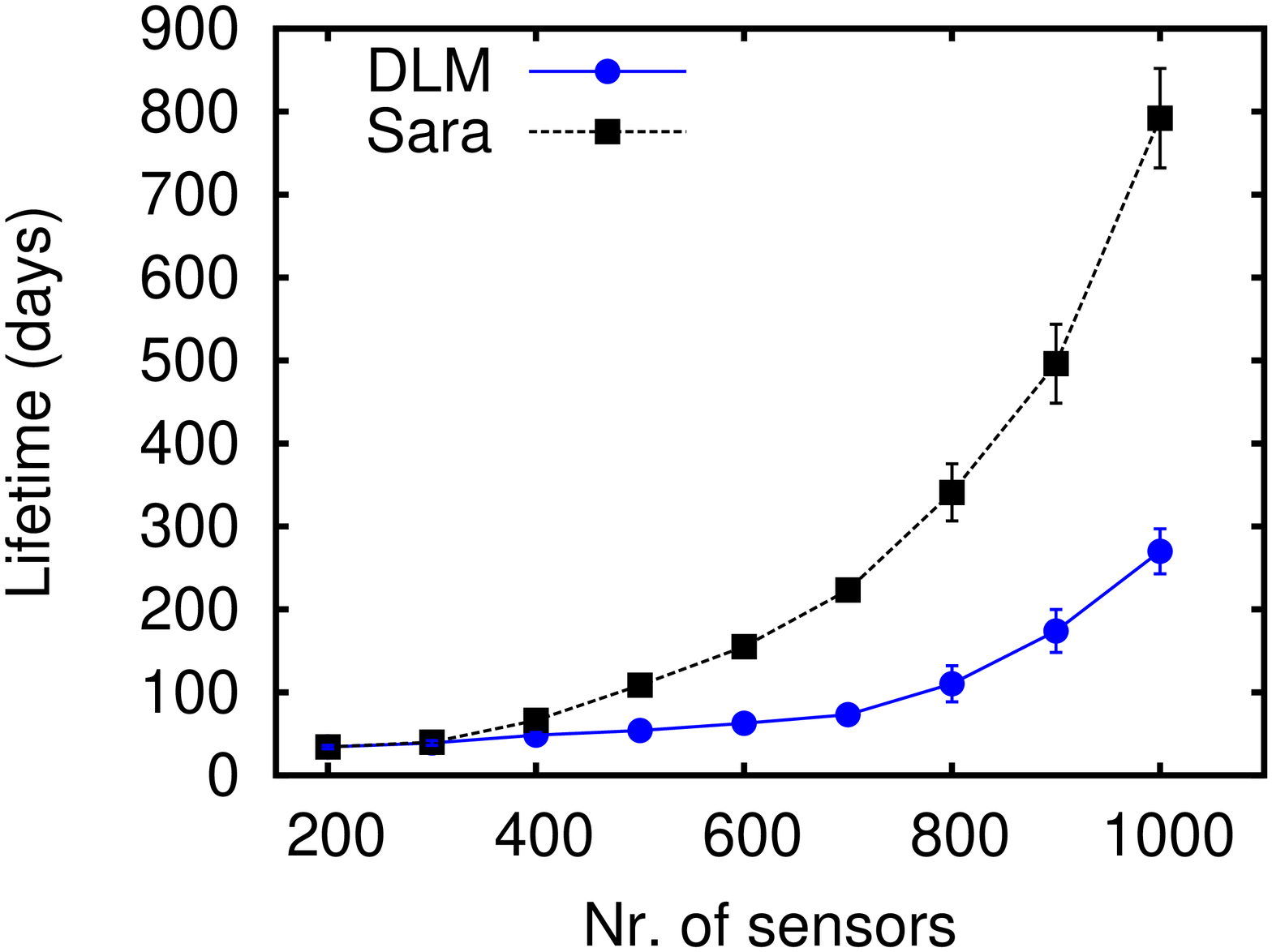}
    \\
\hspace{-0.5cm}{\footnotesize{(a)}}&{\footnotesize{(b)}}&{\footnotesize{(c)}}
\end{tabular}
\end{center}
 \caption{{Adjustable sensors: heteoregeneous setting. Coverage (a), percentage of active sensors (b), and lifetime of the network (c) in the 
operative scenario with {\bf heterogeneous} sensors with {\bf adjustable} sensing radius.}}
\label{fig:nr_VAR_HET}
 \end{figure}

In Figure \ref{fig:nr_VAR_HET}(a) we show the coverage achieved by the set of active sensors after the configuration obtained by running the algorithms.
Notice that in this case we selected the very first execution of the algorithms, 
hence the number of dead sensors is zero for all of them.
Despite the high energy availability of all the sensors,
the algorithm $\gup$ is not able to guarantee the complete coverage of the AoI at any time, even though it activates a very high percentage of sensors as shown in Figure \ref{fig:nr_VAR_HET}(b).
This is due to the way it governs the radius configuration decisions on the basis of Voronoi diagrams.
As we have already mentioned the use of Voronoi diagrams is correct only in the case of homogeneous sensing radii, while in the case of heterogeneous
setting it is necessary to model the coverage responsibility regions of the devices in the Laguerre metric space. 
On the contrary both $\alg$ and $\dlm$ are able to achieve the maximum coverage extent with a significantly 
lower percentage of active sensors.

 Figure \ref{fig:nr_VAR_HET}(b) highlights that, when working with a low number of available sensors, $\dlm$ activates a large fraction of redundant sensors having a small radius. In fact, being $\dlm$ based on a priority criterion which takes into account the number of intersection points when making activation decisions, when the number of sensors is so small, it is not able to give more priority to the sensors which contribute a better coverage.
 When the initial density of the available sensors is higher, the number of intersection points is more significant in reflecting the coverage that each sensor is capable to contribute to, resulting in a higher number of active sensors with larger sensing range.
 

Figure \ref{fig:nr_VAR_HET}(c) shows the network lifetime achieved by $\dlm$ and $\alg$ varying the number of available sensors
 in the heterogeneous setting. Although $\dlm$ is able to work in a heterogeneous scenario  with sensors having different sensing ranges,
it cannot exploit the device capability to adjust their radius.
Therefore the lifetime under $\dlm$  is much shorter than under $\alg$.
In particular, when the number of sensors is 900, the lifetime of $\alg$ is almost twice the lifetime of $\dlm$.

\begin{figure}
\begin{center}
\begin{tabular}[c]{ccc}
\hspace{-0.5cm}
  \includegraphics[width = 0.22\textwidth,angle=-90]{resizable/heterogeneous/lifetime/lifetime_80}
  &{\includegraphics[width=0.22\textwidth,angle=-90]{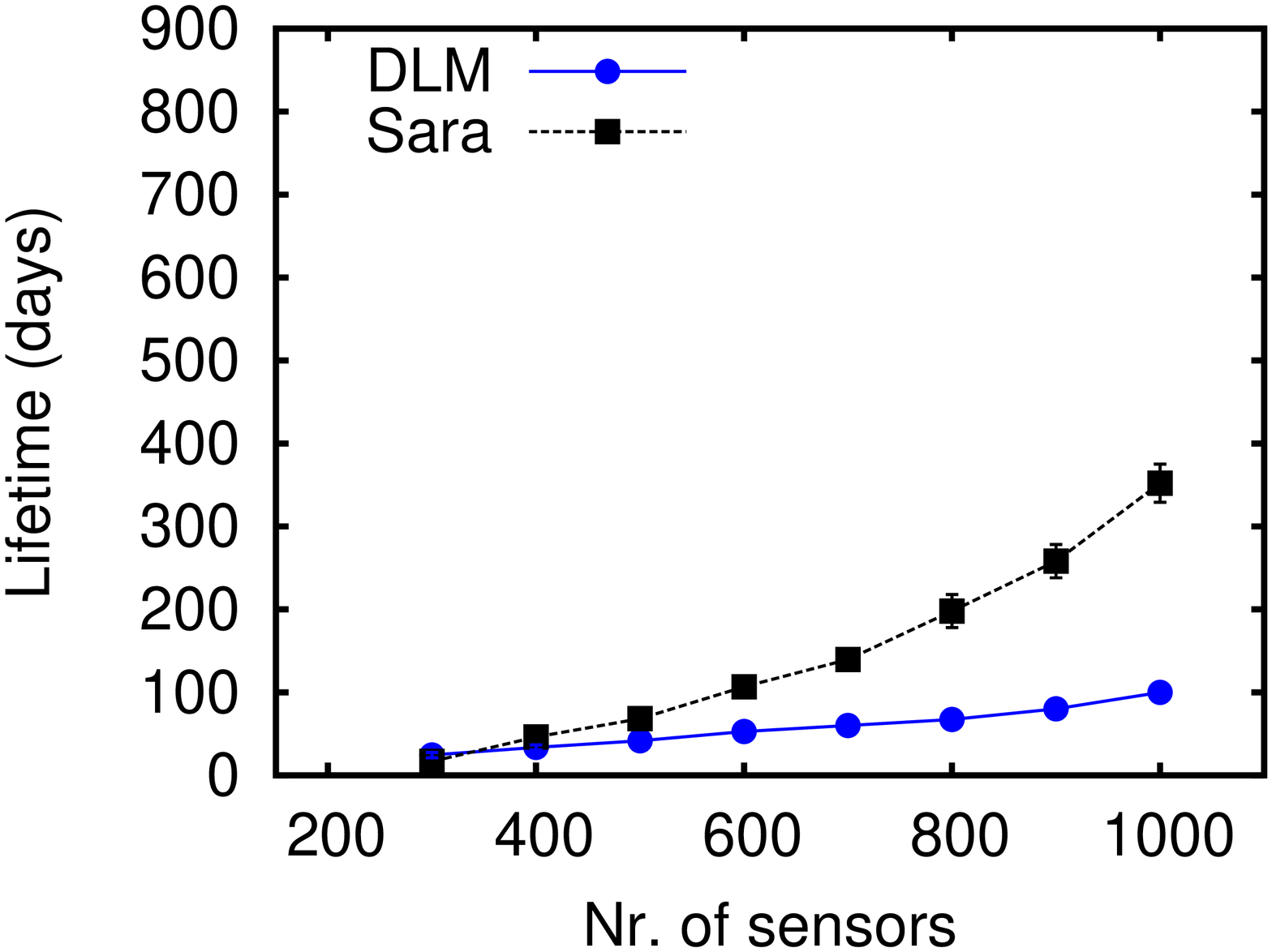}}
  &  {\includegraphics[width=0.22\textwidth,angle=-90]{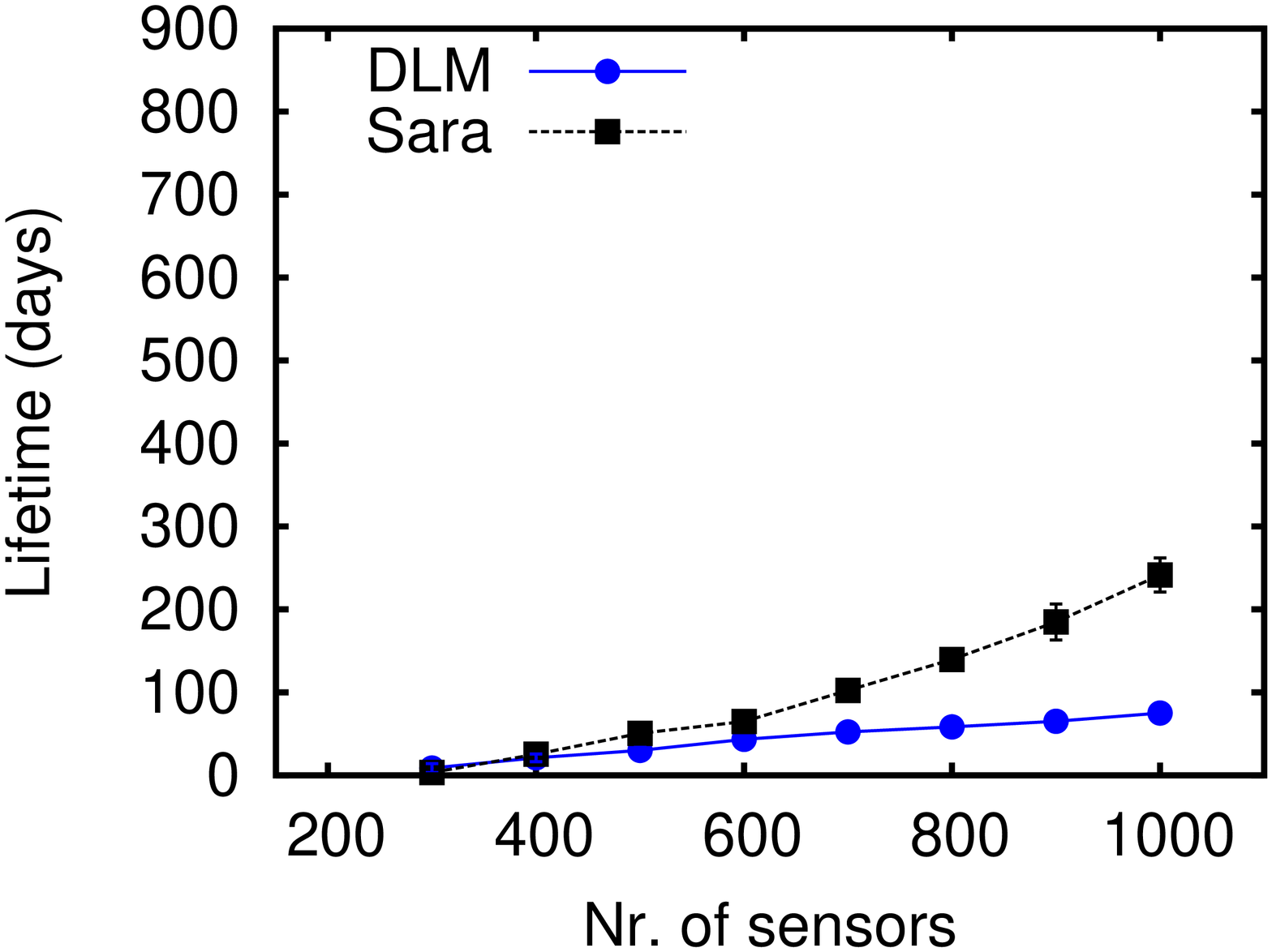}}
    \\
{\footnotesize{(a)}}&{\footnotesize{(b)}}&{\footnotesize{(c)}}
\end{tabular}
\end{center}
 \caption{{Adjustable sensors: heterogeneous setting. Lifetime achieved by $\alg$ and $\dlm$ expressed as the time after which the algorithm is no longer capable to 
 cover more that 80\% (a), 90\% (b) and 95\% (c) of the AoI.}}
\label{fig:lifetimeVAR_HET}
 \end{figure}
 
 In Figure \ref{fig:lifetimeVAR_HET}
 we compare the algorithms in terms of network lifetime by increasing the number of deployed sensors. We consider the time at which the coverage of the AoI goes  below the 80\% (a), 90\% (b) and 95\%(c).
Even under other possible definition of network lifetime, $\alg$ outperforms $\dlm$.

Notice that the network lifetime has a non linear dependence on the number of available sensors, as it increases more than linearly.
This is due to the non-linear dependence of the energy consumption with respect to the sensing range.
Indeed, the more sensors can be activated at small radius, the lower is the energy consumption and the longer is the lifetime.

\subsection{Fixed sensors: heterogeneous setting}

In this section we consider a scenario where sensors
have a fixed sensing radius. We focus on the case where sensors have heterogeneous
sensing capabilities: Half of the sensors have a sensing radius of $3$m while the other
half have a sensing radius of $6$m.
This is the scenario for which $\dlm$ was specifically designed.
In this setting $\gup$  is not able to guarantee maximum coverage in case of
sensor heterogeneity.
Therefore we will display only results
for $\dlm$ and $\alg$.

The experiments show that $\alg$ outperforms $\dlm$ in terms of percentage of the AoI covered
over time (Figure \ref{fig:time_FIX_HET}(a)) and results into a lower number of dead sensors over time (Figure \ref{fig:time_FIX_HET}(c)).
The percentage of awake sensors,
displayed in Figure \ref{fig:time_FIX_HET}(b), shows a similar trend (for the same reason) than that discussed
in Section \ref{sec:homo_var}. $\dlm$ experiences a higher number of awake sensors than $\alg$ during the first $120$ days.
As a consequence, the number of sensors which are put to sleep (obtained as a complement to 1 of the sum of awake and dead sensors)
will be much lower than in $\alg$.
When time increases the reduced number of awake sensors in $\dlm$ reflects the
 high
number of dead nodes, and consequently the poor coverage performance.
These observations explain the fact that $\alg$ experiences longer network lifetimes than $\dlm$.
This improvement is as high as twofold (Figure  \ref{fig:time_FIX_HET}(f)).

Figure \ref{fig:time_FIX_HET}(d) and (e)
shows the percentage of awake sensors with large and small radius under the execution of $\dlm$ and $\alg$, respectively.
{\color{black}It is interesting to note that initially, under $\dlm$, the majority of awake sensors have large radius. Nevertheless, after
very few operative time intervals, nodes with large radius quickly deplete their energy, and after day 100, $\dlm$ can only work with sensors having
 small radius.
On the contrary,} $\alg$ is able to successfully exploit device heterogeneity
from the beginning, by activating sensors with large and small radius
in different percentages, on the basis of coverage requirements.
As a consequence, only at about day 200 $\alg$ works with only sensors having small radius.
For this reason the peak in Figure \ref{fig:time_FIX_HET}(b) in the number
of active sensors is located on the right with respect to the one of $\dlm$.


\begin{figure}[h]
\centering
\begin{tabular}{ccc}
\hspace{-0.5cm}\includegraphics[width = 0.22\textwidth, angle=-90]{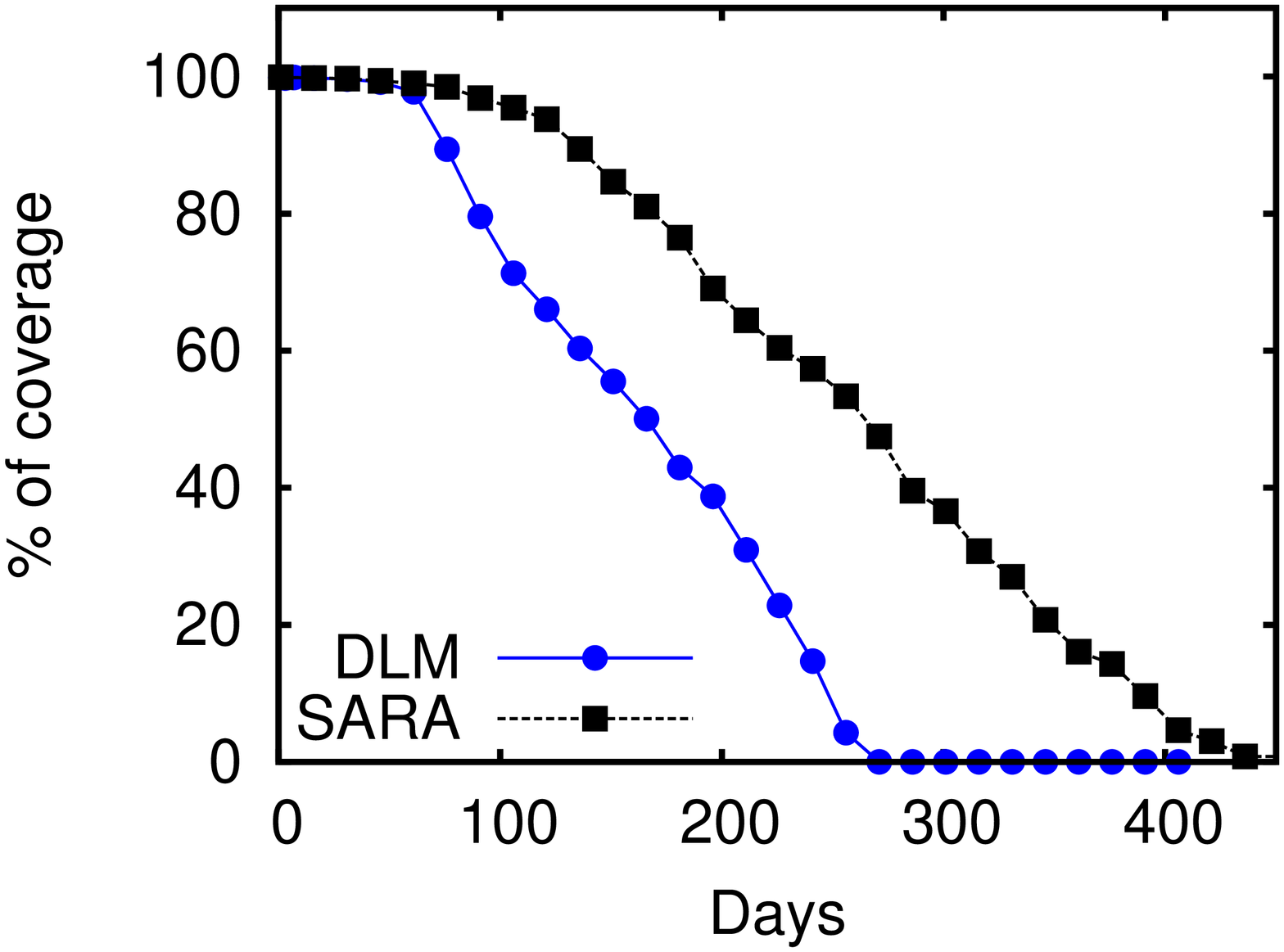} &
\includegraphics[width = 0.22\textwidth,angle=-90]{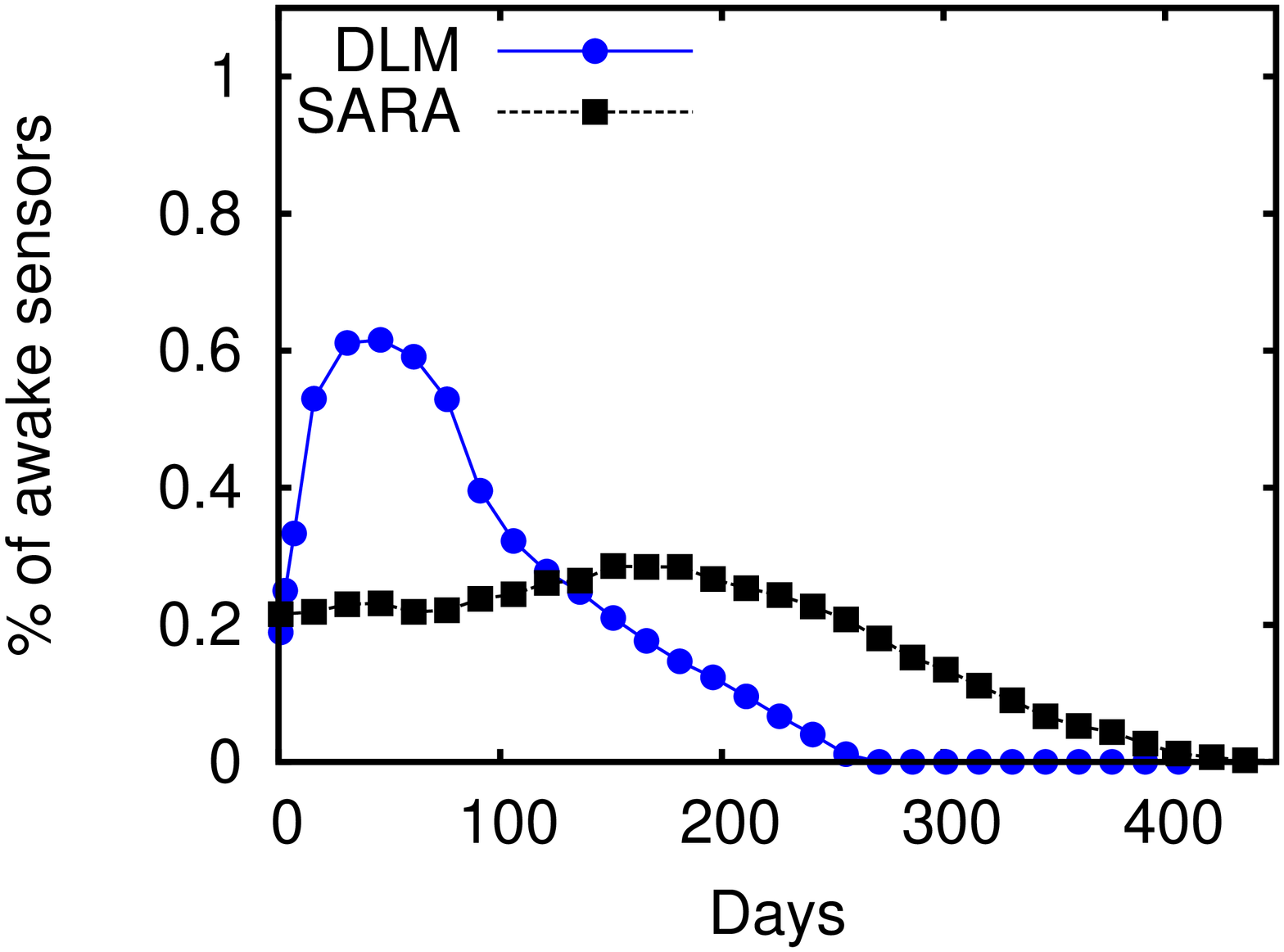}&
\includegraphics[width = 0.22\textwidth,angle=-90]{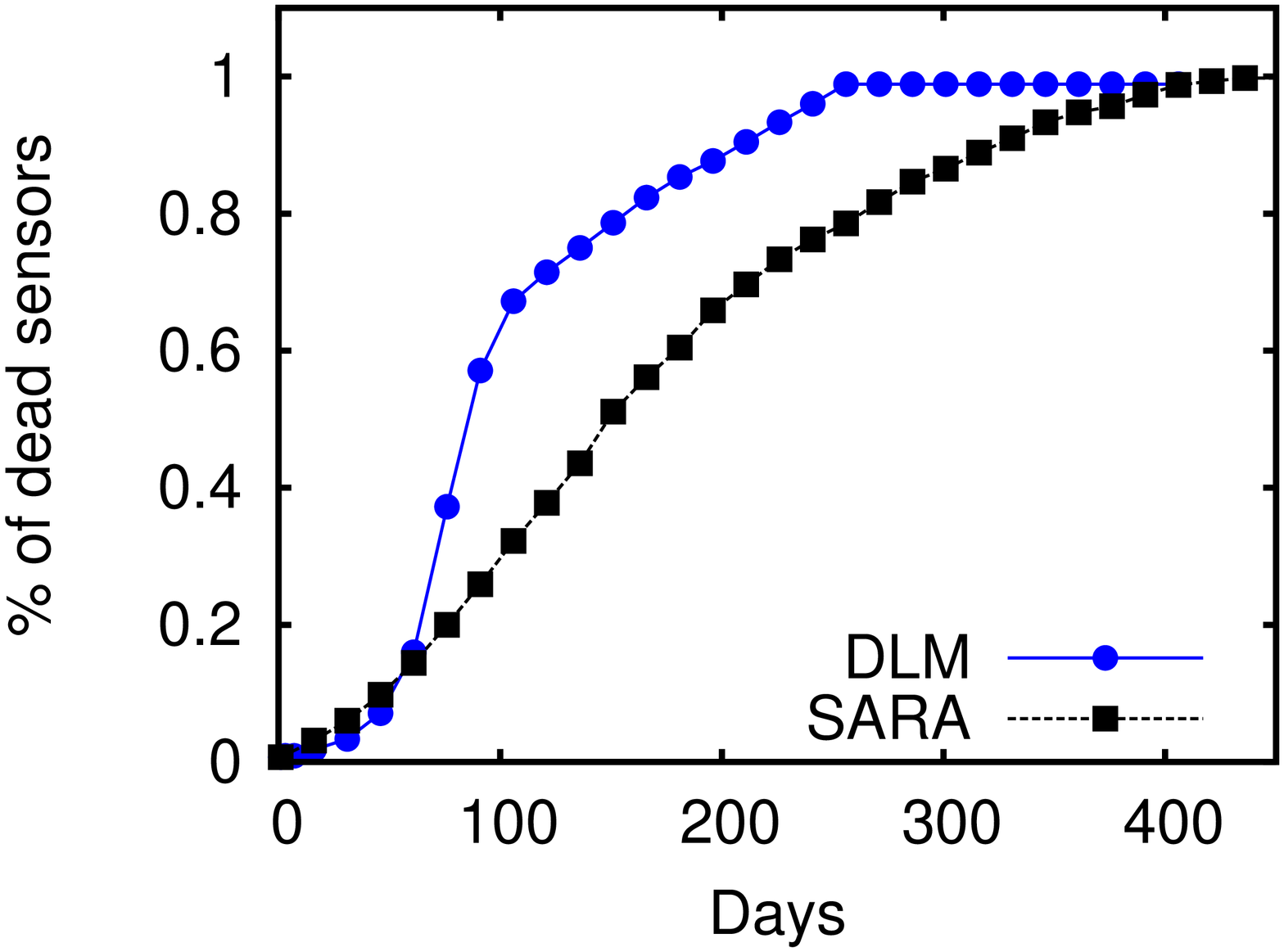}\\
\hspace{-0.5cm}{\footnotesize{(a)}}&{\footnotesize{(b)}}&{\footnotesize{(c)}}\\
\hspace{-0.5cm}
\includegraphics[width = 0.22\textwidth,angle=-90]{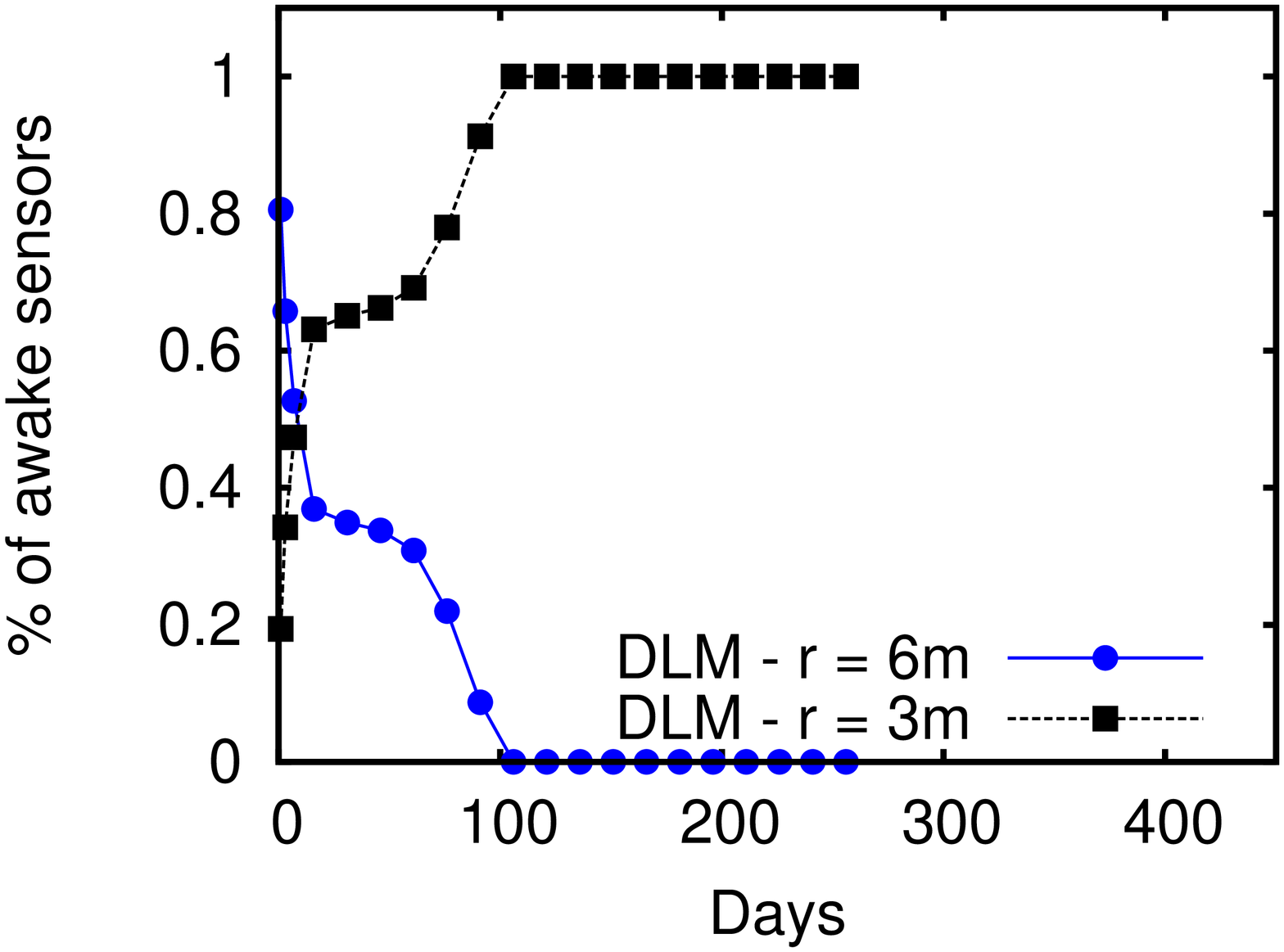}
&
\includegraphics[width = 0.22\textwidth,angle=-90]{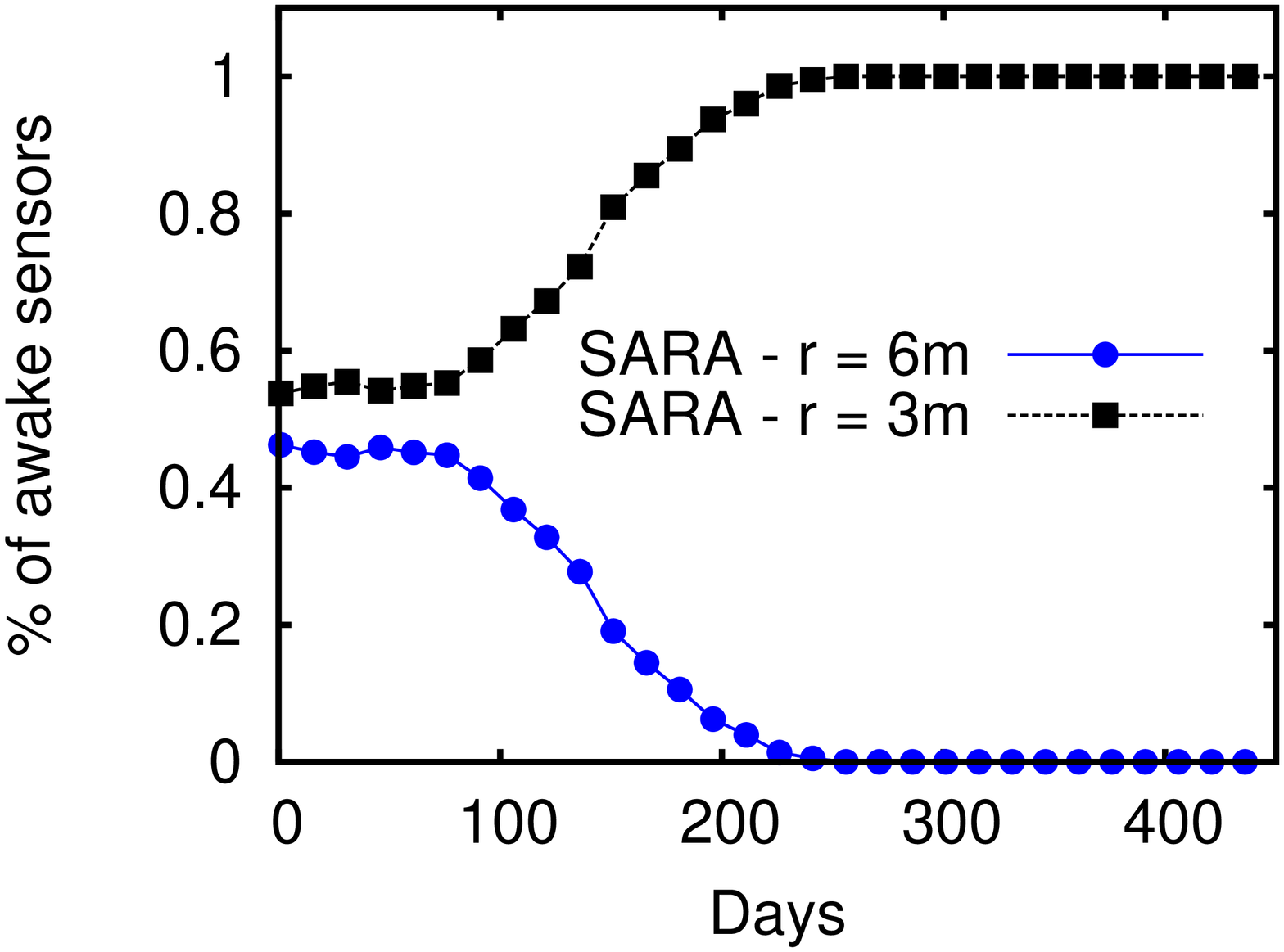}
  &
\includegraphics[width = 0.22\textwidth,angle=-90]{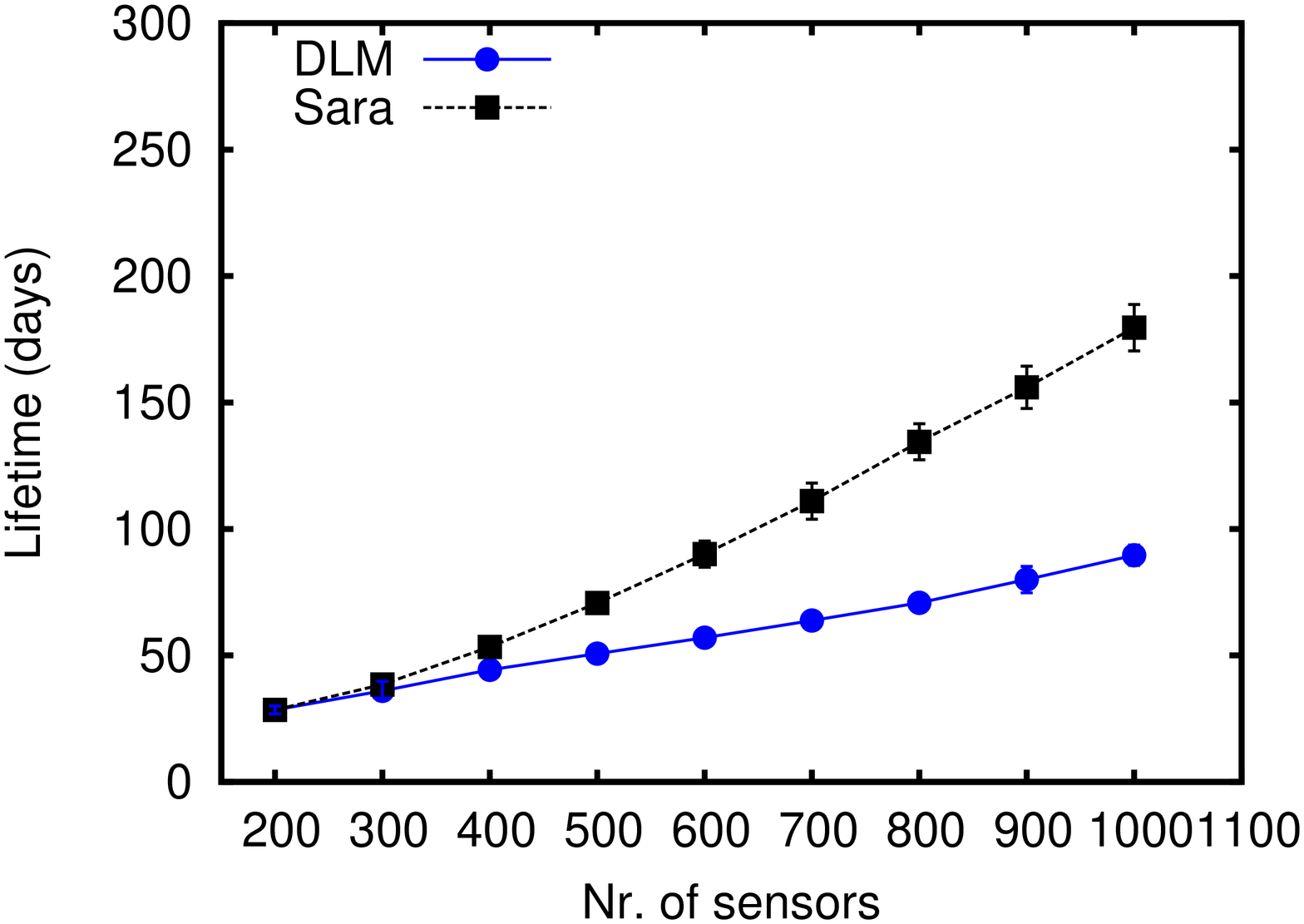}
    \\
\hspace{-0.5cm}{\footnotesize{(d)}}&{\footnotesize{(e)}}&{\footnotesize{(f)}}\\
\end{tabular}
 \caption{Fixed sensors, heterogeneous setting. Scenario with 900 sensors. Comparative analysis of $\alg$ and $\dlm$. Percentage of AoI covered (a),
percentage of awake sensors (b),  percentage of dead sensors (c).
 Composition of the set of awake sensors under $\dlm$ (d) and $\alg$ (e).
Network lifetime (f).
 }
\label{fig:time_FIX_HET}
\end{figure}

\begin{figure}
\begin{center}
\begin{tabular}[c]{ccc}
\hspace{-0.5cm}
  \includegraphics[width = 0.22\textwidth,angle=-90]{fixed/lifetime/lifetime_80}
  &{\includegraphics[width=0.22\textwidth,angle=-90]{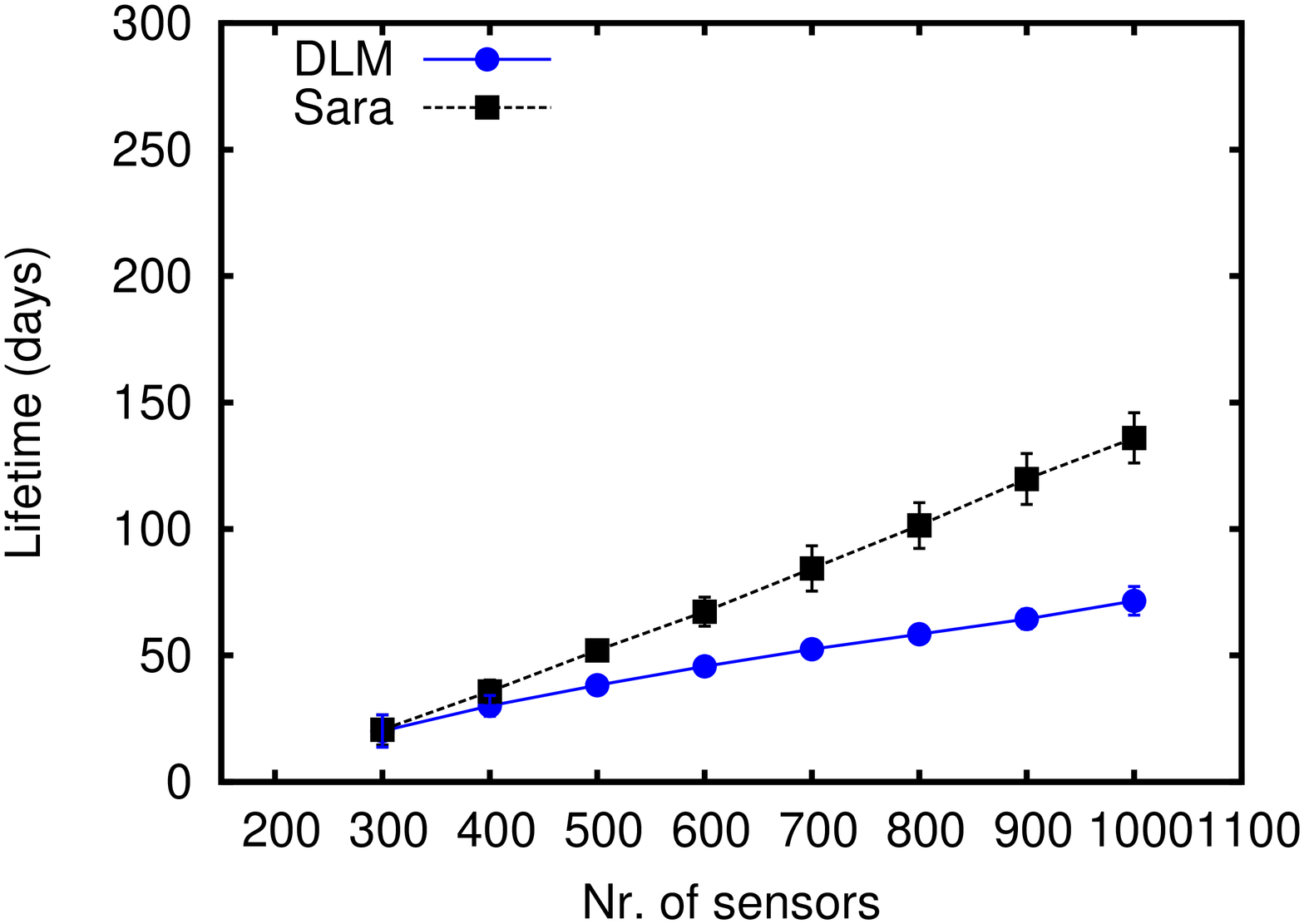}}
  &  {\includegraphics[width=0.22\textwidth,angle=-90]{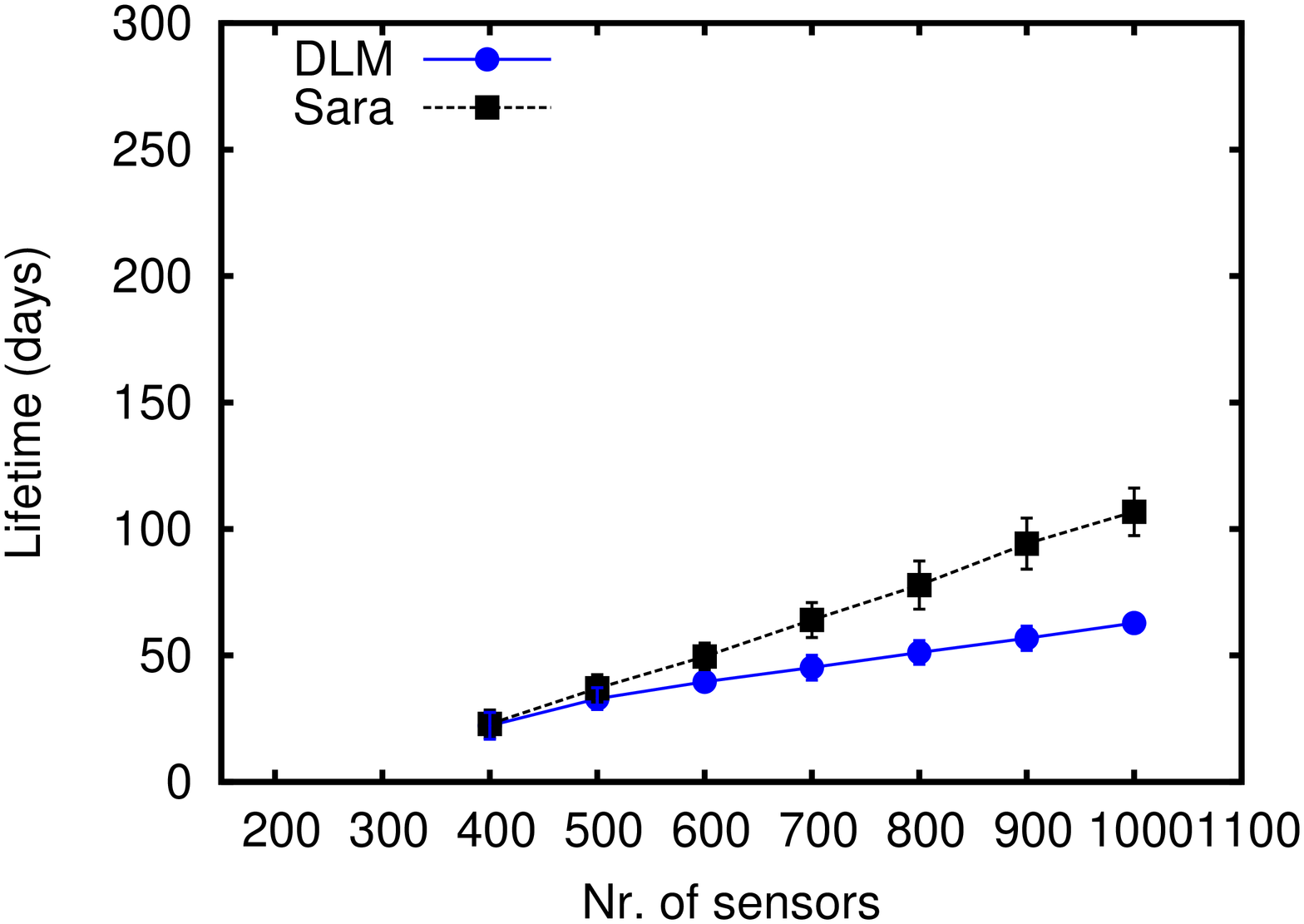}}
    \\
{\footnotesize{(a)}}&{\footnotesize{(b)}}&{\footnotesize{(c)}}
\end{tabular}
\end{center}
 \caption{{Fixed sensors: heterogeneous setting. Lifetime achieved by the three algorithms expressed as the time after which the algorithm is no longer capable to 
 cover more that 80\% (a), 90\% (b) and 95\% (c) of the AoI.}}
\label{fig:lifetimeFIX_HOM}
 \end{figure}
 
 In Figure \ref{fig:lifetimeFIX_HOM}
 we compare the algorithms $\alg$ and $\dlm$ in terms of network lifetime by increasing the number of deployed sensors. We consider the time at which the coverage of the AoI goes  below the 80\% (a), 90\% (b) and 95\%(c).
Notice that, in Figure \ref{fig:lifetimeFIX_HOM}(c) the point corresponding to the deployment of 300 sensors is missing, because even if all the sensors were kept awake, this amount of
sensors would not be sufficient to cover  the 95\% of the AoI.  
Even in this case, although our algorithm does not specifically address a particular notion of lifetime, it outperforms $\dlm$ also under other possible lifetime requirements.


\subsection{Mixed sensors: homogeneous setting}

We consider the most general applicative scenario, with sensors belonging to both classes of fixed and adjustable sensing radius,
we refer to a scenario with 900 uniformly deployed sensors.
The set of available sensors is composed by 50\% of fixed sensors with sensing radius equal to 6m
and 50\% of adjustable sensors with a sensing radius which varies 
in the interval [2m,6m].
Notices that this scenario is considered homogeneous because all the sensors are able to reach the same maximum extension of the
sensing range, no matter which class they belong to.

Figure \ref{fig:time_MIX_HOM}(a) shows the percentage of the AoI that the algorithm $\alg$ is able to cover  
as the time increases.
The figure also shows the percentage of the AoI that is covered by the only sensors with adjustable radius, and by those with fixed radius separately.
It is worth noting that at the first operative time intervals $\alg$ privileges the sensors with adjustable range in the active set, as also detailed in Figure 
 \ref{fig:time_MIX_HOM}(b). 
This is due to the 
higher flexibility of the solution that can be obtained using this class of devices. 
As time progresses, the adjustable sensors that have been used extensively in the previous intervals begin to deplete their energy, thence $\alg$ requires more fixed sensors to be included in the active set.
It should be noted also that, see Figure \ref{fig:time_MIX_HOM}(a), the percentage of dead sensors is about the same for the sensors of the two classes.
This is due to the fact that, in this homogeneous setting, as long as a fixed sensor is activated, it consumes energy at the same rate of an adjustable sensor working at maximum sensing range. 
While this behavior was expected in the case of the algorithms $\dlm$ and $\gup$, as they do not distinguish the two classes, this has to
be considered a nice property for the algorithm $\alg$ as it evidences its capability to do the best with the two classes, exploiting their energy when possible
in an equal manner.
The Figure \ref{fig:time_MIX_HOM}(d) shows the composition of the set of sleeping sensors, which is a complement to the values of figure (b) and (c).

\begin{figure}[h]
\centering
\begin{tabular}{ccc}
\hspace{-0.5cm}\includegraphics[width = 0.22\textwidth, angle=-90]{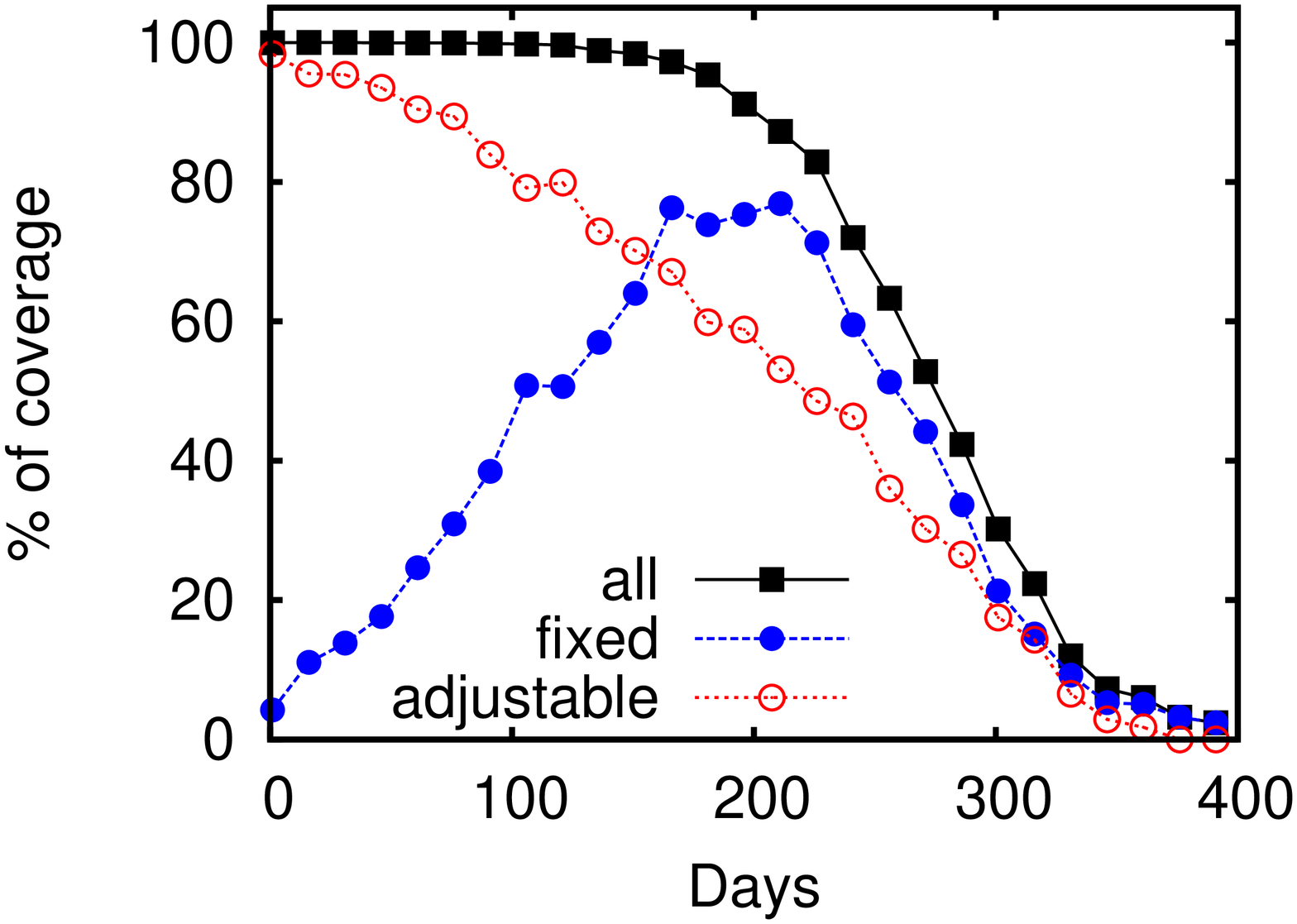} &
\includegraphics[width = 0.22\textwidth,angle=-90]{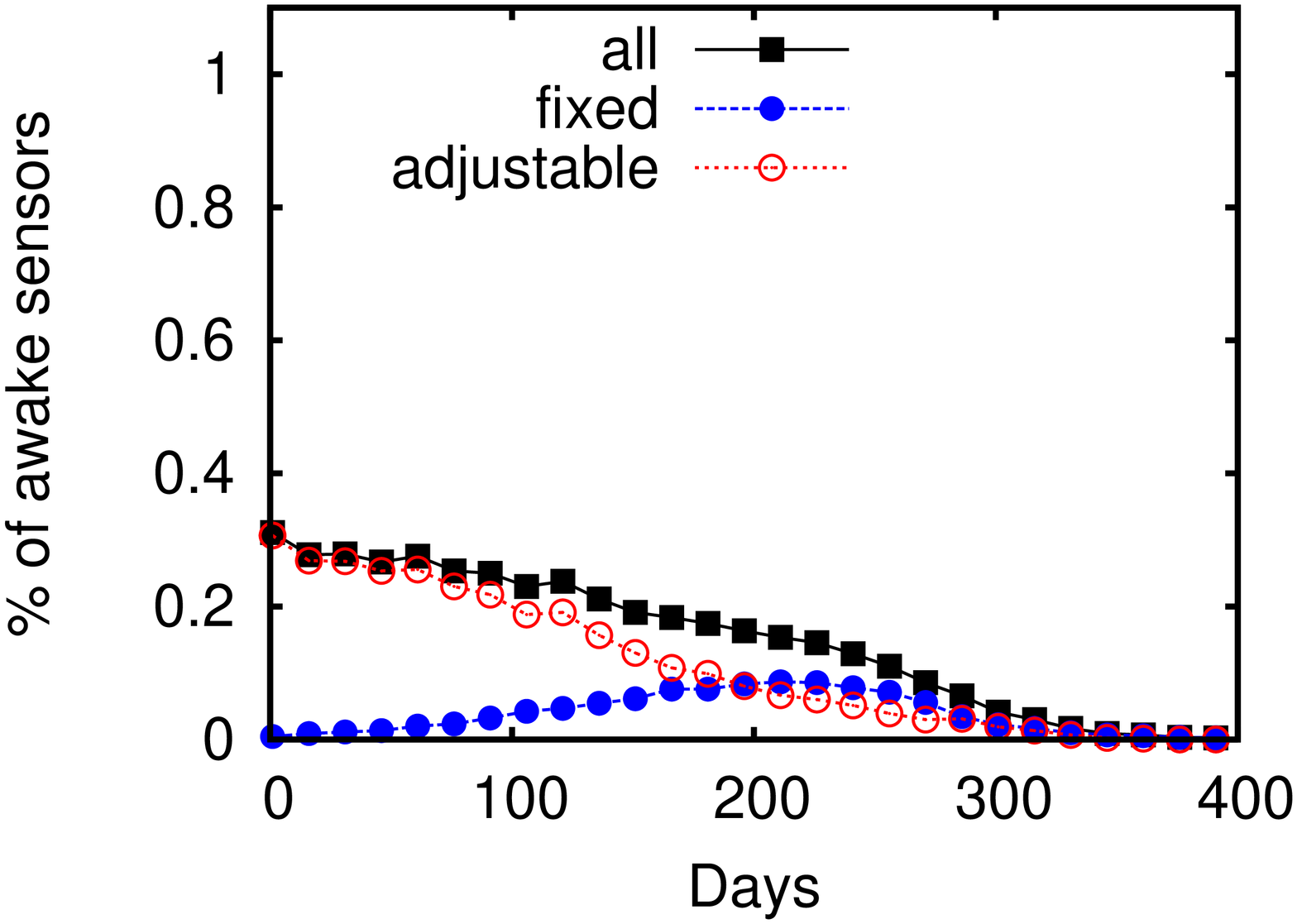}&
\includegraphics[width = 0.22\textwidth,angle=-90]{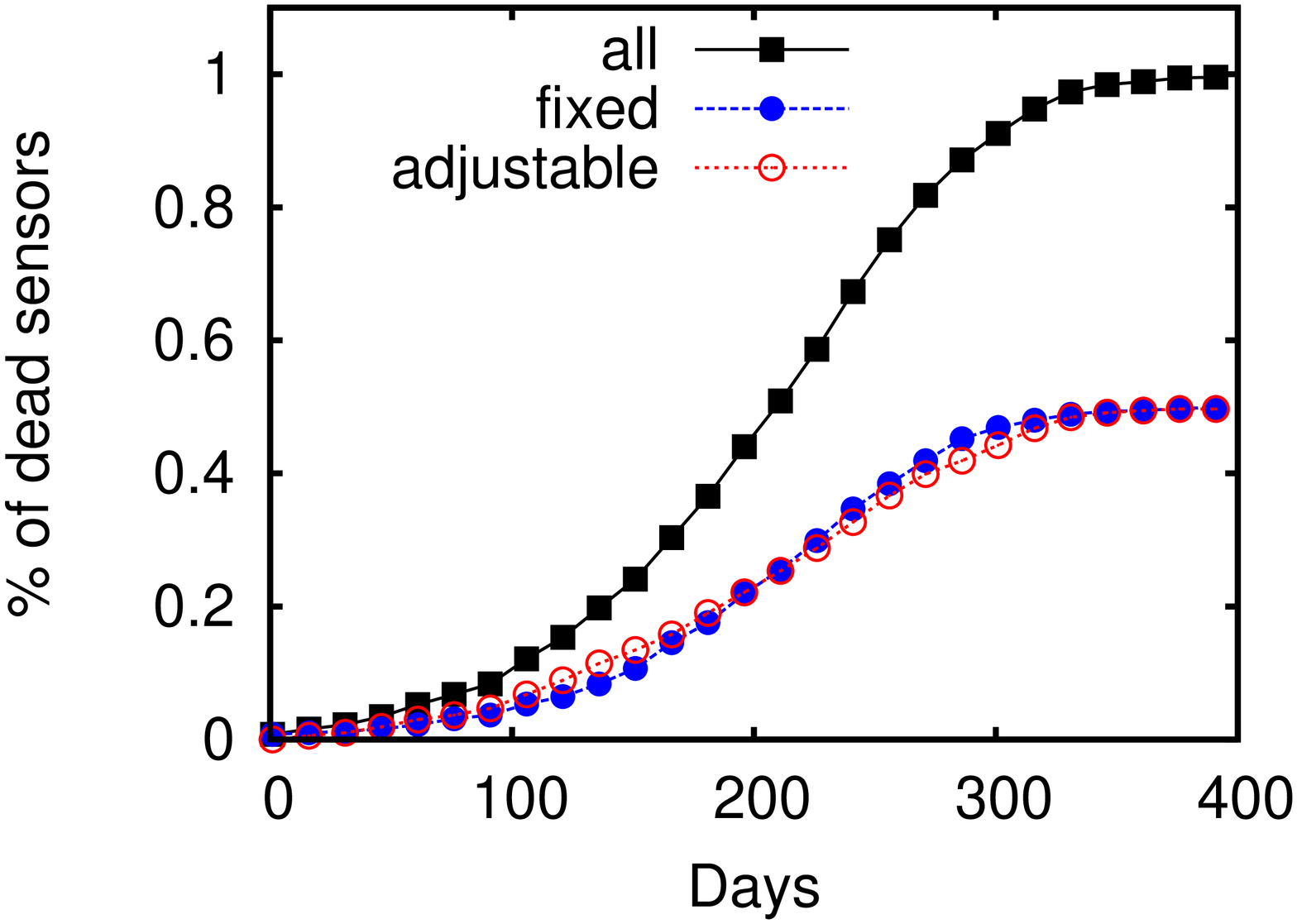}
\\
\hspace{-0.5cm}{\footnotesize{(a)}}&{\footnotesize{(b)}}&{\footnotesize{(c)}}\\
\hspace{-0.5cm}
\includegraphics[width = 0.22\textwidth,angle=-90]{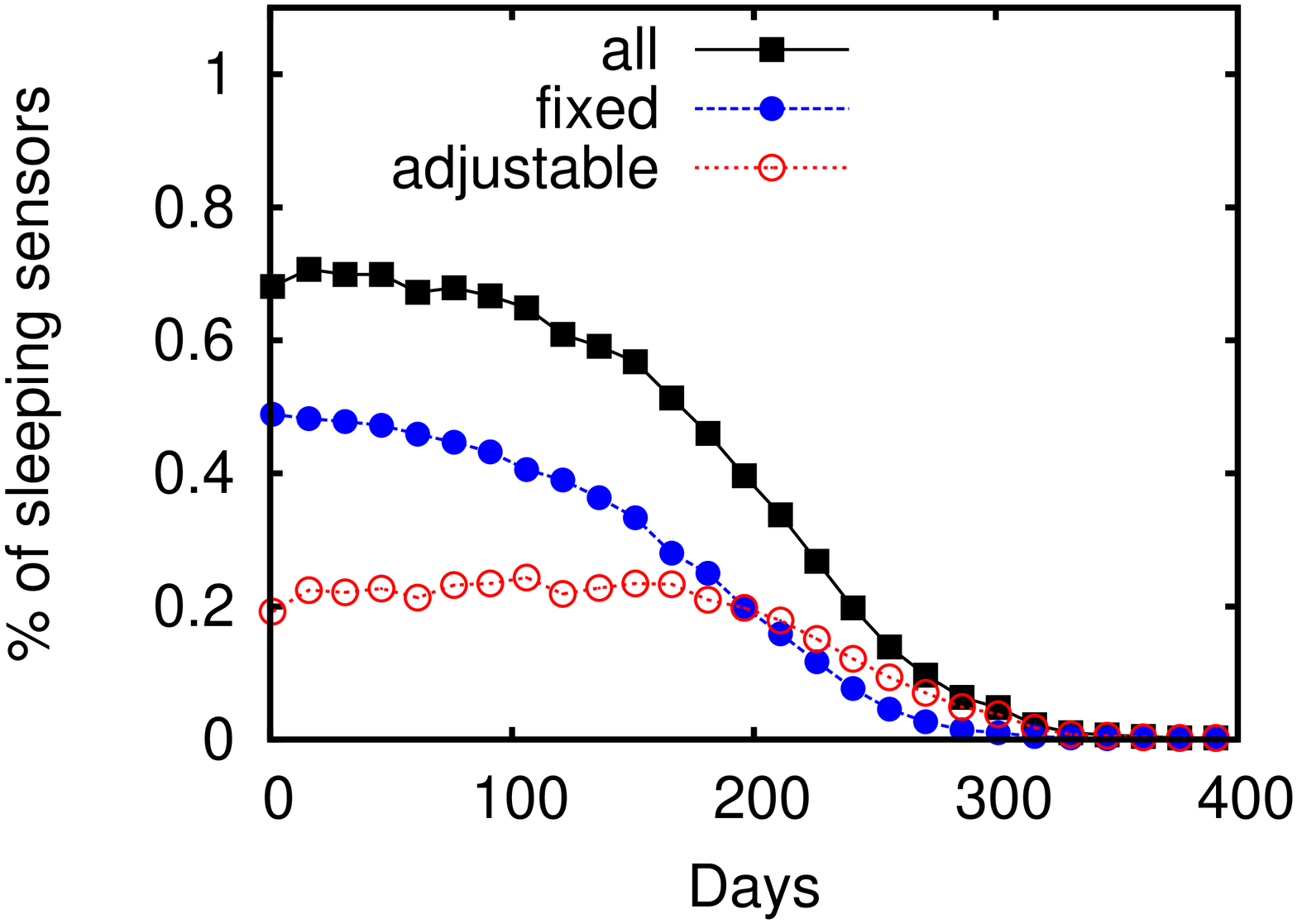}
  &
\includegraphics[width = 0.22\textwidth,angle=-90]{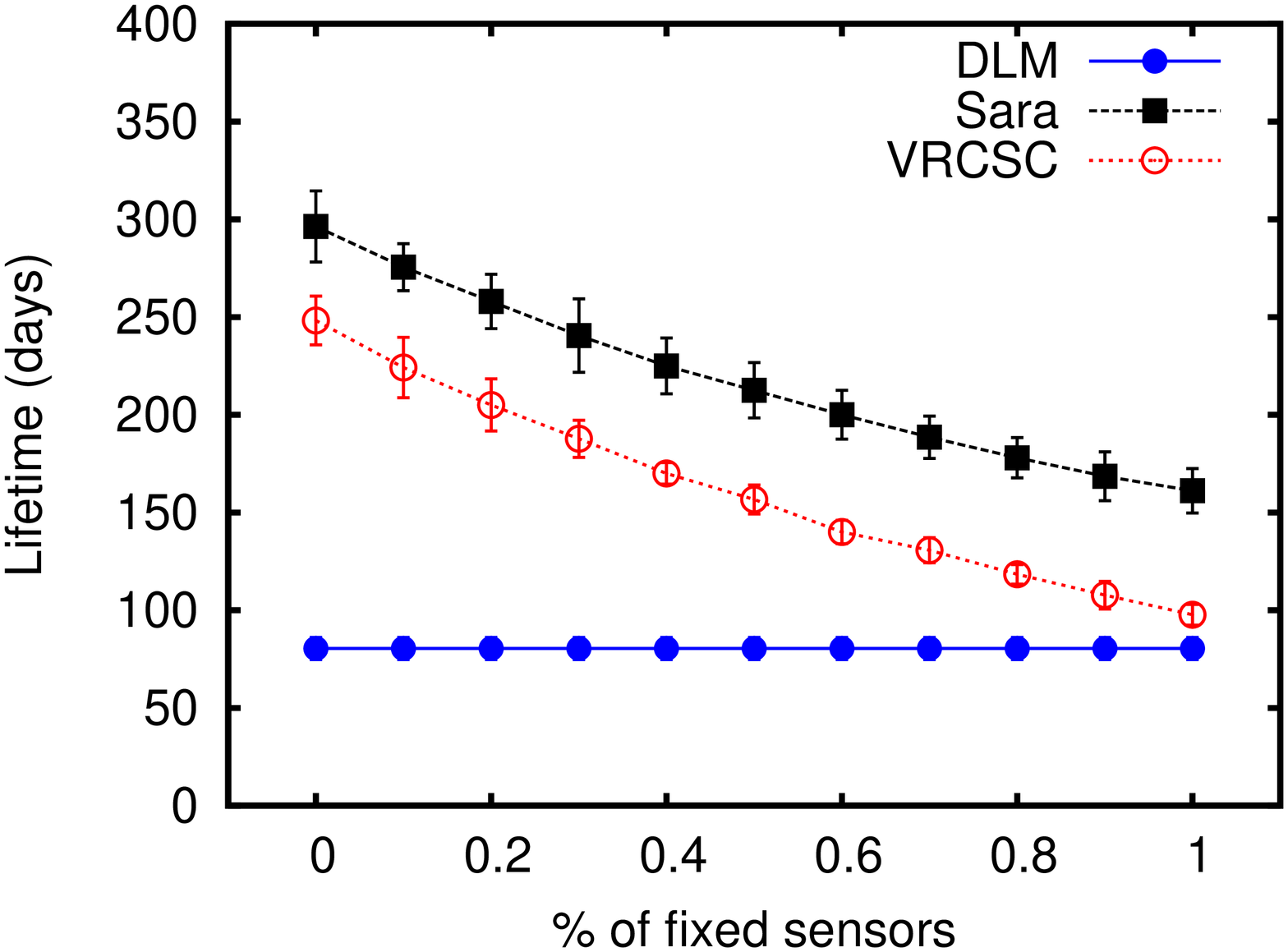}&
\includegraphics[width = 0.22\textwidth,angle=-90]{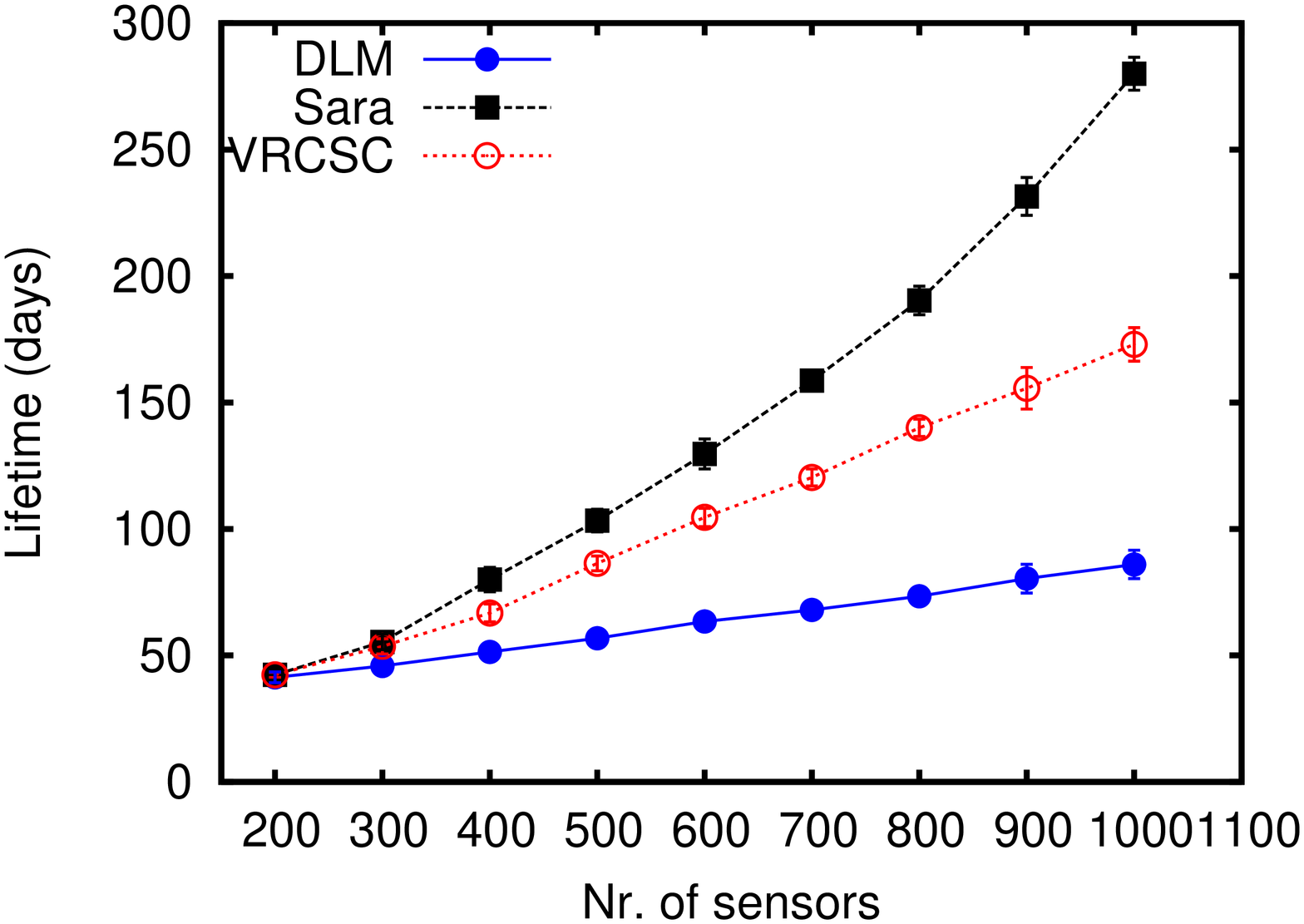}  
    \\
\hspace{-0.5cm}{\footnotesize{(d)}}&{\footnotesize{(e)}}&{\footnotesize{(f)}}\\
\end{tabular}
 \caption{Mixed sensors: homogeneous scenario.
 The maximum radius of adjustable sensors and the radius of fixed sensors are 6m. 
 Case of 900 sensors, 50 \% fixed and 50 \% 
adjustable range: coverage (a), 
active (b),  dead (c), and sleeping (d) sensors.
Lifetime of the network by varying the percentage of  fixed sensors with respect to total (e).
Lifetime of the network by varying the number of available sensors (f). 
}
\label{fig:time_MIX_HOM}
\end{figure}

We now comparatively analize the behavior of the algorithms $\gup$ and $\dlm$ with respect to $\alg$.

Thanks to the device homogeneity, the algorithm $\gup$ in its modified version introduced in Section \ref{sec:two_approaches},
 is able to work in this scenario without creating coverage holes. 
 Nevertheless, the presence of fixed sensors in the available set, compromises the capability of $\gup$ to correctly determine the maximum extent
 of the radius reduction to be adopted by sensors with adjustable range.

The algorithm $\dlm$, in its modified version instead does not find more difficulties when dealing with this scenario than 
those encountered in addressing the case of only adjustable sensors, as it treats every sensor as if it were fixed.

The Figure  \ref{fig:time_MIX_HOM}(e) 
shows the lifetime of the network when the percentage of fixed sensors in the available set increases.
Not surprisingly the performance of $\dlm$ is not affected by this increase as it treats the two classes alike.
Both $\gup$ and $\alg$ show a decreasing behavior of the network lifetime due to the decreasing flexibility of the network.
Indeed it is intuitive that by increasing the percentage of fixed sensors, in the homogeneous case, 
we are significantly reducing the set of possible solutions that can be reached by any algorithm.
Nevertheless $\alg$ is less affected by this phenomenon as it can exploit the capability of the two classes of sensors 
with more specifically tailored decisions.

 All the above considerations justify the significant improvement in terms  of lifetime achieved by $\alg$ with respect to $\dlm$ and $\gup$.
 To highlight this difference we now consider an experiment conducted by varying the number of available sensors, with a set composed of 50\% of sensors
 with adjustable radius, and 50\% with fixed radius. The Figure \ref{fig:time_MIX_HOM}(e) illustrates the behavior of the three algorithms in this setting.
{\color{black}For instance, when the number of sensors is 1000, $\alg$ achieves a lifetime of 280 days, whereas $\gup$ reaches 170 days, and $\dlm$ only 80 days.
}

\begin{figure}
\begin{center}
\begin{tabular}[c]{ccc}
\hspace{-0.5cm}
  \includegraphics[width = 0.22\textwidth,angle=-90]{mixed/homogeneous/lifetime/50-50/lifetime_80}
  &{\includegraphics[width=0.22\textwidth,angle=-90]{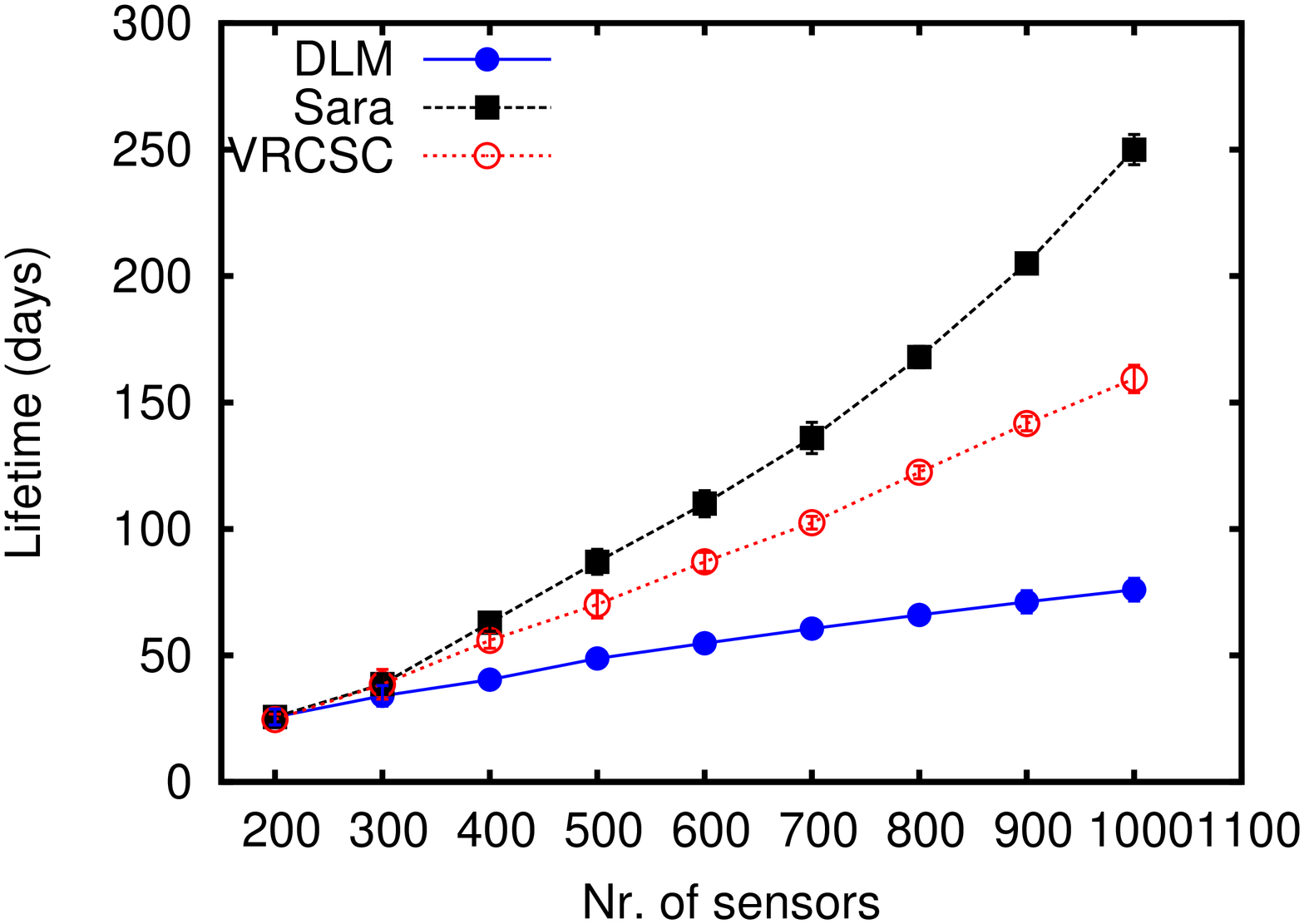}}
  &  {\includegraphics[width=0.22\textwidth,angle=-90]{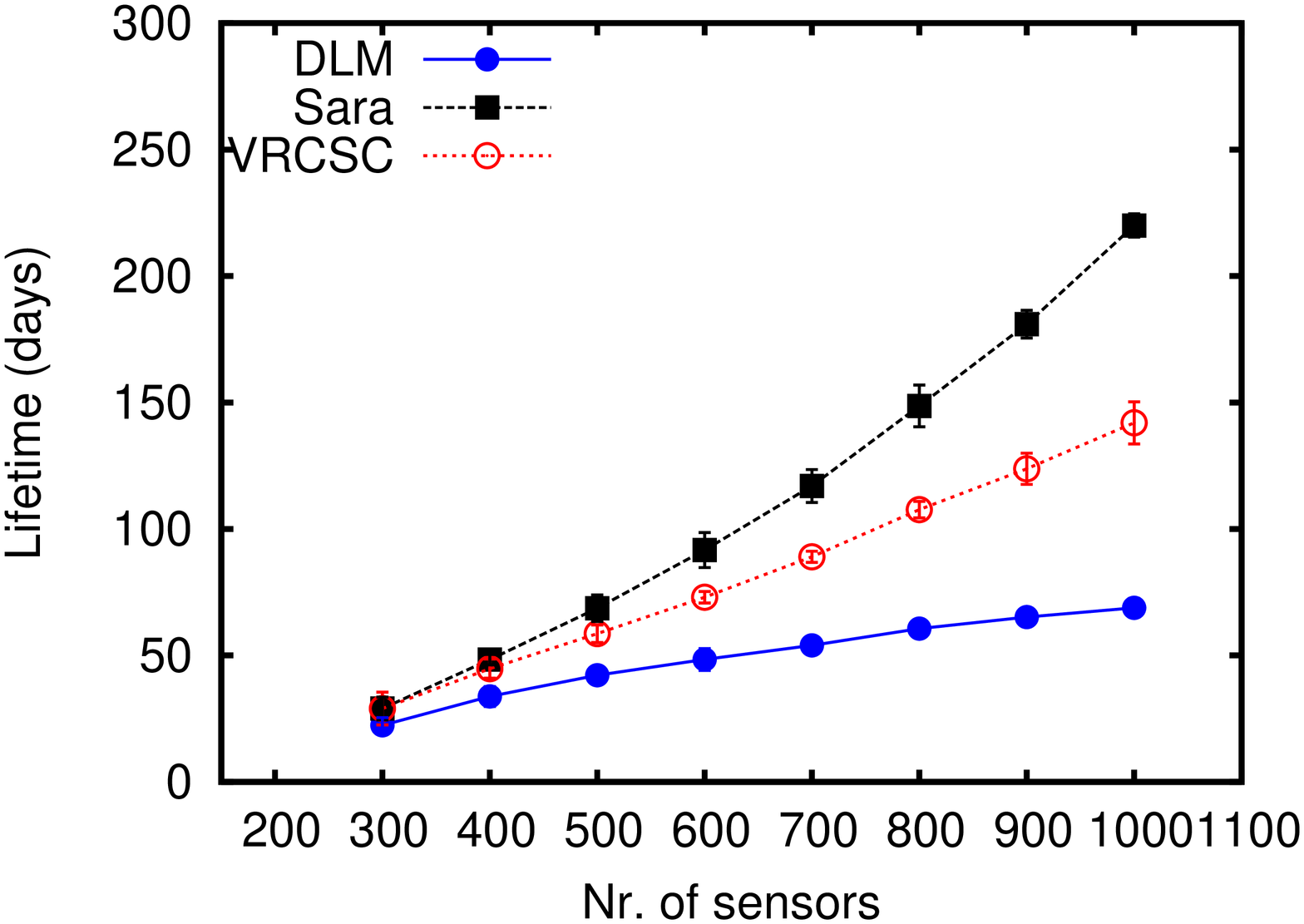}}
    \\
{\footnotesize{(a)}}&{\footnotesize{(b)}}&{\footnotesize{(c)}}
\end{tabular}
\end{center}
 \caption{{Mixed sensors: homogeneous setting. Lifetime achieved by the three algorithms expressed as the time after which the algorithm is no longer capable to 
 cover more that 80\% (a), 90\% (b) and 95\% (c) of the AoI.}}
\label{fig:lifetimeMIX_HOM}
 \end{figure}
 
 In Figure \ref{fig:lifetimeMIX_HOM}
 we compare the algorithms $\alg$, $\gup$ and $\dlm$ in terms of network lifetime by increasing the number of deployed sensors. We consider the time at which the coverage of the AoI goes  below the 80\% (a), 90\% (b) and 95\%(c).
Even in this case, although $\alg$ does not specifically address a particular notion of lifetime, it outperforms the other two also under other possible lifetime requirements.

\begin{figure}
\begin{center}
\begin{tabular}[c]{ccc}
\hspace{-0.5cm}
  \includegraphics[width = 0.22\textwidth,angle=-90]{mixed/homogeneous/lifetime/perc_fixed/lifetime_80}
  &{\includegraphics[width=0.22\textwidth,angle=-90]{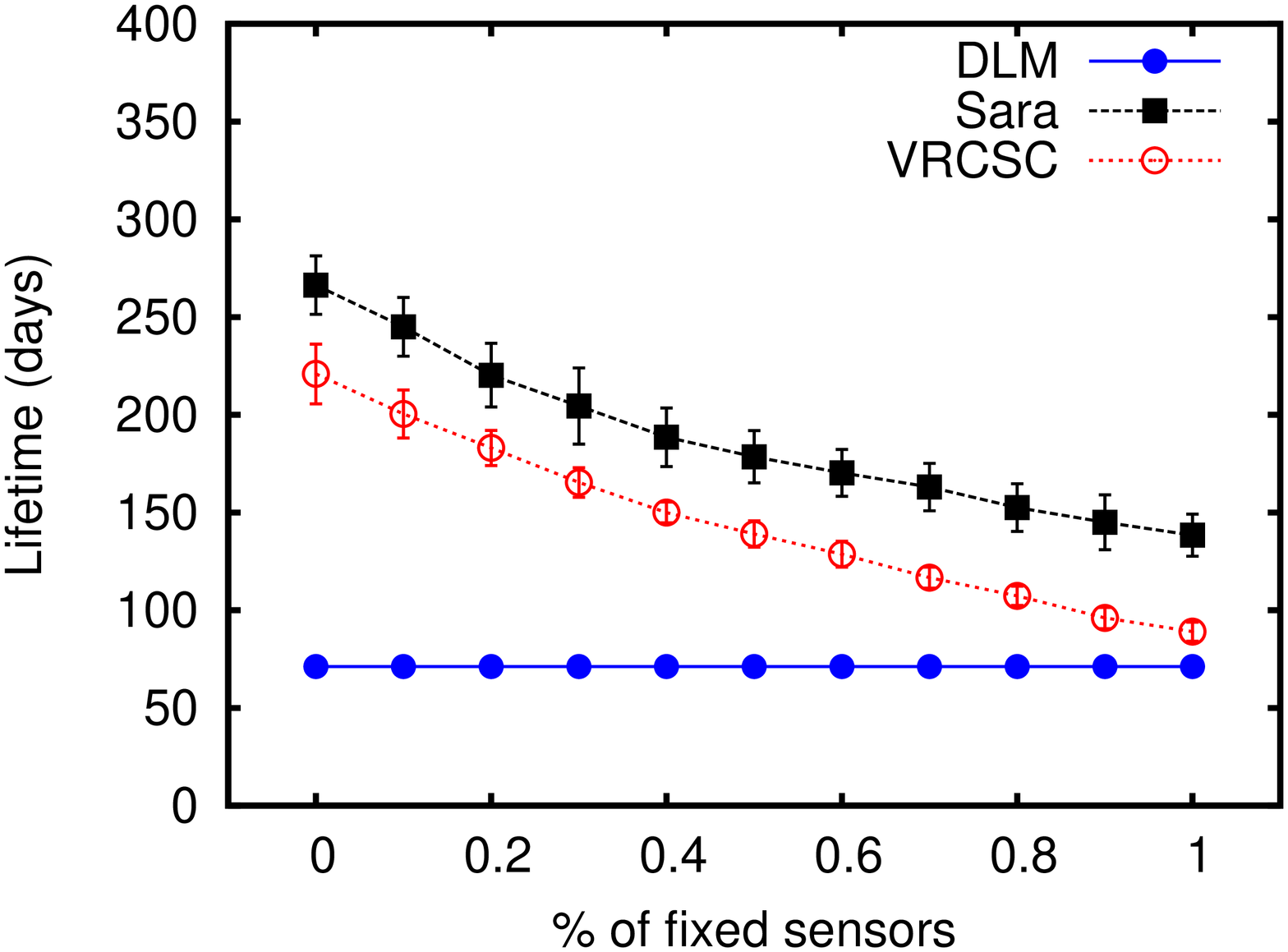}}
  &  {\includegraphics[width=0.22\textwidth,angle=-90]{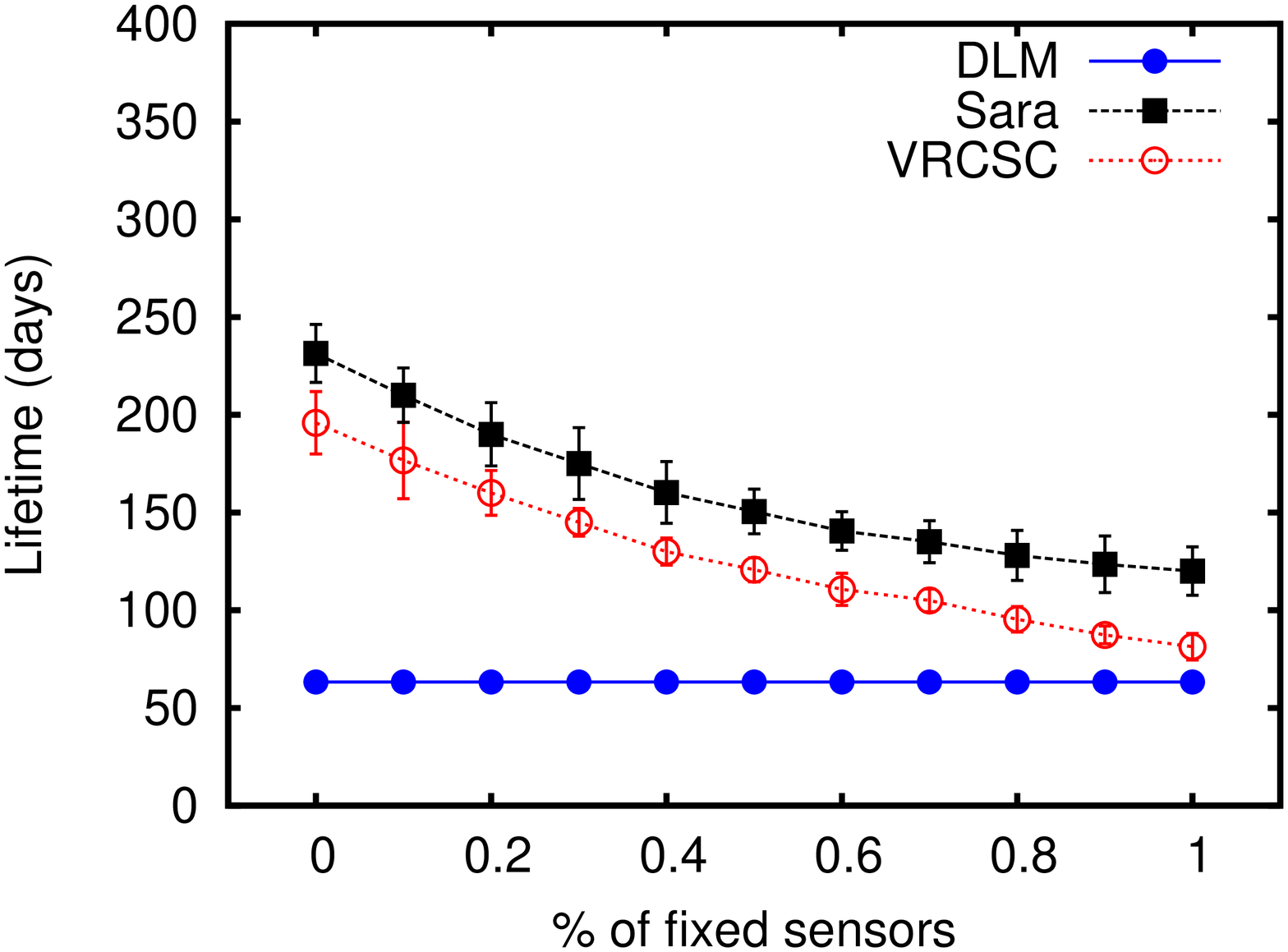}}
    \\
{\footnotesize{(a)}}&{\footnotesize{(b)}}&{\footnotesize{(c)}}
\end{tabular}
\end{center}
 \caption{Mixed sensors: homogeneous setting. Lifetime of the network by varying the percentage of  fixed sensors with respect to total. Lifetime achieved by the three algorithms expressed as the time after which the algorithm is no longer capable to  cover more that 80\% (a), 90\% (b) and 95\% (c) of the AoI.}
\label{fig:lifetimeMIX_HOM_PERC_FIX}
 \end{figure}
 Figure \ref{fig:lifetimeMIX_HOM_PERC_FIX} shows the performance of the three algorithms by varying the percentage of fixed sensors in the available set when 900 sensors are deployed. The figure highlights that even under other possible definitions of lifetime $\alg$ outperforms $\dlm$ and $\gup$.


\subsection{Mixed sensors: heterogeneous setting}
In this latter subsection, we consider the most general applicative scenario, where sensors belong to both classes and where the two classes have heterogeneous sensing capabilities.
In particular, the radius of fixed sensors is 3m long, while the radius of adjustable sensors varies in the interval [2m, 6m].

The qualitative analysis of the results shown in Figures
\ref{fig:time_MIX_HET}(a-f) is analogous to the case of the homogeneous setting.
 Nevertheless, Figure \ref{fig:time_MIX_HET}(c) highlights that, in this case, the set of dead sensors is composed by a higher fraction of adjustable sensors with respect to the homogeneous case.
This is due to the fact that we are considering fixed sensors with lower range, that implies for this class a lower energy consumption rate with respect to the
homogeneous case, resulting in a higher residual energy for the fixed sensors, as shown in Figure \ref{fig:time_MIX_HET}(e).
Notice that, as in all the heterogeneous cases treated in this paper, we did not analyze the behavior of $\gup$ in this scenario, as the comparison cannot
be fair, because $\gup$ does not succeed in completing the coverage of the AoI.
Finally, in Figure \ref{fig:time_MIX_HET}(f) we show that, as expected, the lifetime achieved by $\alg$ is significantly longer than under $\dlm$.
{\color{black}For instance, when the number of sensors is 1000, $\alg$ achieves a lifetime of about 750 days, while $\dlm$ is only capable to last no more than 270 days.}

\begin{figure}[h]
\centering
\begin{tabular}{ccc}
\hspace{-0.5cm}\includegraphics[width = 0.22\textwidth, angle=-90]{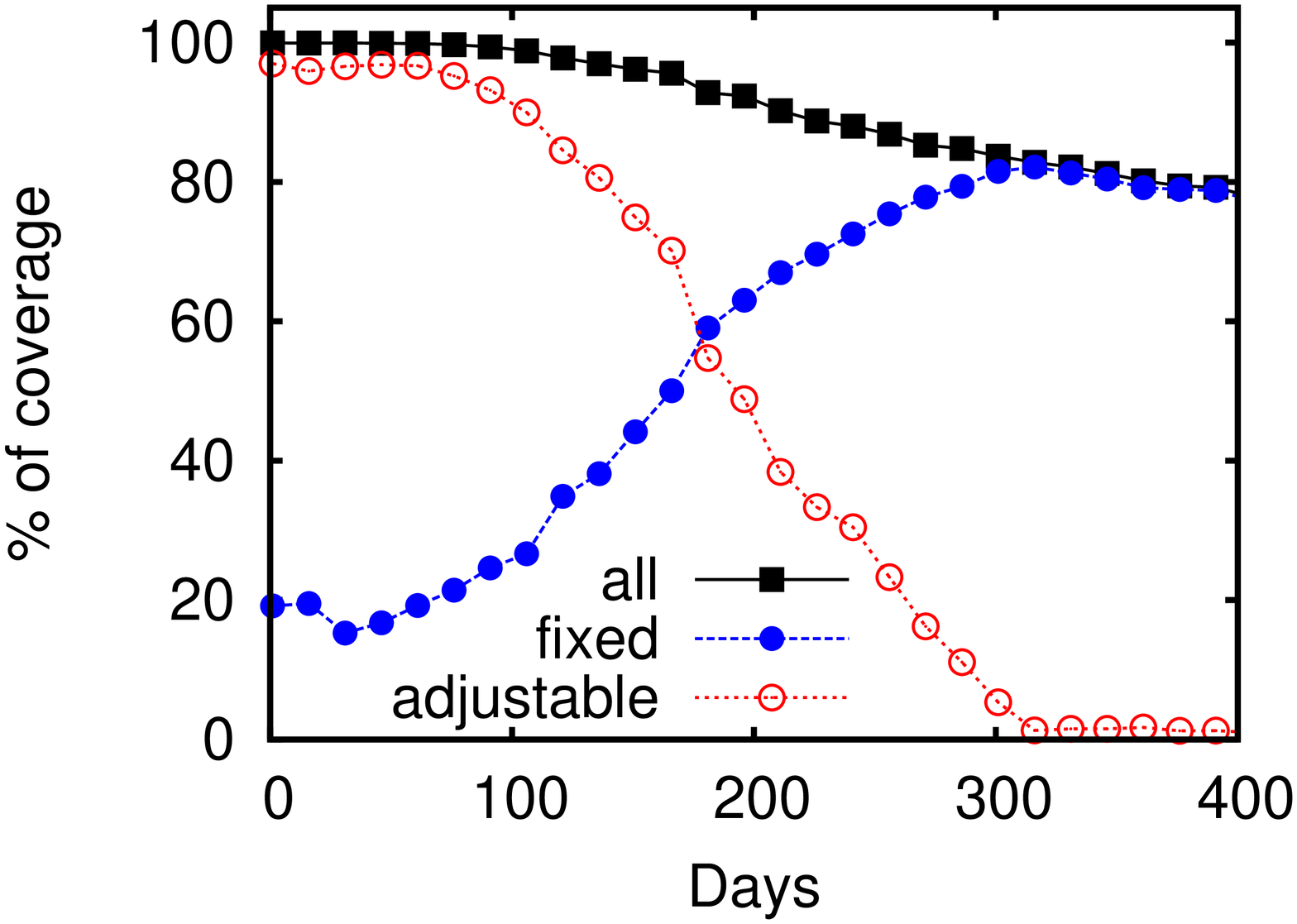} &
\includegraphics[width = 0.22\textwidth,angle=-90]{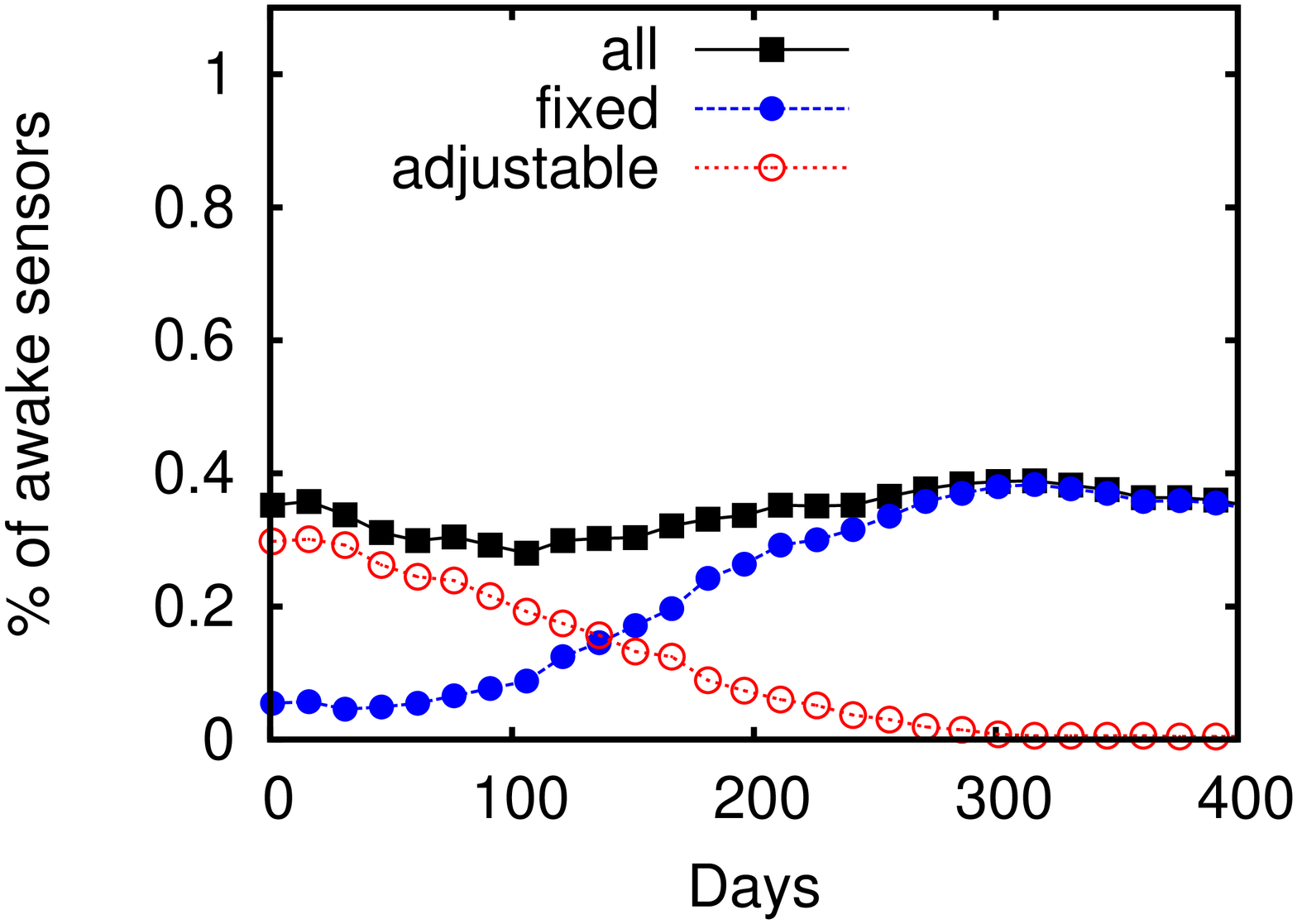}&
\includegraphics[width = 0.22\textwidth,angle=-90]{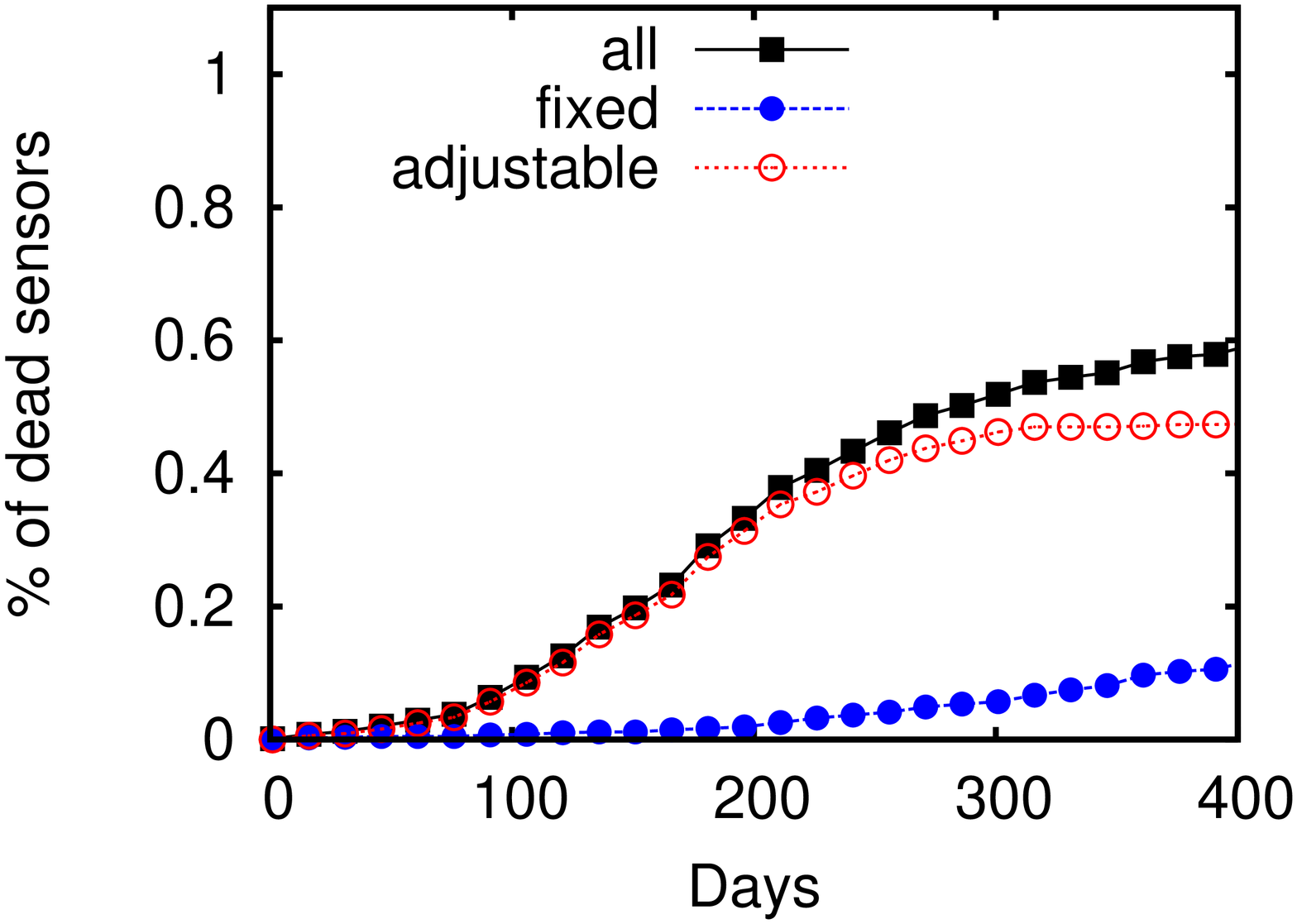}
\\
\hspace{-0.5cm}{\footnotesize{(a)}}&{\footnotesize{(b)}}&{\footnotesize{(c)}}\\
\hspace{-0.5cm}
\includegraphics[width = 0.22\textwidth,angle=-90]{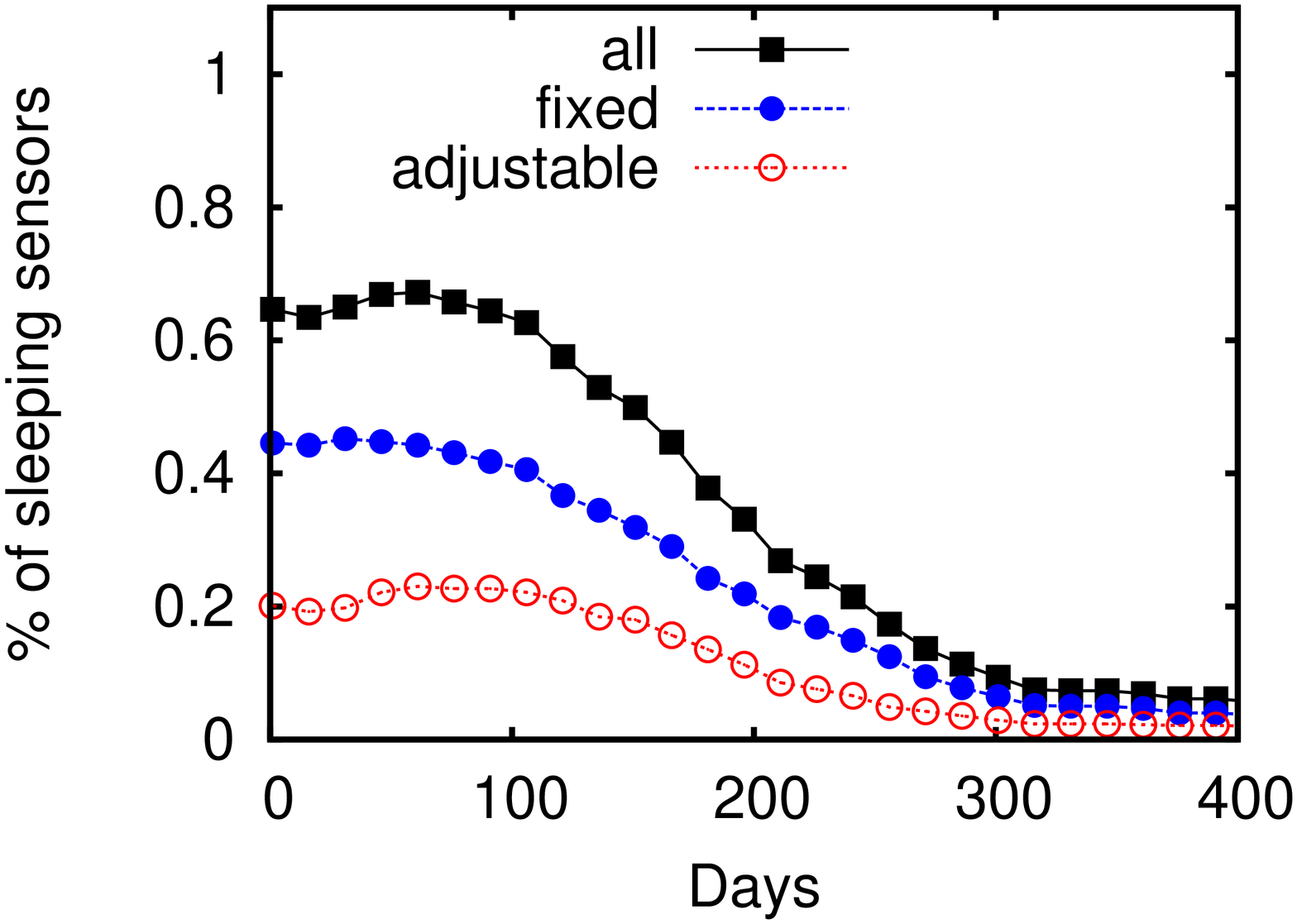}
  &
\includegraphics[width = 0.22\textwidth,angle=-90]{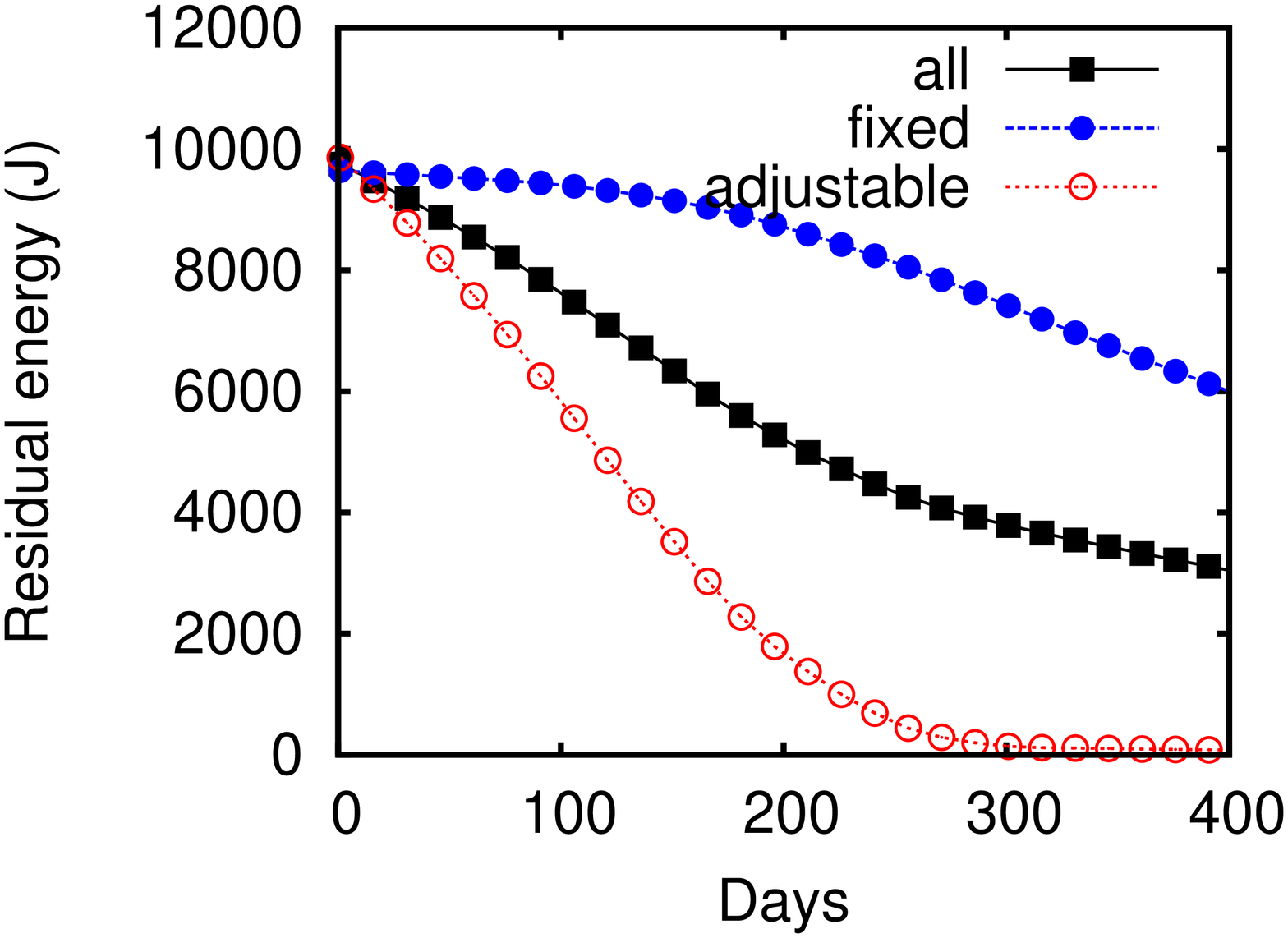} &
\includegraphics[width = 0.22\textwidth,angle=-90]{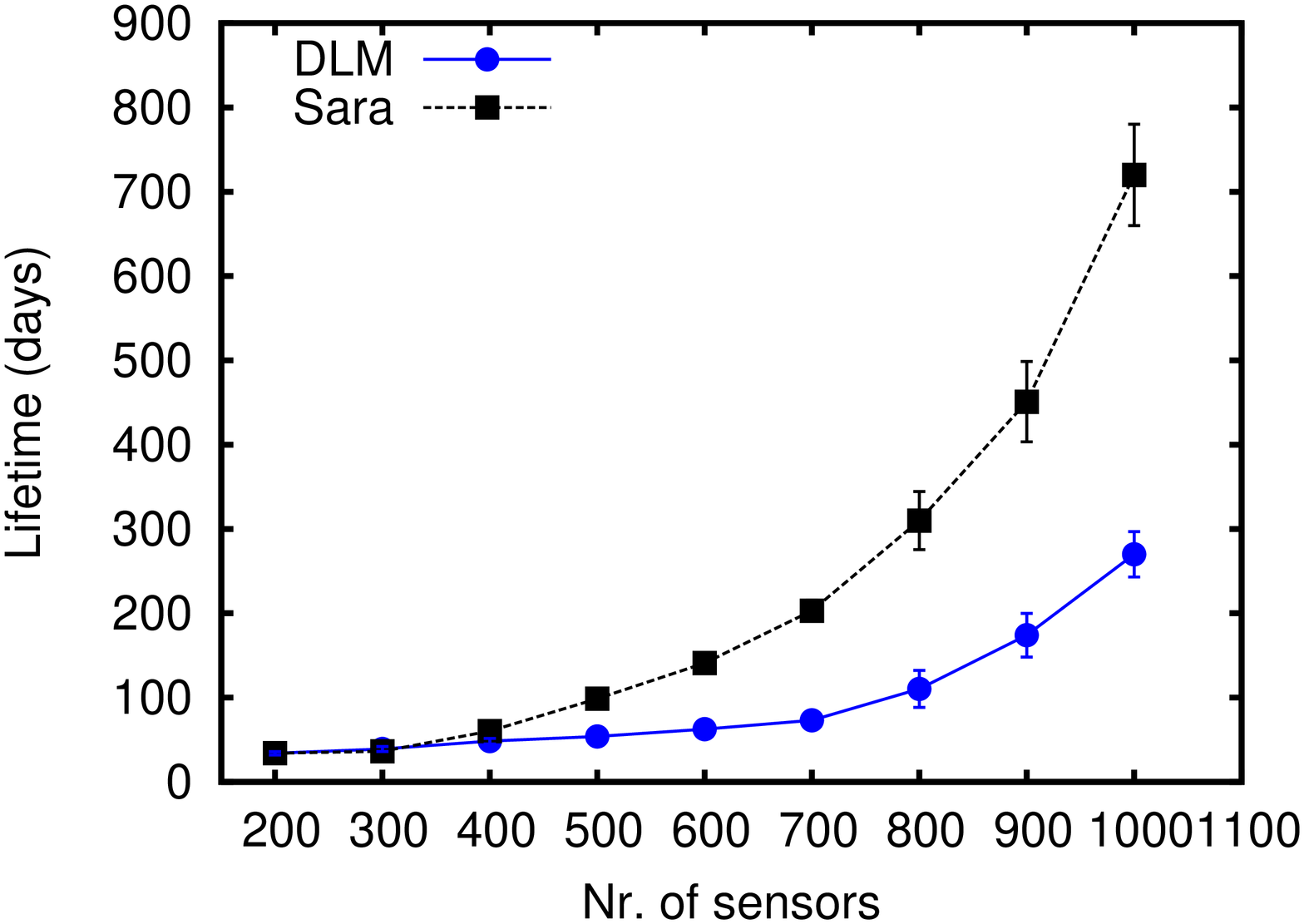} 
    \\
\hspace{-0.5cm}{\footnotesize{(d)}}&{\footnotesize{(e)}}&{\footnotesize{(f)}}\\
\end{tabular}
 \caption{Percentage of coverage (a), 
active (b),  dead (c), sleeping (d) sensors and residual energy (e) in a scenario with 900 {\bf heterogeneously} equipped sensors of {\bf both classes} of devices (50 \% with fixed and 50 \%  with adjustable sensing range).  Lifetime of the network when varying the number of sensors (50 \% of each class).}
\label{fig:time_MIX_HET}
\end{figure}

Notice that in this heterogeneous setting, it does not make sense to analyze the performance of the algorithms when the percentage of the two classes of sensors varies.
This is because the fixed sensors have different  sensing capabilities than the maximum for adjustable sensors.
Therefore, by varying the composition of the mix we would alter the coverage capability of the network.

\begin{figure}
\begin{center}
\begin{tabular}[c]{ccc}
\hspace{-0.5cm}
  \includegraphics[width = 0.22\textwidth,angle=-90]{mixed/heterogeneous/lifetime/50-50/lifetime_80}
  &{\includegraphics[width=0.22\textwidth,angle=-90]{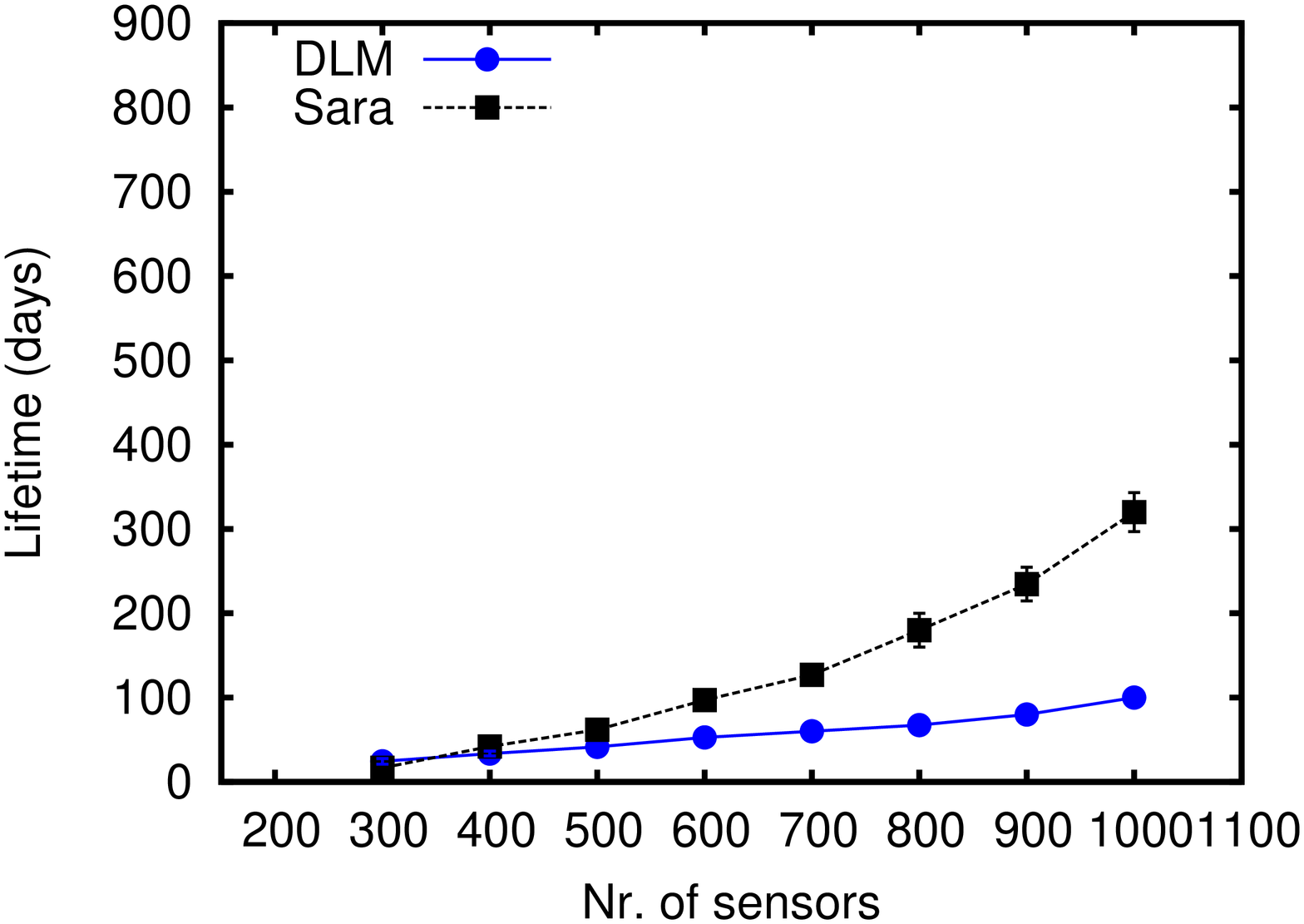}}
  &  {\includegraphics[width=0.22\textwidth,angle=-90]{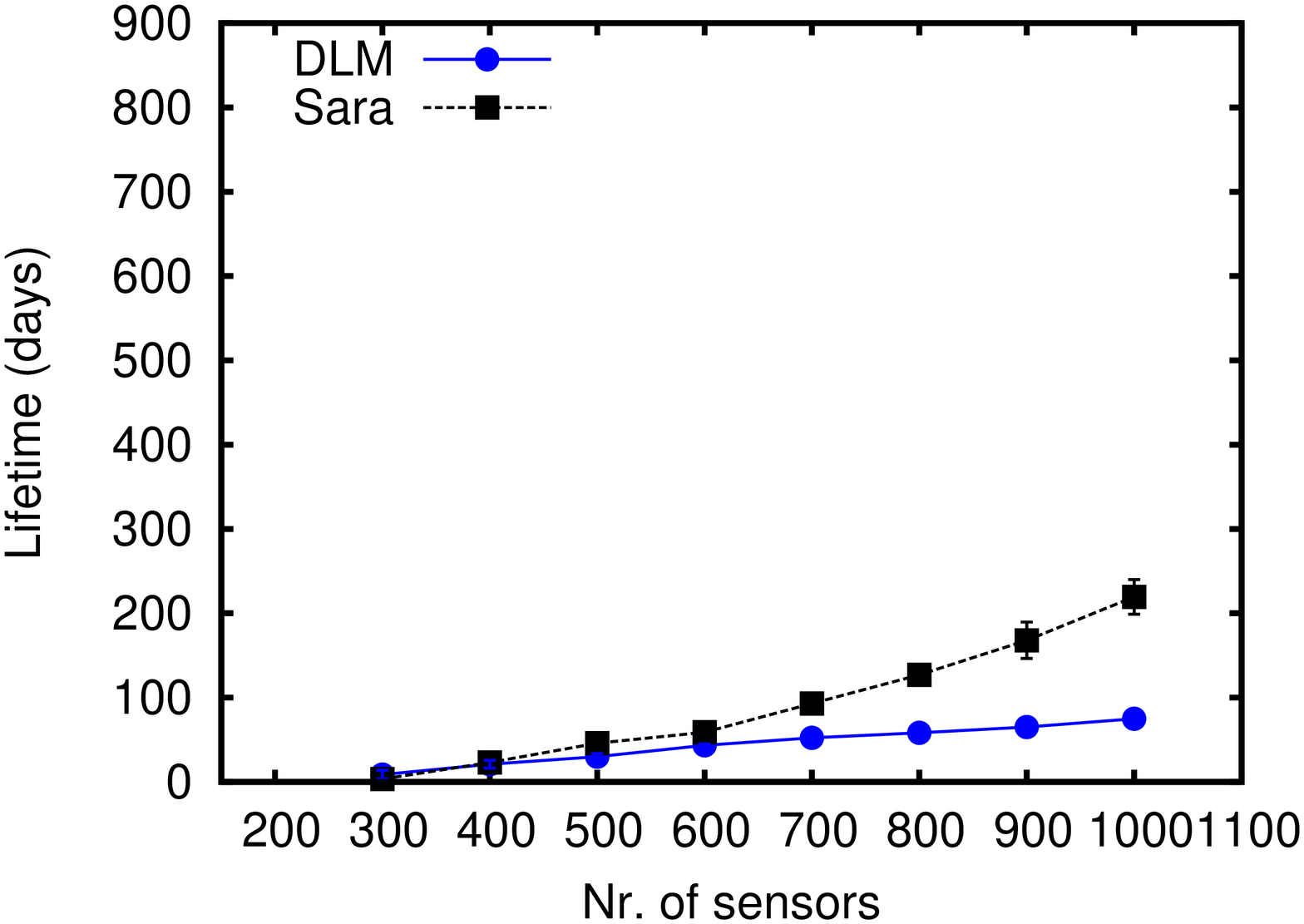}}
    \\
{\footnotesize{(a)}}&{\footnotesize{(b)}}&{\footnotesize{(c)}}
\end{tabular}
\end{center}
 \caption{{Mixed sensors: heterogeneous setting. Lifetime achieved by the three algorithms expressed as the time after which the algorithm is no longer capable to 
 cover more that 80\% (a), 90\% (b) and 95\% (c) of the AoI.}}
\label{fig:lifetimeMIX_HET}
 \end{figure}
 
 In Figure \ref{fig:lifetimeMIX_HET}
 we compare the algorithms $\alg$ and $\dlm$ in terms of network lifetime by increasing the number of deployed sensors. We consider the time at which the coverage of the AoI goes  below the 80\% (a), 90\% (b) and 95\%(c).
Even in this case, although $\alg$ does not specifically address a particular notion of lifetime, it outperforms $\dlm$ also under other possible lifetime requirements.

\section{Related works} \label{sec:related_works}

The problem of exploiting network redundancies to prolong the network lifetime
has been largely investigated in the literature so far.
Depending on the application requirements, the approach to
the problem may vary significantly.
For example, some works only aim at guaranteeing network connectivity,
the SPAN \cite{Chen2002} and ASCENT \cite{Cerpa2004} protocols just to mention the most acknowledged,
without considering coverage issues.
Due to space limitations, in this section we only
consider the works dealing with the problem of completely covering an area of
interest and we refer the reader to the work \cite{TLPsurvey} from Rawaihy et al. for
a survey of sensor scheduling policies in several other applicative
scenarios.

The PEAS protocol proposed by Ye et al. in \cite{YePEAS2003} was designed to address both
coverage and connectivity at the same time.
According to this protocol only a subset of nodes stay awake while the
others are put to sleep.
A sleeping node occasionally wakes up to determine
the presence of coverage holes in its proximity and make activation decisions accordingly.
This approach does not ensure complete coverage, as coverage holes cannot be discovered until a nearby sleeping
sensor wakes up.
Another randomized algorithm is proposed by Xiao et al. in
\cite{Xiao2010}. Different sets of sensors work alternatively according to a probabilistic scheduling.
The authors study the performance of the proposed approach in terms of coverage extension and detection delay.
Differently from the works in \cite{YePEAS2003,Xiao2010}, our approach aims at
ensuring the coverage completeness as long as the available sensors have enough energy.

Xing et al. \cite{Xing2005}  propose the
protocol CPP to achieve $k$-coverage of an area of interest while maintaining the network connectivity..
They  address both coverage and connectivity, and in particular they define
an operative setting in which the  the former implies the latter, namely
when the transmission radius is at least twice the sensing range.
They also provide necessary and sufficient conditions for an area to be $k$-covered.
The authors point out that the network lifetime achieved by their algorithm does not linearly scale
with the number of sensing nodes, due to the higher energy consumption related to periodic beacon messages.
Our work addresses the same operative setting (with $k=1$) with a more aggressive scheduling policy
by resorting to the Laguerre metric space rather than to the Euclidean one, thus allowing a
better scalability.
The geometric analysis made in \cite{Xing2005} is at the basis of several subsequent works, such as the one from
Kaskebar et al. \cite{Kaskebar2009} that we study in more detail in section \ref{sec:two_approaches}.

In the work \cite{Cardei2005} by Cardei, Du et al., the sensor nodes are divided into disjoint sets, such that  at a
specific time only one sensor set is responsible for sensing
the targets, while the sensors of the other sets are kept in a low
power mode.
The sets are scheduled in a round robin manner and  operate for equal time intervals.
The authors prove that finding the maximum number of disjoint sets
is an NP-complete problem. For this reason they propose the use of
a heuristic approach to calculate the set covers on the basis of  a mixed integer programming model.
The main drawback of this approach is that it is  centralized, which is not desirable in a sensor network environment.
The constraint of having disjoint set covers operating for equal time intervals is relaxed in the work \cite{Cardei_infocom_2005} by
Cardei, Thai et al., and two heuristics are proposed, one using linear programming and the other using a greedy approach.

In \cite{Funke2007}, Funke et al. consider the  problem
of selecting a set of awake sensors of minimum cardinality
so that sensing coverage and network connectivity are
maintained. The authors analyze the performance of a
greedy solution for complete coverage showing that
it achieves an approximation factor no better than
$\Omega(\log n)$, where $n$ is the number of sensor nodes. For this reason, the authors also
present algorithms that provide approximate coverage while
the number of nodes selected is a constant factor far from the
optimal solution.

The same problem is addressed by Tian et al. in \cite{Tian2002} and by Bulut et al. in \cite{Bulut2008}.
These works considered
the coverage problem aiming at activating only a minimal number of
sensors and letting the others conserve their energy in a low
power mode. Each sensor periodically
evaluates its sensing area to determine whether it is also covered
by other sensors. Once a sensor has determined its redundancy, it
can deactivate itself. Since several sensors may determine that
they can go to sleep at the same time, a  back-off based policy is
proposed to prevent collisions and impose a unique order of
deactivation. These proposals are similar to the way our algorithm
eliminates the redundancies in the case of sensors endowed with
fixed sensing capabilities. Nevertheless, the way we give priority
to sensors having higher overlaps is completely different, as it
is based on a more refined evaluation of the coverage diagram of
the network deployment.

None of the aforementioned works addresses the problem of continuously covering an area of interest with some or all sensors being
able to modulate their sensing ranges as we do in this paper.
This operative setting, but with discrete coverage targets, is analyzed by Cardei, Wu et al. in \cite{Cardei2006}.
 The proposed solution is based on non-disjoint set cover scheduling. The
approach is centralized and the problem is proved to be
NP-complete.
For this reason the authors provide two heuristics, both centralized and
distributed.

A limitation of all the above mentioned methods, is that they cannot
be dynamically reconfigured to accommodate different density
requirements, being them time-varying or position dependent.
Nevertheless, our approach proves to be very versatile in this
operative scenario, and is also robust to
network heterogeneity and diverse energy availability and
harvesting capacity.

An approach that is able to take account of event dynamics is the
one proposed in \cite{He2009} by He et al., where the scheduling policy is based on a probabilistic
technique. Nevertheless, this work assumes a Boolean sensing model and
does not address
the case of non uniform device capabilities and energy
availability. Furthermore it does not address the case of sensors
endowed with adjustable sensing ranges.

Two recent works by Kaskebar et al. \cite{Kaskebar2009}, and by Zou et al. \cite{Zou2009},
propose the algorithms $\dlm$ and $\gup$, respectively.
These algorithms
are described in deeper details in Section \ref{sec:two_approaches} where we also make performance comparisons with our proposal.

\section{Related work} \label{sec:related_works}

The problem of exploiting network redundancies to prolong the network lifetime
has been largely investigated in the literature so far.
Depending on the application requirements, the approach to
the problem may vary significantly.
For example, some works only aim at guaranteeing network connectivity,
the SPAN \cite{Chen2002} and ASCENT \cite{Cerpa2004} protocols just to mention the most acknowledged,
without considering coverage issues.
Due to space limitations, in this section we only
consider the works dealing with the problem of completely covering an area of
interest and we refer the reader to the work \cite{TLPsurvey} from Rawaihy et al. for
a survey of sensor scheduling policies in several other applicative
scenarios.

The PEAS protocol proposed by Ye et al. in \cite{YePEAS2003} was designed to address both
coverage and connectivity at the same time.
According to this protocol only a subset of nodes stay awake while the
others are put to sleep.
A sleeping node occasionally wakes up to determine
the presence of coverage holes in its proximity and make activation decisions accordingly.
This approach does not ensure complete coverage, as coverage holes cannot be discovered until a nearby sleeping
sensor wakes up.
Another randomized algorithm is proposed by Xiao et al. in
\cite{Xiao2010}. Different sets of sensors work alternatively according to a probabilistic scheduling.
The authors study the performance of the proposed approach in terms of coverage extension and detection delay.
Differently from the works in \cite{YePEAS2003,Xiao2010}, our approach aims at
ensuring the coverage completeness as long as the available sensors have enough energy.

Xing et al. \cite{Xing2005}  propose the
protocol CPP to achieve $k$-coverage of an area of interest while maintaining the network connectivity.
They  address both coverage and connectivity, and in particular they define
an operative setting in which the  the former implies the latter, namely
when the transmission radius is at least twice the sensing range.
They also provide necessary and sufficient conditions for an area to be $k$-covered.
The authors point out that the network lifetime achieved by their algorithm does not linearly scale
with the number of sensing nodes, due to the higher energy consumption related to periodic beacon messages.
Our work addresses the same operative setting (with $k=1$) with a more aggressive scheduling policy
by resorting to the Laguerre metric space rather than to the Euclidean one, thus allowing a
better scalability.
The geometric analysis made in \cite{Xing2005} is at the basis of several subsequent works, such as the one from
Kasbekar et al. \cite{Kaskebar2009} that we study in more detail in section \ref{sec:two_approaches}.

In the work \cite{Cardei2005} by Cardei, Du et al., the sensor nodes are divided into disjoint sets, such that  at a
specific time only one sensor set is responsible for sensing
the targets, while the sensors of the other sets are kept in a low
power mode.
The sets are scheduled in a round robin manner and  operate for equal time intervals.
The authors prove that finding the maximum number of disjoint sets
is an NP-complete problem. For this reason they propose the use of
a heuristic approach to calculate the set covers on the basis of  a mixed integer programming model.
The main drawback of this approach is that it is  centralized, which is not desirable in a sensor network environment.
The constraint of having disjoint set covers operating for equal time intervals is relaxed in the work \cite{Cardei_infocom_2005} by
Cardei, Thai et al., and two heuristics are proposed, one using linear programming and the other using a greedy approach.

In \cite{Funke2007}, Funke et al. consider the  problem
of selecting a set of awake sensors of minimum cardinality
so that sensing coverage and network connectivity are
maintained. The authors analyze the performance of a
greedy solution for complete coverage showing that
it achieves an approximation factor no better than
$\Omega(\log n)$, where $n$ is the number of sensor nodes. For this reason, the authors also
present algorithms that provide approximate coverage while
the number of nodes selected is a constant factor far from the
optimal solution.

The same problem is addressed by Tian et al. in \cite{Tian2002} and by Bulut et al. in \cite{Bulut2008}.
These works considered
the coverage problem aiming at activating only a minimal number of
sensors and letting the others conserve their energy in a low
power mode. Each sensor periodically
evaluates its sensing area to determine whether it is also covered
by other sensors. Once a sensor has determined its redundancy, it
can deactivate itself. Since several sensors may determine that
they can go to sleep at the same time, a  back-off based policy is
proposed to prevent collisions and impose a unique order of
deactivation. These proposals are similar to the way our algorithm
eliminates the redundancies in the case of sensors endowed with
fixed sensing capabilities. Nevertheless, the way we give priority
to sensors having higher overlaps is completely different, as it
is based on a more refined evaluation of the coverage diagram of
the network deployment.

None of the aforementioned works addresses the problem of continuously covering an area of interest with some or all sensors being
able to modulate their sensing ranges as we do in this paper.
This operative setting, but with discrete coverage targets, is analyzed by Cardei, Wu et al. in \cite{Cardei2006}.
 The proposed solution is based on non-disjoint set cover scheduling. The
approach is centralized and the problem is proved to be
NP-complete.
For this reason the authors provide two heuristics, both centralized and
distributed.

A limitation of all the above mentioned methods, is that they cannot
be dynamically reconfigured to accommodate different density
requirements, being them time-varying or position dependent.
Nevertheless, our approach proves to be very versatile in this
operative scenario, and is also robust to
network heterogeneity and diverse energy availability and
harvesting capacity.

An approach that is able to take account of event dynamics is the
one proposed in \cite{He2009} by He et al., where the scheduling policy is based on a probabilistic
technique. Nevertheless, this work assumes a Boolean sensing model and
does not address
the case of non uniform device capabilities and energy
availability. Furthermore it does not address the case of sensors
endowed with adjustable sensing ranges.

Two recent works by Kasbekar et al. \cite{Kaskebar2009}, and by Zou et al. \cite{Zou2009},
propose the algorithms $\dlm$ and $\gup$, respectively.
These algorithms
are described in deeper details in Section \ref{sec:two_approaches} where we also make performance comparisons with our proposal.

{\color{black}
\marginpar{TIZ: Moved from the Algo section. Find a suitable place}
A policy based on the setting of a back-off period for
putting redundant devices to sleep was proposed in \cite{Tian2002,Bulut2008}.
Unlike these previous proposals, we are able to set the sleep
priority of the individual devices on the basis of the parameter $\alpha$ which can be defined
according to specific application goals.
Furthermore, the mentioned proposals do not deal with the case of heterogeneous networks
with the contemporary presence of sensors with fixed and adjustable sensing capabilities.
}

\section{Conclusions} \label{sec:conclusions}
 We 
 proposed a new algorithm for prolonging the lifetime of a heterogeneous wireless sensor network (WSN) through selective Sensor Activation and sensing Radius Adaptation ($\alg$).
 Our approach to joint sensor activation and radio adaptation is very general, and is the first to be applicable to scenarios with devices with adjustable and fixed sensing ranges (heterogeneous WSNs).
 In particular we focus on networks where some devices are able to adjust their sensing range so as to decrease the energy consumption.
 The proposed algorithm is based on a model of the coverage problem which uses  Voronoi-Laguerre diagrams. This model  allows to explicitly take account of device
 heterogeneity.
 We prove the convergence, termination and the Pareto-optimality of our approach. 
 The proposed algorithm achieves longer lifetime and higher coverage than previous solutions in all the considered scenarios.

\section{Appendix}

\marginpar{Rimosso il lemma che diceva che un boundary \'e boundary per tre. Questo lemma era ripetuto come parte del lemma 5.1 gi\'a presente
nella sezione sulle propriet\'a. Ho spezzato il lemma 5.1 cos� i lemmi 10.1 e 5.1 sono entrambi nella sezione sulle propriet� in una forma ridotta pi\'u semplice.}
\subsection*{Details about loose boundary farthest vertices}

To complete our geometrical analysis, we now detail the general methodology according to which
a sensor $s$ can determine whether a boundary farthest vertex $F=f(V(\mathscr{C}))$
of its polygon  is strict or loose. 

{\color{black}

 In order to illustrate the methodology, let us consider the example of Fig. \ref{fig:strict_and_loose}(b).
 The sensors $s$, $s_i$,  $s_k$ and $s_l$ generate a common boundary farthest vertex  $F=f(V(\mathscr{C}))$. 
 This vertex is a strict boundary farthest for all the generating sensors with the exception of $s$.
 Indeed, as also shown in Fig. \ref{fig:loose_radius_reduction}, the sensor $s$ can still reduce its sensing radius without leaving a coverage hole.
 
The edges of the polygon $V({\mathscr C})$ that intersect in $F$ are generated by intersecting the circles $\mathscr{C}$ and $\mathscr{C}_l$ and 
the circles $\mathscr{C}$ and $\mathscr{C}_k$, where $s_l$ and $s_k$ are the Voronoi-Laguerre neighbors that with $s$ generate the common farthest $F$.
 The circles $\mathscr{C}_l$ and $\mathscr{C}_k$ also intersect each other in the point $F'$ that we call  {\em opposite farthest with respect to $s$}.

Notice that if $F$ is a loose farthest vertex, then $F'$ must be internal to the angle formed by the VorLag axes generating the boundary farthest $F$ and  on the side
of  $V(\mathscr{C})$.
Indeed, if $f(V(\mathscr{C}))$ is loose, then
it exists a finite value $\delta$ such that every point at distance less than $\delta$ from $F$ and internal to $V(\mathscr{C})$ is covered by at least another sensor.
\marginpar{La footnote \'e davvero informale ma non mi sembra il caso di dettagliare ulteriormente, forse piuttosto taglierei
la questione dei null. In realt� nessun cerchio pu� contenere quello che per altri � un boundary farthest, non importa che sia null, questo cerchio genererebbe uno suo poligono.}
As a consequence of Theorem \ref{th:test}, the points of the $\delta$-surrounding of $F$ internal to $V(\mathscr{C})$ are also covered  by a Voronoi-Laguerre neighbor
\footnote{The theorem also mentions sensors with null polygons, but the points around a boundary farthest cannot be covered by a null polygon because any sensor added in a position 
where it overlaps a boundary farthest generated by other sensors, would generate its own polygon, containing the point that was a boundary farthest before its addition}.

Therefore let us consider the only Voronoi-Laguerre neighbors of $s$. The neighbors that will be able to cover an area arbitrarily close to $F$ are therefore $s_l$ and $s_k$. 
The only way for these sensor  to avoid leaving any point 1-covered in the surrounding of $F$ is to intersect each other in a point (the opposite farthest $F'$) that is internal
to the axes generating $F$ and on the side of  $V(\mathscr{C})$.

Observe that also the opposite implication holds as, if $F'$ is included in the angle formed by the VorLag axes generating ${f(V(\mathscr C}))$, 
then ${\mathscr C}_l \cup {\mathscr C}_k$ cover both $F$ and any other point at distance less than $\delta$ from it, for some finite value of $\delta$, and hence $F$ is a loose farthest.
}
We can summarize the previous reasonings in the following:

\medskip

{\bf Characterization of loose farthest vertices.}
{\em 
Consider a sensor $s$, and let  the farthest vertex $f(V({\mathscr C}))$ of its Voronoi-Laguerre cell from $s$ be determined as the intersection point of 
the  edges of $V({\mathscr C})$ lying on the axes formed by $s$ and $s_l$ and by $s$ and $s_k$. 
$f(V({\mathscr C}))$ is loose for $s$ if and only if the opposite farthest $F'$ with respect to $s$ lies inside the angle formed by the VorLag axes generating $f(V({\mathscr C}))$ which contains $s$.
}

\medskip

Figure \ref{fig:strict_and_loose} evidences some examples of positions of the circles $\mathscr{C}$, $\mathscr{C}_l$ 
and $\mathscr{C}_k$ that generate  a strict farthest (in (a)), and a loose one (in (b)).  
The figure highlights the position of the farthest vertex $F$
and of the opposite farthest $F'$.
%
Notice that the opposite farthest $F'$ can be external to $V(\mathscr{C})$ (on the opposite side of the polygon with respect to $F$)
and still determine a loose farthest vertex situation.

%
%

\bibliographystyle{acmtrans}
\bibliography{bibliografia}
\end{document}